\documentclass[letterpaper,onecolumn,accepted=2026-01-14]{quantumarticle}
\pdfoutput=1
\usepackage{xcolor}
\definecolor{prxlink}{RGB}{46,48,146}
\usepackage[colorlinks=true,allcolors=prxlink]{hyperref}
\usepackage[ruled,vlined]{algorithm2e}
\usepackage{amsfonts,amssymb,amsthm,bm,graphicx,mathtools,physics,qcircuit,tikz,upgreek,amsmath}
\usepackage[T1]{fontenc}
\usepackage[numbers,sort&compress]{natbib}
\theoremstyle{plain}
\newtheorem{theorem}{Theorem}[section]
\newtheorem{proposition}[theorem]{Proposition}
\newtheorem{lemma}[theorem]{Lemma}

\theoremstyle{definition}

\newtheorem{assumption}[theorem]{Assumption}
\theoremstyle{remark}

\newcommand{\ci}{\ensuremath{\mathrm{i}}}
\newcommand{\cpi}{\ensuremath{\uppi}}

\newcommand{\change}[1]{{#1}}

\begin{document}
\title{Arbitrary Polynomial Separations in Trainable Quantum Machine Learning}
\author{Eric R.\ Anschuetz}
\email{eans@caltech.edu}
\affiliation{Institute for Quantum Information and Matter, Caltech, Pasadena, CA, USA}
\affiliation{Walter Burke Institute for Theoretical Physics, Caltech, Pasadena, CA, USA}
\affiliation{MIT Center for Theoretical Physics, Cambridge, MA, USA}
\author{Xun Gao}
\affiliation{JILA and Department of Physics, CU Boulder, Boulder, CO, USA}

\begin{abstract}
    Recent theoretical results in quantum machine learning have demonstrated a general trade-off between the expressive power of quantum neural networks (QNNs) and their trainability; as a corollary of these results, practical exponential separations in expressive power over classical machine learning models are believed to be infeasible as such QNNs take a time to train that is exponential in the model size. We here circumvent these negative results by constructing a hierarchy of efficiently trainable QNNs that exhibit unconditionally provable, polynomial memory separations of arbitrary constant degree over classical neural networks---including state-of-the-art models, such as Transformers---in performing a classical sequence modeling task. This construction is also computationally efficient, as each unit cell of the introduced class of QNNs only has constant gate complexity. We show that contextuality---informally, a quantitative notion of semantic ambiguity---is the source of the expressivity separation, suggesting that other learning tasks with this property may be a natural setting for the use of quantum learning algorithms.
\end{abstract}

\maketitle

\section{Introduction}

The field of quantum computing has made remarkable strides in recent years, with many quantum experiments demonstrating the potential for quantum devices to efficiently solve classically intractable problems~\cite{arute2019quantum,doi:10.1126/science.abe8770,PhysRevLett.127.180501,PhysRevLett.127.180502,ZHU2022240,hangleiter2022}. These experiments have largely involved performing sampling tasks due to the fact that quantum devices naturally sample from distributions that are generally regarded as difficult for classical computers to sample from~\cite{10.1145/1993636.1993682,bouland2019complexity,v016a011}. This observation has motivated recent study~\cite{schuld2015introduction,biamonte2017quantum,PhysRevX.8.021050,PhysRevLett.121.040502,PhysRevLett.122.040504,PhysRevResearch.1.033063} into the feasibility of using quantum devices to perform a class of machine learning tasks known as \emph{generative modeling}: given samples from some target distribution $p\left(\bm{y}\right)$, can one efficiently learn and then sample from some model distribution $q\left(\bm{y}\right)$ closely approximating $p\left(\bm{y}\right)$ using a quantum algorithm? Furthermore, to determine whether the overhead of constructing a quantum device to implement such a quantum algorithm is necessary, one is also interested in the question of \emph{quantum advantage}: are there any quantum-efficient learning tasks which \emph{require} quantum resources to efficiently perform?

Recent theoretical analysis of quantum machine learning (QML) algorithms has determined that, indeed, such learning tasks do exist~\cite{doi:10.1126/sciadv.aat9004,coyle2020born,PhysRevResearch.2.033125,sweke2021quantum,liu2021rigorous,gilfuster2024relation}. These results rely on reducing machine learning tasks to classically hard computational problems---such as the discrete logarithm---which are known to be easy to solve using quantum devices~\cite{doi:10.1137/S0097539795293172}. The associated learning algorithms are then problem-specific, i.e., require prior knowledge of the classically-hard problem used in the reduction. For this reason, while these constructions are theoretically useful for assessing the promise of quantum learning algorithms, it is unclear how to train such QML algorithms in real-world learning settings.

One might hope that these problem-specific learning algorithms could be replaced by a problem-agnostic algorithm such as gradient descent. This is possible in classical machine learning, where it is known that gradient descent efficiently optimizes typical loss functions even if in the worst case the loss can be nonconvex~\cite{pmlr-v38-choromanska15,chaudhari2017energy}. Unfortunately, it is known that this is not the case in the quantum setting; generally, the training of QML models using problem-agnostic algorithms is difficult due to \emph{quantum trainability barriers} (QTBs) such as the presence of ``barren plateaus'' (exponential loss function concentration)~\cite{mcclean2018barren,cerezo2021cost,napp2022quantifying,cerezo2020impact,arrasmith2021effect} or poor local minima in the loss landscape~\cite{anschuetz2021critical,anschuetz2022barren,anschuetz2024unified}. More recently, it has been shown that QTBs can only be avoided when the dynamics of the network are constrained to belong to a low-dimensional subspace of the full Hilbert space~\cite{fontana2023adjoint,ragone2023unified,anschuetz2024unified}. In many such scenarios efficient classical simulation algorithms exist for the quantum model, putting into question the practical utility of QML models~\cite{anschuetz2022efficient,Goh_2025,cerezo2023does}. Though there are still rigorous quantum-classical separations remaining in QML even with these restrictions~\cite{gao2021enhancing,zhao2024entanglement,anschuetzgao2022}, the separations are modest, making a practical experimental implementation showcasing an advantage unlikely due to the large overhead of quantum error correction~\cite{PRXQuantum.2.010103}.

We here for the first time balance considerations of (problem-agnostic) trainability, efficiency, and large quantum advantage by constructing a hierarchy of quantum neural networks (QNNs) called \emph{$k$-hypergraph recurrent neural networks} ($k$-HRNNs) satisfying three simultaneous properties:
\begin{enumerate}
    \item \textbf{Trainability.} $k$-HRNNs avoid QTBs.
    \item \textbf{Efficiency.} Each unit cell of a $k$-HRNN has quantum gate complexity independent of the model size.
    \item \textbf{Large quantum advantage.} $k$-HRNNs exhibit arbitrary polynomial memory advantages over \emph{any} classical neural network in performing a certain generative modeling task over sequences.
\end{enumerate}
Namely, we show that for every $k$-HRNN of model size $n$ there exists a certain sequence learning task that can be performed by the model to zero error. We then show that no classical network with fewer than $\binom{n}{k}-1$ neurons can perform this task to \emph{any} finite cross entropy. Importantly, the classical networks for which we prove a memory lower bound include state-of-the-art models such as Transformers~\cite{10.5555/3295222.3295349}. $k$-HRNNs achieve this balance of efficient trainability and high expressiveness by having a considerable amount of structure: the latent state of the network is constrained to be a \emph{quantum hypergraph state}~\cite{takeuchi2019resource,Webster2022xpstabiliser}. We show that this structure is sufficient to prevent QTBs from occurring in the associated loss landscape of the network, while still giving the network enough freedom to maintain a large quantum-classical expressivity separation.

\change{Our result is asymptotically optimal, as there exists a classical algorithm capable of simulating $k$-HRNNs using memory $\operatorname{\Theta}\left(\binom{n}{k}\right)$~\cite{PhysRevLett.97.190501}. We emphasize that while an arbitrary-polynomial time separation persists during inference, this classical simulation algorithm only exhibits a \emph{cubic} time overhead during training due to the relative inefficiency of training on a quantum computer. Similar classical simulation methods were previously thought to inhibit any large quantum-classical separation in QML~\cite{anschuetz2022efficient,Goh_2025,cerezo2023does}. Our work manifests one important detail of these results: while this conventional wisdom is true during \emph{training} (at least when the algebraic structure is known a priori~\cite{mints2025fragmentationefficientlylearnablequantum}), large separations can be maintained during \emph{inference}. As inference is by far the most expensive cost of deploying networks today~\cite{MLSYS2022_462211f6,mavromatis2024computinglimitsempiricalstudy}, our results are an important example of how large quantum-classical separations in machine learning can exist even though these efficient classical simulation algorithms are known.}

Our separation is also interpretable: we show that the classical memory lower bound follows from the large amount of \emph{contextuality}---in a technical sense~\cite{kochen1975problem,Abramsky_2011}---in the learning task we consider. Informally, the semantic meaning of any given token strongly depends on other tokens in the sequence. We prove that the presence of this contextuality inhibits efficient classical representations of the data. Though our result is proven on a specific data set, this suggests sequential data with a similar correlation structure may exhibit similar classical memory lower bounds. Indeed, contextuality has previously been used as a rigorous model for linguistic ambiguity~\cite{abramsky2014}, relational databases~\cite{doi:10.1142/9789814730617_0002}, and other settings where data has a ``global inconsistency'' structure~\cite{abramsky_et_al,e22090981}. One interpretation of our results is then that this correlation structure inhibits efficient classical representations of data but not necessarily efficient quantum representations of data, motivating problems exhibiting semantic ambiguity as natural settings in which to use QML algorithms.

Our results also have implications beyond machine learning; our separation can also be viewed as a quantum-classical communication separation in sampling from a certain conditional distribution $p\left(y_1,\ldots,y_\ell\mid x_1,\ldots,x_\ell\right)$, where $\ell$ parties each receive one $x_i$ and output the associated $y_i$. Viewed through this lens, our results differ from previously known quantum-classical communication complexity separations~\cite{10.1145/301250.301343,PhysRevLett.87.167902,10.1145/1148109.1148119,10.1145/1250790.1250866,PhysRevLett.130.080802,kallaugher2023exponential,manna2024unbounded} not only in the strength of the error model we consider, but also as the quantum algorithm capable of performing this task is efficient---each party only requires constant gate complexity. While there are many parties---the sequence lengths are on the order of $\binom{n}{k}$ in our learning task---this is natural in a streaming setting, where one hopes to minimize the per-token resource requirements in processing a long stream of tokens.

We proceed in the following way. In Sec.~\ref{sec:background}, we review basic concepts in quantum computation and neural sequence learning. Then, in Sec.~\ref{sec:hrnn}, we introduce a new class of QNN models and demonstrate that it avoids quantum barriers to trainability. We then introduce a sequence learning task in Sec.~\ref{sec:main_result_main_text}, and give large classical communication-complexity lower-bounds in performing this task; simultaneously, we show that the QNN model we introduce is able to efficiently perform this task. Finally, we end in Sec.~\ref{sec:conc} with discussion on the implications of our results.

\section{Preliminaries}\label{sec:background}

\subsection{Quantum Generative Models}\label{sec:quant_gen_mods}

We first review quantum generative models. Given a quantum state $\ket{\psi}$ on either an $n$-qubit or $n$-qumode system, there is a natural induced probability distribution given by the \emph{measurement distribution}:
\begin{equation}
    q\left(y_1,\ldots,y_\ell\right)=\bra{\psi}\varPi_{y_1}\ldots\varPi_{y_\ell}\ket{\psi},
\end{equation}
where $\varPi_{y_i}$ are Hermitian projectors (assumed here to be diagonal). More generally, one can consider a measurement distribution parameterized by $\bm{\theta}\in\mathbb{R}^p$:
\begin{equation}
    q_{\bm{\theta}}\left(y_1,\ldots,y_\ell\right)=\bra{\psi}U_{\bm{\theta}}^\dagger\varPi_{y_1}\ldots\varPi_{y_\ell} U_{\bm{\theta}}\ket{\psi},
\end{equation}
where $U_{\bm{\theta}}$ is a parameterized unitary operator. When $\bm{\theta}$ is optimized such that $q_{\bm{\theta}}$ approximates some distribution of interest $p$, we say that $q_{\bm{\theta}}\left(y_1,\ldots,y_\ell\right)$ is a \emph{quantum generative model}. On a quantum computer this optimization is typically performed by optimizing the maximum mean discrepency (MMD) loss function, which can be written as minimizing the expectation value of some diagonal Hermitian operator $O_{\mathcal{D}}$ depending only on the training set $\mathcal{D}$~\cite{rudolph2023trainability}:
\begin{equation}
    \ell\left(\bm{\theta}\right)=\bra{\psi}U_{\bm{\theta}}^\dagger O_{\mathcal{D}}U_{\bm{\theta}}\ket{\psi}.
\end{equation}

If the class $U_{\bm{\theta}}$ is chosen generically, this optimization is provably hard due to the presence of \emph{quantum trainability barriers} (QTB)~\cite{fontana2023adjoint,ragone2023unified,anschuetz2024unified}. Such barriers include the exponential concentration of typical loss functions over uniformly sampled $\bm{\theta}$ (commonly known as the presence of ``barren plateaus''), or a loss landscape dominated by poor local minima. QTBs are provably avoided when~\cite{rudolph2023trainability,fontana2023adjoint,ragone2023unified,anschuetz2024unified}:
\begin{enumerate}
    \item $U_{\bm{\theta}}$ is constrained to belong to a $\operatorname{poly}\left(n\right)$-dimensional Lie subgroup $G$ of the unitary group, and
    \item both $\ci\ket{\psi}\bra{\psi}$ and $\ci O_{\mathcal{D}}$ belong to the algebra generating $G$.
\end{enumerate}
Indeed, one can show that training such models via gradient descent takes time at least $\operatorname{\Omega}\left(\dim\left(G\right)\right)$~\cite{fontana2023adjoint,ragone2023unified}.

Finally, we introduce some notation. We will often reference the $X$ and $Z$ Pauli operators:
\begin{equation}
    X=\begin{pmatrix}
    0 & 1\\
    1 &0
    \end{pmatrix},\quad Z=\begin{pmatrix}
    1 & 0\\
    0 & -1
    \end{pmatrix},
\end{equation}
with subscript $P_i$ shorthand for the $n$-fold tensor product $I_2\otimes\ldots\otimes P\otimes\ldots\otimes I_2$ for Pauli operator $P$ at the $i$th position. We will also consider the canonical position and momentum operators $\hat{q}_i,\hat{p}_i$ formally acting on an infinite-dimensional Hilbert space and obeying the so-called canonical commutation relation (in units where $\hbar=1/2$):
\begin{equation}
    \left[\hat{q}_j,\hat{p}_k\right]=\frac{\ci}{2}\delta_{jk}.
\end{equation}

\subsection{Neural Sequence Learning}\label{sec:classical_sequence_learning}

\begin{figure}
    \begin{center}
        \includegraphics[width=0.6\linewidth]{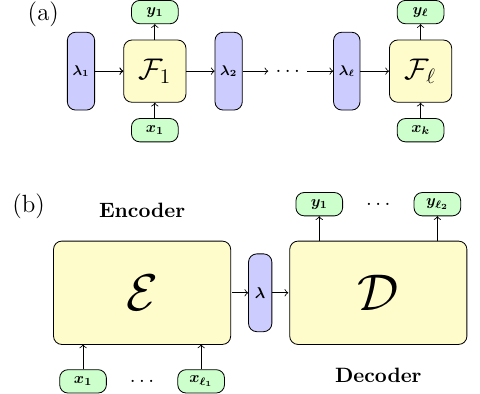}
        \caption{(a) An autoregressive neural sequence model. The model autoregressively takes input tokens $\bm{x_i}$, and outputs decoded tokens $\bm{y_i}$, with map $\mathcal{F}_i$. The model also has an unobserved internal memory with state $\bm{\lambda_i}\in L$ after decoding the token $\bm{x_{i-1}}$ that $\mathcal{F}_i$ can depend on. (b) A general encoder-decoder model. $\mathcal{E}$ encodes the input $\bm{x}$ into some latent representation $\bm{\lambda}\in L$. A decoder $\mathcal{D}$ then outputs the decoded sequence $\bm{y}$.\label{fig:classical_models}}
    \end{center}
\end{figure}

\emph{Sequence-to-sequence} or \emph{sequence} learning~\cite{10.5555/2969033.2969173} is the approximation of some given conditional distribution $p\left(\bm{y}\mid\bm{x}\right)$ with a model distribution $q\left(\bm{y}\mid\bm{x}\right)$, where the vectors $\bm{x},\bm{y}$ are sequences of vectors $\left(\bm{x_1},\ldots,\bm{x_{\ell_1}}\right),\left(\bm{y_1},\ldots,\bm{y_{\ell_2}}\right)$; we call each vector ($\bm{x_i}$ or $\bm{y_i}$) a \emph{token}. This framework encompasses sentence translation tasks~\cite{10.5555/2969033.2969173}, speech recognition~\cite{Prabhavalkar2017}, image captioning~\cite{Vinyals_2015_CVPR}, and many more practical problems.

Sequence modeling today is typically performed using \emph{neural sequence models}. Generally, these models are functions parameterized by some $\bm{\theta}\in\mathbb{R}^p$ that take as input the sequence $\bm{x}$ and output a sample from the parameterized model distribution $q_{\bm{\theta}}\left(\bm{y}\mid\bm{x}\right)$. The parameters of these functions are trained to minimize an appropriate loss function, such as the (forward) empirical cross entropy:
\begin{equation}
        \hat{H}\left(p,q\mid\bm{\theta}\right)=-\frac{1}{\left\lvert\mathcal{T}\right\rvert}\sum\limits_{\left(\bm{x},\bm{y}\right)\in\mathcal{T}}p\left(\bm{y}\mid\bm{x}\right)\log\left(q_{\bm{\theta}}\left(\bm{y}\mid\bm{x}\right)\right),
        \label{eq:emp_cross_ent_main_text}
\end{equation}
where $\mathcal{T}=\left\{\left(\bm{x_i},\bm{y_i}\right)\right\}$ are samples from $p\left(\bm{x},\bm{y}\right)$. The backward empirical cross entropy is similarly defined, with $p\leftrightarrow q_{\bm{\theta}}$.

One of the unique challenges when performing sequence modeling is the potentially large data dimensions due to data being comprised of arbitrarily long sequences. Thus, to maintain a resource scaling independent of the input sequence length, neural sequence models usually fall into one of two classes: \emph{autoregressive sequence models}~\cite{Hopfield2554,10.1162/neco.1997.9.8.1735,cho-etal-2014-learning} or \emph{encoder-decoder models}~\cite{10.5555/2969033.2969173,10.5555/3295222.3295349}. These classes of models are illustrated in Figs.~\ref{fig:classical_models}(a) and~\ref{fig:classical_models}(b), respectively; in the former, the trained network components are the unit cells $\mathcal{F}_i$, while in the latter the network is composed of an encoder $\mathcal{E}$ and a decoder $\mathcal{D}$. Both models make use of a \emph{latent vector} $\bm{\lambda}$ belonging to a \emph{latent space} $L$ in order to model long-range correlations between the tokens in a given data set. While the limit of large $\dim\left(\bm{\lambda}\right)$ yields more expressive models, it also (typically) yields models less efficient to sample from. For this reason, we will use the dimension of $\bm{\lambda}$ as the metric by which we evaluate the efficiency of a given model.

We have not yet stated any requirements on the model functions $\mathcal{F}_i$, $\mathcal{E}$, or $\mathcal{D}$. If these functions were allowed to be fully general, they could in principle implement arbitrarily complicated, discontinuous functions. In practice, however, these $\mathcal{F}_i$, $\mathcal{E}$, and $\mathcal{D}$ are composed of elementary functions with some sort of smoothness structure such that efficient training of the network is possible. For this reason, we assume classical networks satisfy two properties on some finite-volume subspace $\mathcal{X}$ of the space of all inputs:
\begin{assumption}[$C^2$ on $\mathcal{X}$]
    In the autoregressive setting, each $\mathcal{F}_i$ is continuously second-differentiable on $\mathcal{X}$. In the encoder-decoder setting, $\mathcal{E}$ is continuously second-differentiable on $\mathcal{X}$.\label{ass:c2}
\end{assumption}
\begin{assumption}[Strongly Morse on $\mathcal{X}$]
    In the autoregressive setting, the Jacobian of each $\mathcal{F}_i$ has determinant lower-bounded by a constant on $\mathcal{X}$. In the encoder-decoder setting, the Jacobian of $\mathcal{E}$ has determinant lower-bounded by a constant on $\mathcal{X}$.\label{ass:strong_sub}
\end{assumption}

Assumption~\ref{ass:c2} is straightforward and is automatically satisfied by models composed of linear layers along with rectified linear units, softmax functions, sigmoid functions, hyperbolic tangent functions, and other standard nonlinearities. Assumption~\ref{ass:strong_sub} is slightly more subtle. Stated another way, we require that there is a finite-volume subspace of the space of inputs that are not near-critical points of the network. This is effectively a requirement that the model is \emph{strongly Morse}~\cite{pmlr-v89-mokhtari19a} in a multivariate sense. Neural networks being strongly Morse is a common assumption~\cite{64871590-990e-34a9-8236-2954b5e72da7,pmlr-v89-mokhtari19a,yang2021,dixit2023accelerated}, and certain classes of neural networks are known to be almost surely Morse over settings of their parameters~\cite{kurochkin2021neural}. In Appendix~\ref{sec:proof_of_express_sep} we give other natural settings where these Assumptions are automatically satisfied.

\section{Quantum \texorpdfstring{$k$}{k}-Hypergraph Recurrent Neural Networks}\label{sec:hrnn}

We now introduce the hierarchy of quantum models that exhibit arbitrary polynomial expressivity separations over classical neural networks satisfying Assumptions~\ref{ass:c2} and~\ref{ass:strong_sub}. We call the $k$th level of this hierarchy of models \emph{$k$-hypergraph recurrent neural networks} ($k$-HRNNs) for reasons that will soon become clear. We give a construction of these models in both a continuous variable (CV) setting and a qubit setting. In Appendix~\ref{sec:hrnns_as_q_rnns}, we also show that our $k$-HRNN construction has a natural interpretation as the quantization of a classical autoregressive architecture.

\subsection{Qumode Construction}\label{sec:qumode_hrnn_inf}

In a similar fashion to classical recurrent neural networks~\cite{Hopfield2554}, $k$-HRNNs are autoregressive models---they are given input tokens $\bm{x_i}$ sequentially, and output tokens $\bm{y_i}$ sequentially. The model also has an $n$-qumode \emph{latent state} $\ket{\lambda}$, generalizing the $n$-dimensional latent vector $\bm{\lambda}\in\mathbb{R}^n$ in the classical case. We use $\ket{\lambda_i}$ to denote this state just prior to receiving the input token $\bm{x_i}$. We assume that the initial latent state is the position-squeezed state $\ket{\lambda_1}=\ket{\hat{q}=0}^{\otimes n}$.

The trained components of the network are classical, given by four classical networks that are functions of the input $\bm{x_i}$ at a given time step:
\begin{align}
    \bm{\alpha}\left(\bm{x_i}\right)&=\left(v_1,\ldots,v_k\right)\in\binom{\left[n\right]}{k};\\
    \bm{\beta}\left(\bm{x_i}\right)&=\left(v_1,\ldots,v_k\right)\in\binom{\left[n\right]}{k};\\
    \gamma\left(\bm{x_i}\right)&\in\mathbb{R};\\
    \bm{\kappa}\left(\bm{x_i}\right)&=\left(\bm{\phi},\bm{\theta}\right)\in\mathbb{R}^k\times\mathbb{R}^{2^k}.
\end{align}
Here, $\binom{\left[n\right]}{k}$ is the set of all sets of $k$ distinct elements chosen from $\left\{1,\ldots,n\right\}$. In what follows, we use the quantum gate:
\begin{equation}
    \operatorname{C}^{k-1}Z_{\overline{v}}\left(\gamma\right)=\exp\left(-\ci\gamma\prod_{i=1}^k\hat{q}_{v_i}\right),
\end{equation}
i.e., the continuous-variable analogue of a multi-control $Z$-rotation acting on qumodes indexed by $\overline{v}\in\binom{\left[n\right]}{k}$. We also define:
\begin{equation}\label{eq:g_def}
    G_{\overline{v}}\left(\bm{\phi},\bm{\theta}\right):=\exp\left(-\ci\sum_{i=1}^k\phi_i\hat{p}_{v_i}\right)\exp\left(-\ci\sum_{\overline{w}\in 2^{\overline{v}}}\theta_{\overline{w}}\prod_{w_i\in\overline{w}}\hat{q}_{w_i}\right),
\end{equation}
where here $2^{\overline{v}}$ denotes the power set of $\overline{v}$. For appropriate choices of $\bm{\phi}$ and $\bm{\theta}$, $G_{\overline{v}}\left(\bm{\phi},\bm{\theta}\right)$ can be interpreted as a stabilizer of a \emph{(weighted) hypergraph state}~\cite{takeuchi2019resource}, a class of states in one-to-one correspondence with weighted hypergraphs. Given a hypergraph with hyperedges $\overline{v}\in E$ and corresponding weights $w_{\overline{v}}$, the associated hypergraph state is:
\begin{equation}
    \ket{\psi}=\prod_{\overline{v}\in E}C^{\left\lvert\overline{v}\right\rvert-1}Z_{\overline{v}}\left(w_{\overline{v}}\right)\ket{\hat{p}=0}^{\otimes n}.
\end{equation}
In our setting, the maximal cardinality of $\overline{v}\in E$ is $k$. These states are reviewed in further detail in Appendix~\ref{sec:seq_learning_background}.

\begin{figure}
    \begin{center}
        \includegraphics[width=\linewidth]{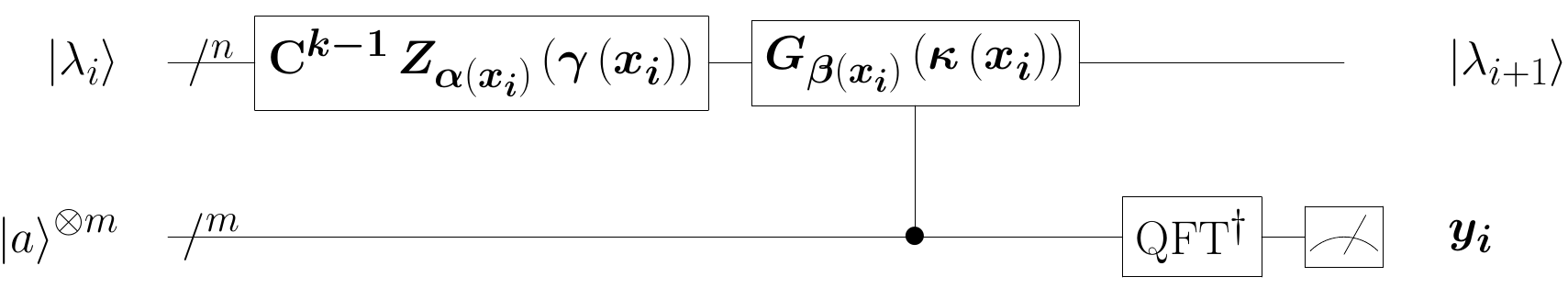}
        \caption{A $k$-hypergraph recurrent neural network ($k$-HRNN) unit cell as described in Sec.~\ref{sec:hrnn}, with Greek letters denoting trained classical networks and $G_\cdot$ as in Eq.~\eqref{eq:g_def}. $\bm{\alpha}$ and $\bm{\beta}$ output $k$-tuples of qumodes on which the associated operators have support. $\gamma$ and $\bm{\kappa}$ output the associated rotation angles. The ancillary $\ket{a}^{\otimes m}$ and $\text{QFT}$ perform phase estimation of the operators $G_\cdot$ as described in Sec.~\ref{sec:hrnn}. Measurement is either in the position basis (in the qumode setting) or in the computational basis (in the qubit setting).\label{fig:hrnn}}
    \end{center}
\end{figure}
We now describe the construction of a unit cell of a $k$-HRNN, where we assume for now that the dimension of each output token $y_i$ is $1$. First, the previously-described qumode analogue of a multi-control $Z$ rotation is applied to the latent state of the model:
\begin{equation}
    \ket{\lambda_i'}=\operatorname{C}^{k-1}Z_{\bm{\alpha}\left(\bm{x_i}\right)}\left(\gamma\left(\bm{x_i}\right)\right)\ket{\lambda_i}.
\end{equation}
Next, phase estimation (modulo $2\cpi$) of the operator $G_{\bm{\beta}\left(\bm{x_i}\right)}\left(\bm{\kappa}\left(\bm{x_i}\right)\right)$ is performed~\cite{PhysRevA.93.052304}. The measurement result $y_i$ is the output of the cell, and the post-measurement state $\ket{\lambda_{i+1}}$ is the new latent state of the model. The quantum circuit implementing this sequence of transformations is illustrated in Fig.~\ref{fig:hrnn}, where $\ket{a}=\ket{\text{GKP}}$ denotes the position-squeezed \emph{GKP state}~\cite{PhysRevA.64.012310}:
\begin{equation}
    \ket{\text{GKP}}\propto\sum\limits_{q=-\infty}^\infty\ket{\hat{q}=\cpi q}
\end{equation}
and $\operatorname{QFT}$ denotes the CV quantum Fourier transform $\mathcal{F}^{\otimes m}$. Higher-dimensional outputs $\bm{y_i}\in\mathbb{R}^m$ can be treated by performing $m$ measurements of operators of the form of $G_{\overline{v}}\left(\bm{\phi},\bm{\theta}\right)$ instead of just one. More concretely, one considers ordered sets of classical networks $\left(\bm{\beta}_h\left(\bm{x_i}\right)\right)_{h=1}^m$ and $\left(\bm{\kappa}_h\left(\bm{x_i}\right)\right)_{h=1}^m$, and sequentially performs phase estimation of the $G_{\bm{\beta}_h\left(\bm{x_i}\right)}\left(\bm{\kappa}_h\left(\bm{x_i}\right)\right)$ to read out the elements of $\bm{y_i}$.

By construction, qumode $k$-HRNNs have evolution generated via the commutator by the operators $\left\{\hat{p}_i,\prod_{v_j\in\overline{v}}\hat{q}_{v_j}\right\}$, where $\overline{v}$ is of degree at most $k$. This set is closed under the commutator as:
\begin{equation}
    \left[\hat{p}_i,\prod_{v_j\in\overline{v}}\hat{q}_{v_j}\right]=\frac{\ci}{2}\delta_{i\in\overline{v}}\prod_{v_j\neq i\in\overline{v}}\hat{q}_{v_j}\propto\prod_{w_j\in\overline{w}:=\overline{v}\setminus\left\{i\right\}}\hat{q}_{w_j}.\label{eq:comm_rel}
\end{equation}
Thus, for constant $k$, the dimension of the Lie subgroup $G$ to which the dynamics are constrained is $\operatorname{O}\left(n^k\right)$. Furthermore, the initial state of the network is a stabilizer state that is stabilized by group elements generated by the $\hat{p}_i$. Finally, measurement is in the position basis, and the MMD loss discussed in Sec.~\ref{sec:quant_gen_mods} is a function of expectation values of low-degree polynomials of the $\hat{q}_i$~\cite{rudolph2023trainability}. These facts taken together imply the absence of quantum trainability barriers (QTBs), as discussed in Sec.~\ref{sec:quant_gen_mods}.

Indeed, QTBs can be avoided more directly with only a small overhead in time complexity during training. While training a $k$-HRNN on a quantum devices would take time at least $\operatorname{\Omega}\left(n^k\right)$ (see Sec.~\ref{sec:quant_gen_mods}), there exists a known classical algorithm for simulating Lie subgroup-constrained dynamics which, when applied to estimating the loss of a $k$-HRNN, achieves a time complexity of $\operatorname{O}\left(n^{2.8k}\right)$ per unit cell~\cite{PhysRevLett.97.190501}. For this reason, QTBs can be circumvented by avoiding training on a quantum device entirely, incurring only a cubic overhead in time complexity! We thus emphasize that the separation we discuss in Sec.~\ref{sec:main_result_main_text} should be interpreted as an arbitrary-degree polynomial quantum-classical resource separation only during inference, i.e., after the model has been trained and deployed. We further discuss this subtlety in Appendix~\ref{sec:deets_on_complexity}.

\subsection{Qubit Construction}

The qubit construction is conceptually identical to the qumode construction, so we only highlight its differences. In the qubit setting, $\ket{\lambda}^{\otimes n}$ is an $n$-qubit quantum state initialized in the all-zero state $\ket{\lambda_1}=\ket{0}^{\otimes n}$, and the gates applied by the network are of the form:
\begin{align}
    \operatorname{C}^{k-1}Z_{\overline{v}}\left(\gamma\right)&=\exp\left(-\ci\gamma\ket{\bm{0}}\bra{\bm{0}}_{\overline{v}}\right),\\
    G_{\overline{v}}\left(\bm{\phi},\bm{\theta}\right)&:=\exp\left(-\ci\sum_{i=1}^k\phi_i X_{v_i}\right)\exp\left(-\ci\sum_{\overline{w}\in 2^{\overline{v}}}^k\theta_{\overline{w}}\ket{\bm{1}}\bra{\bm{1}}_{\overline{w}}\right),
\end{align}
where we have used the shorthand:
\begin{equation}
    \ket{\bm{1}}\bra{\bm{1}}_{\overline{v}}:=\prod_{v_i\in\overline{v}}\left(\frac{I_{2^n}-Z_{v_i}}{2}\right).
\end{equation}
Phase estimation of $G_{\overline{v}}\left(\bm{\phi},\bm{\theta}\right)$ can be performed using the same circuit as in the CV setting, diagrammed in Fig.~\ref{fig:hrnn}. In practice, phase estimation can also be performed to finite precision using the standard qubit-based algorithm; that is, where $\ket{a}=\ket{+}$ and $\operatorname{QFT}$ denotes the qubit-based quantum Fourier transform~\cite{Nielsen_Chuang_2010_fourier}. All other aspects of the network are identical to the CV setting.

The main distinction between the qubit setting and the qumode setting is that, unlike the qumode setting, the subspace explored by the qubit $k$-HRNN is \emph{not} low-dimensional. Indeed, in the qubit setting sampling from a $k$-HRNN is at least as difficult as sampling from the well-known instantaneous quantum polynomial-time (IQP) circuit class, widely believed to be exponentially difficult to classically simulate to small total variation distance even when $k=2$~\cite{bremner2017achievingquantum}. For this reason we cannot use the same arguments as in the qumode setting for the avoidance of QTBs in the training of the qubit-based network. In Appendix~\ref{sec:qubit_based_sep} we give a slightly more complicated construction of the qubit-based model which we argue is efficient to train, but no longer has the advantage of being independent of the token index $i$. \change{Unfortunately, the model no longer being translationally invariant complicates practical implementations in a typical streaming setting when using this training method. One potential heuristic to circumvent this shortcoming might be to ``round'' the CV model---which can be efficiently trained classically in a translationally invariant setting, as discussed in Sec.~\ref{sec:qumode_hrnn_inf}---to a qubit model at constant precision. We hope to explore this approach further in future work.}

\section{Arbitrary Polynomial Expressivity Separations Over Classical Neural Networks}\label{sec:main_result_main_text}

\subsection{\texorpdfstring{$\left(\ell,n,k\right)$}{(l,n,k)}-Hypergraph Stabilizer Measurement Translation}\label{sec:k_chmt}

Having introduced $k$-HRNNs, we now describe the classical sequence modeling task with which we will prove our separation. For clarity we here omit some technical details which are left for Appendix~\ref{sec:task_desc}.

We call the task we construct \emph{$\left(\ell,n,k\right)$-hypergraph stabilizer measurement translation} ($\left(\ell,n,k\right)$-HSMT), with both a qumode- and a qubit-based variant. Unless otherwise stated, we assume that $k$ is constant with respect to $n$. For both variants of $\left(\ell,n,k\right)$-HSMT, we consider $\ell$-long input sequences
\begin{equation}
    \bm{x}=\left(\bm{x_1},\ldots,\bm{x_\ell}\right)^\intercal,
\end{equation}
where
\begin{equation}
    \bm{x_i}=\left(\bm{\alpha_i},\bm{\beta_i},\gamma_i,\bm{\kappa_i}\right)\in\left[n\right]^k\times\left[n\right]^k\times\mathbb{R}\times\left(\mathbb{R}^k\times\mathbb{R}^{2^k}\right).
\end{equation}
The task is to sample from the measurement distribution $p\left(y_1,\ldots,y_\ell\mid\bm{x_1},\ldots,\bm{x_\ell}\right)$ of either a qumode- or a qubit-based $k$-HRNN, where the $i$th unit cell applies operations parameterized by---as described in Sec.~\ref{sec:hrnn}---$\left(\bm{\alpha_i},\bm{\beta_i},\gamma_i,\bm{\kappa_i}\right)$. By construction, a $k$-HRNN can sample from this distribution exactly with $n$ qumodes or qubits and a constant number of gates per unit cell.

\subsection{Expressivity Separations}

We now give memory lower bounds on classical neural networks performing $\left(\ell,n,k\right)$-HSMT to a finite backward cross entropy whenever $k\geq 2$. We here only give a sketch of our arguments, leaving detailed proofs for Appendix~\ref{sec:proof_of_express_sep}. Our main results follow from two facts that we prove:
\begin{proposition}
    Let $p\left(\bm{y}\mid\bm{x}\right)$ be the distribution corresponding to the $\left(\ell,n,k\right)$-HSMT task for some $\ell\geq\binom{n}{k}+n$. There exists a $\binom{n}{k}-1$-dimensional subspace $\mathcal{X}$ of inputs where the following holds: consider any triplet of distinct $\binom{n}{k}$-token sequences $\left(\bm{x},\bm{x'},\bm{x''}\right)\in\mathcal{X}^{\times 3}$. There exists some $n$-token sequence $\bm{\tilde{x}}$ such that, for all $\bm{y}$,
    \begin{equation}\label{eq:prod_equal_zero}
        p\left(\bm{y}\mid\bm{x}\oplus\bm{\tilde{x}}\right)\times p\left(\bm{y}\mid\bm{x'}\oplus\bm{\tilde{x}}\right)\times p\left(\bm{y}\mid\bm{x''}\oplus\bm{\tilde{x}}\right)=0.
    \end{equation}
    Furthermore, this property is due to \emph{contextuality}.\label{prop:cont_meas_seqs}
\end{proposition}
\begin{proposition}
    Any classical neural network satisfying Assumptions~\ref{ass:c2} and~\ref{ass:strong_sub} with latent space dimension less than $\binom{n}{k}-1$ maps some finite-volume subspace of $\mathcal{X}$ to a single point in latent space.\label{prop:fiber_bundle_struct}
\end{proposition}

Before proceeding with a proof sketch, we give some intuition behind these two Propositions---as well as their implications---by considering $k=2$ as a concrete example.

\subsubsection{Intuition Behind Proposition~\ref{prop:cont_meas_seqs}}

We begin with a simple example of Proposition~\ref{prop:cont_meas_seqs} for $n=2$ qubits. Recall that the $\left(\ell,2,2\right)$-HSMT correponds to simulating the dynamics of a $2$-HRNN for $\ell$ time steps. We consider a triplet of $2$-token input sequences $\left(\bm{x},\bm{x'},\bm{x''}\right)$ which, when conditioned on the first two output tokens $y_1=+1$ and $y_2=+1$, respectively correspond to the preparation of three latent states:
\begin{align}
    \ket{\lambda_2}&:=\ket{++},\\
    \ket{\lambda_2'}&:=\operatorname{C}Z_{1,2}\ket{++},\\
    \ket{\lambda_2''}&:=\ket{00}
\end{align}
in the associated $k$-HRNN. Stabilizers of these three states are written as the rows of Table~\ref{tab:magic_square}. Note that the rows and columns all multiply to the identity matrix $I$ outside of the third column, which multiplies to $-I$. This is a famous example of a \emph{Mermin--Peres magic square}~\cite{PhysRevLett.65.3373}: no classical assignment of these observables are consistent with the parity constraints enforced by their products, even though each row and column of operators are mutually commuting. This property is known as \emph{contextuality}, which counterintuitively is known to be satisfiable in quantum theory~\cite{kochen1975problem,Abramsky_2011}.
\begin{table}
    \begin{center}
        \caption{An example of quantum contextuality using a Mermin--Peres magic square~\cite{PhysRevLett.65.3373}. All operators in each row and column commute. Additionally, the product of each row and column is the identity operator, except for the final column, which gives minus the identity. Thus, definite classical values cannot be assigned to each operator without yielding a contradiction.\label{tab:magic_square}}
        \begin{tabular}{c|c|c}
            $X_1$ & $X_2$ & $X_1 X_2$\\\hline
            $X_1 Z_2$ & $Z_1 X_2$ & $-X_1 Z_1 X_2 Z_2$\\\hline
            $Z_2$ & $Z_1$ & $Z_1 Z_2$
        \end{tabular}
    \end{center}
\end{table}

We now use this magic square to complete the sequences $\left(\bm{x},\bm{x'},\bm{x''}\right)$ to three sequences with associated conditional distributions having infinite cross entropy. We consider a $2$-token sequence $\bm{\tilde{x}}$ corresponding to the measurement of $Z_1 Z_2$ followed by $X_1 X_2$. From the definition of $\left(\ell,n,k\right)$-HSMT, we have:
\begin{align}
    p\left(y_3=+1,y_4=-1\mid\bm{x}\oplus\bm{\tilde{x}},y_1=+1,y_2=+1\right)&=0;\\
    p\left(y_3=+1,y_4=+1\mid\bm{x'}\oplus\bm{\tilde{x}},y_1=+1,y_2=+1\right)&=0;\\
    p\left(y_3=-1\mid\bm{x''}\oplus\bm{\tilde{x}},y_1=+1,y_2=+1\right)&=0.
\end{align}
This is enough to show Eq.~\eqref{eq:prod_equal_zero} when additionally conditioned on $y_1=y_2=+1$. Similar arguments for all other choices of $y_1$ and $y_2$ show this is indeed the case for all $\bm{y}$. This argument also generalizes for any choice of three states $\ket{\lambda_2},\ket{\lambda_2'},\ket{\lambda_2''}$, as long as they are stabilized by operators comprising the rows of a Mermin--Peres magic square.

For general $k$---as well as for qumodes---the proof of the Proposition proceeds similarly to the logic laid out here, though now the magic square is composed of operators in the $k-1$st level of the Clifford hierarchy~\cite{gottesman1999demonstrating} rather than Pauli operators. Formal proofs for qubits and qumodes are given in Appendix~\ref{sec:dist_meas_seqs} as the proofs of Lemmas~\ref{lemma:hypergraph_state_context} and~\ref{lemma:qumode_hypergraph_state_context}, respectively.

\subsubsection{Intuition Behind Proposition~\ref{prop:fiber_bundle_struct}}

We now sketch our proof of Proposition~\ref{prop:fiber_bundle_struct}, which relies on Assumptions~\ref{ass:c2} and~\ref{ass:strong_sub}. At a basic level, this Proposition is a statement that if the classical neural network is ``sufficiently nice,'' $\mathcal{X}$ cannot be injectively represented in the latent space $L$ of the network if $\dim\left(L\right)<\dim\left(\mathcal{X}\right)=\binom{n}{k}-1$.

To demonstrate this we show that any such classical network exhibits a \emph{fiber bundle} structure on some finite-volume subspace of inputs, where each fiber is ``sufficiently large'' such that it intersects $\mathcal{X}$ at three points (or more). We give a sketch of this structure in Fig.~\ref{fig:fiber_bundle}. Intuitively, this result can be seen as a consequence of the inverse function theorem, where Assumption~\ref{ass:strong_sub} ensures that the theorem holds on some finite-volume subspace of inputs. More formally we use a so-called ``global implicit function theorem'' described in Ref.~\cite{rabier1997ehresmann}.

Formal proofs of this Proposition differ slightly depending on whether the classical model is an autoregressive model or an encoder-decoder model. These two cases are given as Theorems~\ref{thm:online_sep_qubit} and~\ref{thm:enc_dec_sep_qubit}, respectively, both stated and proved in Appendix~\ref{sec:proof_of_express_sep}.
\begin{figure}
    \begin{center}
        \includegraphics[width=0.5\linewidth]{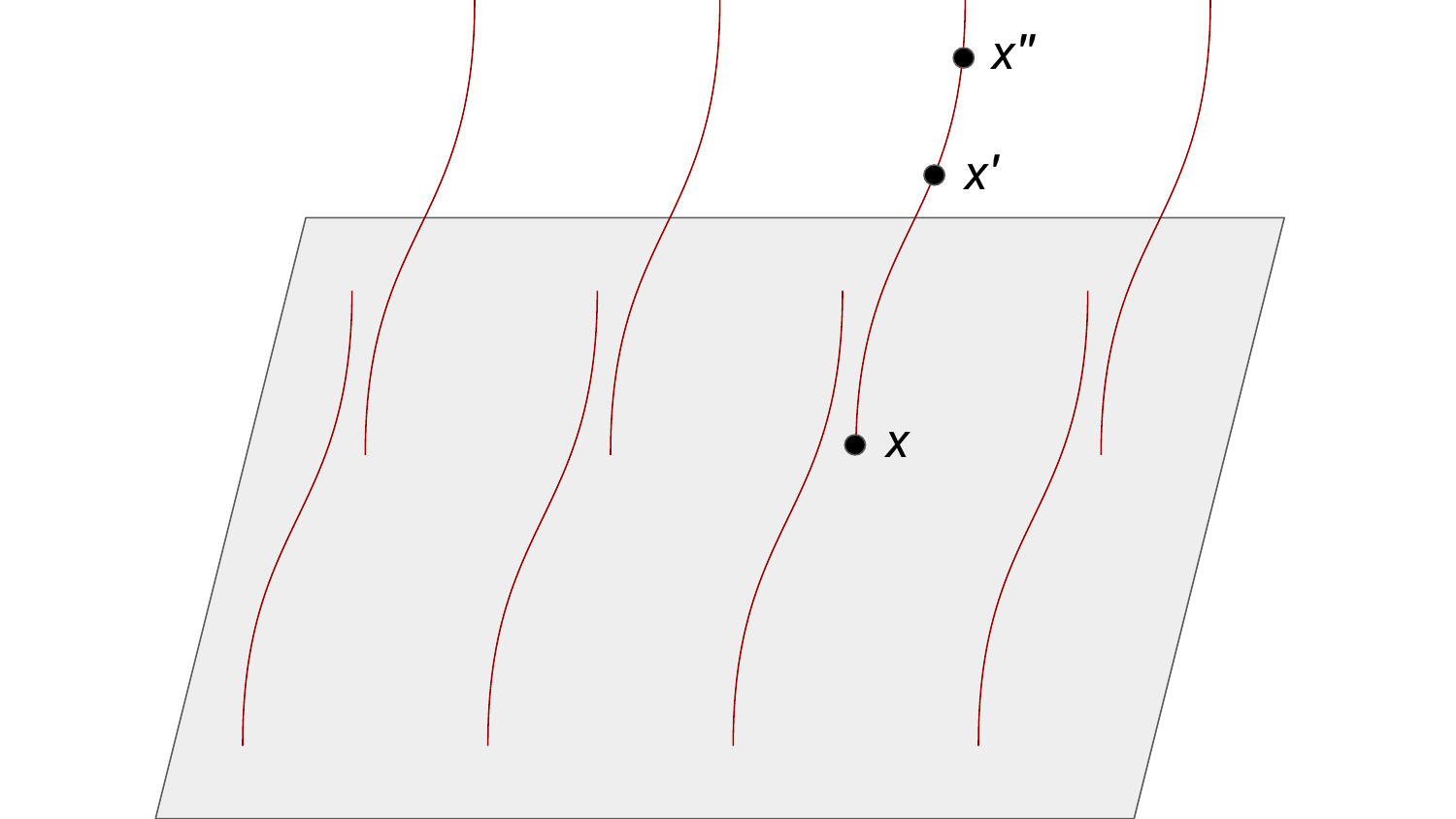}
        \caption{A representation of the noninjectivity of a classical network performing $\left(\ell,n,k\right)$-HSMT as described in Proposition~\ref{prop:fiber_bundle_struct}. On some finite-volume subspace of inputs, any classical network noninjectively maps all points on a red line (a ``fiber'') to the same point in latent space, isomorphic to the gray parallelogram (the ``base space''). $\bm{x},\bm{x'},\bm{x''}\in\mathcal{X}$ label input sequences as described in Proposition~\ref{prop:cont_meas_seqs}. Due to the noninjectivity of the network, $\bm{x}$, $\bm{x'}$, and $\bm{x''}$ are indistinguishable to the classical neural network.\label{fig:fiber_bundle}}
    \end{center}
\end{figure}

\subsubsection{Main Result}

Our main result follows from Propositions~\ref{prop:cont_meas_seqs} and~\ref{prop:fiber_bundle_struct}. As previously stated, our formal proofs and statements depend on whether the classical model is an autoregressive model (stated as  Theorem~\ref{thm:online_sep_qubit}) or an encoder-decoder model (stated as Theorem~\ref{thm:enc_dec_sep_qubit}), both stated and proved in Appendix~\ref{sec:proof_of_express_sep}. For simplicity, we here gloss over the details which distinguish these two Theorems and give a high-level proof sketch of our separations.
\begin{theorem}[Memory lower bound for $\left(\ell,n,k\right)$-HSMT, informal]
    Consider the qubit or qumode $\left(\ell,n,k\right)$-HSMT task as described in Sec.~\ref{sec:k_chmt} with $\ell\geq\binom{n}{k}+n$. A classical neural network satisfying Assumptions~\ref{ass:c2} and~\ref{ass:strong_sub} must have latent space dimension at least $\binom{n}{k}-1$ to perform the $\left(\ell,n,k\right)$-HSMT task to any finite backward cross entropy. A $k$-HRNN on $n$ qubits or qumodes and constant gate complexity per unit cell can perform this task to zero error.\label{thm:mem_lb_inf}
\end{theorem}
\begin{proof}[Proof sketch]
    We prove this statement by contradiction. Assume there exists a classical model of latent space dimension less than $\binom{n}{k}-1$ that performs the $\left(\ell,n,k\right)$-HSMT task to some finite backward cross entropy. Let $\mathcal{X}$ be as in Proposition~\ref{prop:cont_meas_seqs}. By Proposition~\ref{prop:fiber_bundle_struct}, there exists some $\binom{n}{k}$-token sequences $\bm{x},\bm{x'},\bm{x''}\in\mathcal{X}$ that map to the same point in the latent space of the classical model. By Proposition~\ref{prop:cont_meas_seqs}, these inputs can be completed to $\ell$-token sequences $\bm{x}\oplus\bm{\tilde{x}}$, $\bm{x'}\oplus\bm{\tilde{x}}$, and $\bm{x''}\oplus\bm{\tilde{x}}$, where the associated conditional distributions in the $\left(\ell,n,k\right)$-HSMT task have disjoint supports. However, as the model is unable to distinguish these inputs---as it maps these inputs to the same point in latent space---it is required to sample from some point in the intersection of their supports. This yields a contradiction.
\end{proof}

Before concluding we note that this lower bound also holds when $k$ scales with $n$, i.e., for $k=n/2$ this demonstrates an exponential classical memory lower bound for $\left(\ell,n,k\right)$-HSMT. The trade-off is that the $k=n/2$-HRNN which performs this task to zero error is not efficiently trainable by the discussion of Sec.~\ref{sec:quant_gen_mods}, though this separation may be of independent interest as an exponential quantum-classical communication complexity separation where the parties in the quantum setting have identical strategies implemented using only a constant number of gates.

\section{Discussion}\label{sec:conc}

Our results construct a hierarchy quantum neural networks that not only avoids the trainability barriers typically encountered in quantum machine learning models, but also more expressive than classical neural networks by an arbitrary polynomial factor. We achieve this by considering a certain sequence learning task exhibiting a large amount of contextuality, and then demonstrating that the presence of this contextuality leads to a classical memory lower-bound in performing the task. Though the sequence lengths needed to witness a separation are large---on the order of $n^k$---this is natural in a streaming setting, where one wishes to perform sequence modeling in a way that minimizes the resources required to perform inference per token from an arbitrarily-long sequence.

Our introduced hierarchy of quantum models tracts nicely with experimental capabilities: given the ability to implement $k$-local gates, one can achieve a $\operatorname{\Theta}\left(\binom{n}{k}\right)$-memory separation over classical neural networks in performing the translation task we here introduce. This makes our hierarchy a natural fit for implementation on Rydberg atom-based quantum architectures which naturally implement multi-control $Z$ gates~\cite{isenhower2011multibit}. Our model also has the benefit of being composed of operations for which error correction has successfully been experimentally demonstrated~\cite{bluvstein2023logical}. Cavity quantum electrodynamical (QED) systems are also a natural platform to implement our model, as coherent information from the CV degrees of freedom can be stored in the cavity while atomic degrees of freedom are read out to perform the necessary measurements~\cite{Kerman_2013,RevModPhys.93.025005}.

More work is required for demonstrating the practicality of our introduced advantage. First, qubit-based $k$-HRNNs are only efficiently trainable in a setting where the model is not translationally invariant; though this matters less in an offline learning scenario, it is unclear how a qubit $k$-HRNN would perform inference in a purely online setting. An alternative approach could be to consider an encoder-decoder variant of the qubit $k$-HRNN where translational invariance is maintained in both the encoder and decoder, though not together. Second, our protocol is here considered in a noise-free context, far from the experimental capabilities of most current quantum devices. However, our results follow from measurement sequences similar to those seen in the setting of quantum pseudo-telepathy games~\cite{Quanta22}, and previous work~\cite{bravyi2020quantum,caha2023colossal} has demonstrated instances where quantum separations based on pseudo-telepathy games can be made noise resilient. We hope to explore these connections in future work.

\change{Furthermore, it is not yet clear how commonly contextuality---the specific correlational structure with which we prove our separation---appears in classical data. Concurrent work~\cite{teo2025kcontextualityheuristicmemoryseparations} has begun to investigate this, showing that standard machine learning data sets do possess nontrivial contextual correlations and that this inhibits classical representations of the data. Of course, this does not imply that learning such data automatically yields a ``quantum advantage'' regarding model size. Quantum easiness is the more difficult side of the argument; outside of constructed settings such as ours, it is difficult to analytically show that a data set can be efficiently represented quantumly. Our result can be interpreted as a statement that semantic ambiguity measured as nontrivial contextuality inhibits classical representations of data but does not necessarily inhibit quantum ones, pointing toward data with this correlational structure as promising empirical testing grounds for quantum machine learning.}

HRNNs demonstrate that quantizing a very simple class of recurrent neural networks is enough to achieve a large expressivity separation over classical neural networks on a sequence modeling task. Given the tremendous progress in the performance of large language models in recent years with the introduction of, for instance, GPT-3~\cite{brown2020} and GPT-4~\cite{openai2023}, this begs the natural question: what can be achieved by quantizing more sophisticated classical models? We hope in the future to address this.

\section*{Acknowledgements}

We thank Cameron Calcluth and Mikhail D.\ Lukin for helpful discussion. E.R.A. was funded in part by the Walter Burke Institute for Theoretical Physics at Caltech and in part by the DARPA ONISQ program (grant number W911NF2010021). X.G. acknowledges support from the U.S.\ Department of Energy, Office of Science, National Quantum Information Science Research Centers, Quantum Systems Accelerator, NSF PFC grant No.\ PHYS 2317149, and start-up grants from CU Boulder.

\bibliography{main}

\begin{thebibliography}{100}%
\makeatletter
\providecommand \@ifxundefined [1]{%
 \@ifx{#1\undefined}
}%
\providecommand \@ifnum [1]{%
 \ifnum #1\expandafter \@firstoftwo
 \else \expandafter \@secondoftwo
 \fi
}%
\providecommand \@ifx [1]{%
 \ifx #1\expandafter \@firstoftwo
 \else \expandafter \@secondoftwo
 \fi
}%
\providecommand \natexlab [1]{#1}%
\providecommand \enquote  [1]{``#1''}%
\providecommand \bibnamefont  [1]{#1}%
\providecommand \bibfnamefont [1]{#1}%
\providecommand \citenamefont [1]{#1}%
\providecommand \href@noop [0]{\@secondoftwo}%
\providecommand \href [0]{\begingroup \@sanitize@url \@href}%
\providecommand \@href[1]{\@@startlink{#1}\@@href}%
\providecommand \@@href[1]{\endgroup#1\@@endlink}%
\providecommand \@sanitize@url [0]{\catcode `\\12\catcode `\$12\catcode
  `\&12\catcode `\#12\catcode `\^12\catcode `\_12\catcode `\%12\relax}%
\providecommand \@@startlink[1]{}%
\providecommand \@@endlink[0]{}%
\providecommand \url  [0]{\begingroup\@sanitize@url \@url }%
\providecommand \@url [1]{\endgroup\@href {#1}{\urlprefix }}%
\providecommand \urlprefix  [0]{URL }%
\providecommand \Eprint [0]{\href }%
\providecommand \doibase [0]{http://dx.doi.org/}%
\providecommand \selectlanguage [0]{\@gobble}%
\providecommand \bibinfo  [0]{\@secondoftwo}%
\providecommand \bibfield  [0]{\@secondoftwo}%
\providecommand \translation [1]{[#1]}%
\providecommand \BibitemOpen [0]{}%
\providecommand \bibitemStop [0]{}%
\providecommand \bibitemNoStop [0]{.\EOS\space}%
\providecommand \EOS [0]{\spacefactor3000\relax}%
\providecommand \BibitemShut  [1]{\csname bibitem#1\endcsname}%
\let\auto@bib@innerbib\@empty
\bibitem [{\citenamefont {Arute}\ \emph {et~al.}(2019)\citenamefont {Arute}
  \emph {et~al.}}]{arute2019quantum}%
  \BibitemOpen
  \bibfield  {author} {\bibinfo {author} {\bibfnamefont {F.}~\bibnamefont
  {Arute}} \emph {et~al.},\ }\bibfield  {title} {\enquote {\bibinfo {title}
  {Quantum supremacy using a programmable superconducting processor},}\ }\href
  {\doibase 10.1038/s41586-019-1666-5} {\bibfield  {journal} {\bibinfo
  {journal} {Nature}\ }\textbf {\bibinfo {volume} {574}},\ \bibinfo {pages}
  {505} (\bibinfo {year} {2019})}\BibitemShut {NoStop}%
\bibitem [{\citenamefont {Zhong}\ \emph {et~al.}(2020)\citenamefont {Zhong},
  \citenamefont {Wang}, \citenamefont {Deng}, \citenamefont {Chen},
  \citenamefont {Peng}, \citenamefont {Luo}, \citenamefont {Qin}, \citenamefont
  {Wu}, \citenamefont {Ding}, \citenamefont {Hu} \emph
  {et~al.}}]{doi:10.1126/science.abe8770}%
  \BibitemOpen
  \bibfield  {author} {\bibinfo {author} {\bibfnamefont {H.-S.}\ \bibnamefont
  {Zhong}}, \bibinfo {author} {\bibfnamefont {H.}~\bibnamefont {Wang}},
  \bibinfo {author} {\bibfnamefont {Y.-H.}\ \bibnamefont {Deng}}, \bibinfo
  {author} {\bibfnamefont {M.-C.}\ \bibnamefont {Chen}}, \bibinfo {author}
  {\bibfnamefont {L.-C.}\ \bibnamefont {Peng}}, \bibinfo {author}
  {\bibfnamefont {Y.-H.}\ \bibnamefont {Luo}}, \bibinfo {author} {\bibfnamefont
  {J.}~\bibnamefont {Qin}}, \bibinfo {author} {\bibfnamefont {D.}~\bibnamefont
  {Wu}}, \bibinfo {author} {\bibfnamefont {X.}~\bibnamefont {Ding}}, \bibinfo
  {author} {\bibfnamefont {Y.}~\bibnamefont {Hu}},  \emph {et~al.},\ }\bibfield
   {title} {\enquote {\bibinfo {title} {Quantum computational advantage using
  photons},}\ }\href {\doibase 10.1126/science.abe8770} {\bibfield  {journal}
  {\bibinfo  {journal} {Science}\ }\textbf {\bibinfo {volume} {370}},\ \bibinfo
  {pages} {1460} (\bibinfo {year} {2020})}\BibitemShut {NoStop}%
\bibitem [{\citenamefont {Wu}\ \emph {et~al.}(2021)\citenamefont {Wu},
  \citenamefont {Bao}, \citenamefont {Cao}, \citenamefont {Chen}, \citenamefont
  {Chen}, \citenamefont {Chen}, \citenamefont {Chung}, \citenamefont {Deng},
  \citenamefont {Du}, \citenamefont {Fan} \emph
  {et~al.}}]{PhysRevLett.127.180501}%
  \BibitemOpen
  \bibfield  {author} {\bibinfo {author} {\bibfnamefont {Y.}~\bibnamefont
  {Wu}}, \bibinfo {author} {\bibfnamefont {W.-S.}\ \bibnamefont {Bao}},
  \bibinfo {author} {\bibfnamefont {S.}~\bibnamefont {Cao}}, \bibinfo {author}
  {\bibfnamefont {F.}~\bibnamefont {Chen}}, \bibinfo {author} {\bibfnamefont
  {M.-C.}\ \bibnamefont {Chen}}, \bibinfo {author} {\bibfnamefont
  {X.}~\bibnamefont {Chen}}, \bibinfo {author} {\bibfnamefont {T.-H.}\
  \bibnamefont {Chung}}, \bibinfo {author} {\bibfnamefont {H.}~\bibnamefont
  {Deng}}, \bibinfo {author} {\bibfnamefont {Y.}~\bibnamefont {Du}}, \bibinfo
  {author} {\bibfnamefont {D.}~\bibnamefont {Fan}},  \emph {et~al.},\
  }\bibfield  {title} {\enquote {\bibinfo {title} {Strong quantum computational
  advantage using a superconducting quantum processor},}\ }\href {\doibase
  10.1103/PhysRevLett.127.180501} {\bibfield  {journal} {\bibinfo  {journal}
  {Phys. Rev. Lett.}\ }\textbf {\bibinfo {volume} {127}},\ \bibinfo {pages}
  {180501} (\bibinfo {year} {2021})}\BibitemShut {NoStop}%
\bibitem [{\citenamefont {Zhong}\ \emph {et~al.}(2021)\citenamefont {Zhong},
  \citenamefont {Deng}, \citenamefont {Qin}, \citenamefont {Wang},
  \citenamefont {Chen}, \citenamefont {Peng}, \citenamefont {Luo},
  \citenamefont {Wu}, \citenamefont {Gong}, \citenamefont {Su} \emph
  {et~al.}}]{PhysRevLett.127.180502}%
  \BibitemOpen
  \bibfield  {author} {\bibinfo {author} {\bibfnamefont {H.-S.}\ \bibnamefont
  {Zhong}}, \bibinfo {author} {\bibfnamefont {Y.-H.}\ \bibnamefont {Deng}},
  \bibinfo {author} {\bibfnamefont {J.}~\bibnamefont {Qin}}, \bibinfo {author}
  {\bibfnamefont {H.}~\bibnamefont {Wang}}, \bibinfo {author} {\bibfnamefont
  {M.-C.}\ \bibnamefont {Chen}}, \bibinfo {author} {\bibfnamefont {L.-C.}\
  \bibnamefont {Peng}}, \bibinfo {author} {\bibfnamefont {Y.-H.}\ \bibnamefont
  {Luo}}, \bibinfo {author} {\bibfnamefont {D.}~\bibnamefont {Wu}}, \bibinfo
  {author} {\bibfnamefont {S.-Q.}\ \bibnamefont {Gong}}, \bibinfo {author}
  {\bibfnamefont {H.}~\bibnamefont {Su}},  \emph {et~al.},\ }\bibfield  {title}
  {\enquote {\bibinfo {title} {Phase-programmable {G}aussian boson sampling
  using stimulated squeezed light},}\ }\href {\doibase
  10.1103/PhysRevLett.127.180502} {\bibfield  {journal} {\bibinfo  {journal}
  {Phys. Rev. Lett.}\ }\textbf {\bibinfo {volume} {127}},\ \bibinfo {pages}
  {180502} (\bibinfo {year} {2021})}\BibitemShut {NoStop}%
\bibitem [{\citenamefont {Zhu}\ \emph {et~al.}(2022)\citenamefont {Zhu},
  \citenamefont {Cao}, \citenamefont {Chen}, \citenamefont {Chen},
  \citenamefont {Chen}, \citenamefont {Chung}, \citenamefont {Deng},
  \citenamefont {Du}, \citenamefont {Fan}, \citenamefont {Gong} \emph
  {et~al.}}]{ZHU2022240}%
  \BibitemOpen
  \bibfield  {author} {\bibinfo {author} {\bibfnamefont {Q.}~\bibnamefont
  {Zhu}}, \bibinfo {author} {\bibfnamefont {S.}~\bibnamefont {Cao}}, \bibinfo
  {author} {\bibfnamefont {F.}~\bibnamefont {Chen}}, \bibinfo {author}
  {\bibfnamefont {M.-C.}\ \bibnamefont {Chen}}, \bibinfo {author}
  {\bibfnamefont {X.}~\bibnamefont {Chen}}, \bibinfo {author} {\bibfnamefont
  {T.-H.}\ \bibnamefont {Chung}}, \bibinfo {author} {\bibfnamefont
  {H.}~\bibnamefont {Deng}}, \bibinfo {author} {\bibfnamefont {Y.}~\bibnamefont
  {Du}}, \bibinfo {author} {\bibfnamefont {D.}~\bibnamefont {Fan}}, \bibinfo
  {author} {\bibfnamefont {M.}~\bibnamefont {Gong}},  \emph {et~al.},\
  }\bibfield  {title} {\enquote {\bibinfo {title} {Quantum computational
  advantage via 60-qubit 24-cycle random circuit sampling},}\ }\href {\doibase
  10.1016/j.scib.2021.10.017} {\bibfield  {journal} {\bibinfo  {journal} {Sci.
  Bull.}\ }\textbf {\bibinfo {volume} {67}},\ \bibinfo {pages} {240} (\bibinfo
  {year} {2022})}\BibitemShut {NoStop}%
\bibitem [{\citenamefont {Hangleiter}\ and\ \citenamefont
  {Eisert}(2022)}]{hangleiter2022}%
  \BibitemOpen
  \bibfield  {author} {\bibinfo {author} {\bibfnamefont {D.}~\bibnamefont
  {Hangleiter}}\ and\ \bibinfo {author} {\bibfnamefont {J.}~\bibnamefont
  {Eisert}},\ }\href@noop {} {\enquote {\bibinfo {title} {Computational
  advantage of quantum random sampling},}\ } (\bibinfo {year} {2022}),\ \Eprint
  {http://arxiv.org/abs/2206.04079}{arXiv:2206.04079 [quant-ph]}\BibitemShut
  {NoStop}%
\bibitem [{\citenamefont {Aaronson}\ and\ \citenamefont
  {Arkhipov}(2011)}]{10.1145/1993636.1993682}%
  \BibitemOpen
  \bibfield  {author} {\bibinfo {author} {\bibfnamefont {S.}~\bibnamefont
  {Aaronson}}\ and\ \bibinfo {author} {\bibfnamefont {A.}~\bibnamefont
  {Arkhipov}},\ }\bibfield  {title} {\enquote {\bibinfo {title} {The
  computational complexity of linear optics},}\ }in\ \href {\doibase
  10.1145/1993636.1993682} {\emph {\bibinfo {booktitle} {Proceedings of the
  Forty-Third Annual ACM Symposium on Theory of Computing}}},\ \bibinfo {series
  and number} {STOC '11}\ (\bibinfo  {publisher} {Association for Computing
  Machinery},\ \bibinfo {address} {New York, NY, USA},\ \bibinfo {year}
  {2011})\ pp.\ \bibinfo {pages} {333--342}\BibitemShut {NoStop}%
\bibitem [{\citenamefont {Bouland}\ \emph {et~al.}(2019)\citenamefont
  {Bouland}, \citenamefont {Fefferman}, \citenamefont {Nirkhe},\ and\
  \citenamefont {Vazirani}}]{bouland2019complexity}%
  \BibitemOpen
  \bibfield  {author} {\bibinfo {author} {\bibfnamefont {A.}~\bibnamefont
  {Bouland}}, \bibinfo {author} {\bibfnamefont {B.}~\bibnamefont {Fefferman}},
  \bibinfo {author} {\bibfnamefont {C.}~\bibnamefont {Nirkhe}}, \ and\ \bibinfo
  {author} {\bibfnamefont {U.}~\bibnamefont {Vazirani}},\ }\bibfield  {title}
  {\enquote {\bibinfo {title} {On the complexity and verification of quantum
  random circuit sampling},}\ }\href {\doibase 10.1038/s41567-018-0318-2}
  {\bibfield  {journal} {\bibinfo  {journal} {Nat. Phys.}\ }\textbf {\bibinfo
  {volume} {15}},\ \bibinfo {pages} {159} (\bibinfo {year} {2019})}\BibitemShut
  {NoStop}%
\bibitem [{\citenamefont {Aaronson}\ and\ \citenamefont
  {Gunn}(2020)}]{v016a011}%
  \BibitemOpen
  \bibfield  {author} {\bibinfo {author} {\bibfnamefont {S.}~\bibnamefont
  {Aaronson}}\ and\ \bibinfo {author} {\bibfnamefont {S.}~\bibnamefont
  {Gunn}},\ }\bibfield  {title} {\enquote {\bibinfo {title} {On the classical
  hardness of spoofing linear cross-entropy benchmarking},}\ }\href {\doibase
  10.4086/toc.2020.v016a011} {\bibfield  {journal} {\bibinfo  {journal} {Theory
  Comput.}\ }\textbf {\bibinfo {volume} {16}},\ \bibinfo {pages} {1} (\bibinfo
  {year} {2020})}\BibitemShut {NoStop}%
\bibitem [{\citenamefont {Schuld}\ \emph {et~al.}(2015)\citenamefont {Schuld},
  \citenamefont {Sinayskiy},\ and\ \citenamefont
  {Petruccione}}]{schuld2015introduction}%
  \BibitemOpen
  \bibfield  {author} {\bibinfo {author} {\bibfnamefont {M.}~\bibnamefont
  {Schuld}}, \bibinfo {author} {\bibfnamefont {I.}~\bibnamefont {Sinayskiy}}, \
  and\ \bibinfo {author} {\bibfnamefont {F.}~\bibnamefont {Petruccione}},\
  }\bibfield  {title} {\enquote {\bibinfo {title} {An introduction to quantum
  machine learning},}\ }\href {\doibase 10.1080/00107514.2014.964942}
  {\bibfield  {journal} {\bibinfo  {journal} {Contemp. Phys.}\ }\textbf
  {\bibinfo {volume} {56}},\ \bibinfo {pages} {172} (\bibinfo {year}
  {2015})}\BibitemShut {NoStop}%
\bibitem [{\citenamefont {Biamonte}\ \emph {et~al.}(2017)\citenamefont
  {Biamonte}, \citenamefont {Wittek}, \citenamefont {Pancotti}, \citenamefont
  {Rebentrost}, \citenamefont {Wiebe},\ and\ \citenamefont
  {Lloyd}}]{biamonte2017quantum}%
  \BibitemOpen
  \bibfield  {author} {\bibinfo {author} {\bibfnamefont {J.}~\bibnamefont
  {Biamonte}}, \bibinfo {author} {\bibfnamefont {P.}~\bibnamefont {Wittek}},
  \bibinfo {author} {\bibfnamefont {N.}~\bibnamefont {Pancotti}}, \bibinfo
  {author} {\bibfnamefont {P.}~\bibnamefont {Rebentrost}}, \bibinfo {author}
  {\bibfnamefont {N.}~\bibnamefont {Wiebe}}, \ and\ \bibinfo {author}
  {\bibfnamefont {S.}~\bibnamefont {Lloyd}},\ }\bibfield  {title} {\enquote
  {\bibinfo {title} {Quantum machine learning},}\ }\href {\doibase
  10.1038/nature23474} {\bibfield  {journal} {\bibinfo  {journal} {Nature}\
  }\textbf {\bibinfo {volume} {549}},\ \bibinfo {pages} {195} (\bibinfo {year}
  {2017})}\BibitemShut {NoStop}%
\bibitem [{\citenamefont {Amin}\ \emph {et~al.}(2018)\citenamefont {Amin},
  \citenamefont {Andriyash}, \citenamefont {Rolfe}, \citenamefont
  {Kulchytskyy},\ and\ \citenamefont {Melko}}]{PhysRevX.8.021050}%
  \BibitemOpen
  \bibfield  {author} {\bibinfo {author} {\bibfnamefont {M.~H.}\ \bibnamefont
  {Amin}}, \bibinfo {author} {\bibfnamefont {E.}~\bibnamefont {Andriyash}},
  \bibinfo {author} {\bibfnamefont {J.}~\bibnamefont {Rolfe}}, \bibinfo
  {author} {\bibfnamefont {B.}~\bibnamefont {Kulchytskyy}}, \ and\ \bibinfo
  {author} {\bibfnamefont {R.}~\bibnamefont {Melko}},\ }\bibfield  {title}
  {\enquote {\bibinfo {title} {Quantum {B}oltzmann machine},}\ }\href {\doibase
  10.1103/PhysRevX.8.021050} {\bibfield  {journal} {\bibinfo  {journal} {Phys.
  Rev. X}\ }\textbf {\bibinfo {volume} {8}},\ \bibinfo {pages} {021050}
  (\bibinfo {year} {2018})}\BibitemShut {NoStop}%
\bibitem [{\citenamefont {Lloyd}\ and\ \citenamefont
  {Weedbrook}(2018)}]{PhysRevLett.121.040502}%
  \BibitemOpen
  \bibfield  {author} {\bibinfo {author} {\bibfnamefont {S.}~\bibnamefont
  {Lloyd}}\ and\ \bibinfo {author} {\bibfnamefont {C.}~\bibnamefont
  {Weedbrook}},\ }\bibfield  {title} {\enquote {\bibinfo {title} {Quantum
  generative adversarial learning},}\ }\href {\doibase
  10.1103/PhysRevLett.121.040502} {\bibfield  {journal} {\bibinfo  {journal}
  {Phys. Rev. Lett.}\ }\textbf {\bibinfo {volume} {121}},\ \bibinfo {pages}
  {040502} (\bibinfo {year} {2018})}\BibitemShut {NoStop}%
\bibitem [{\citenamefont {Schuld}\ and\ \citenamefont
  {Killoran}(2019)}]{PhysRevLett.122.040504}%
  \BibitemOpen
  \bibfield  {author} {\bibinfo {author} {\bibfnamefont {M.}~\bibnamefont
  {Schuld}}\ and\ \bibinfo {author} {\bibfnamefont {N.}~\bibnamefont
  {Killoran}},\ }\bibfield  {title} {\enquote {\bibinfo {title} {Quantum
  machine learning in feature {H}ilbert spaces},}\ }\href {\doibase
  10.1103/PhysRevLett.122.040504} {\bibfield  {journal} {\bibinfo  {journal}
  {Phys. Rev. Lett.}\ }\textbf {\bibinfo {volume} {122}},\ \bibinfo {pages}
  {040504} (\bibinfo {year} {2019})}\BibitemShut {NoStop}%
\bibitem [{\citenamefont {Killoran}\ \emph {et~al.}(2019)\citenamefont
  {Killoran}, \citenamefont {Bromley}, \citenamefont {Arrazola}, \citenamefont
  {Schuld}, \citenamefont {Quesada},\ and\ \citenamefont
  {Lloyd}}]{PhysRevResearch.1.033063}%
  \BibitemOpen
  \bibfield  {author} {\bibinfo {author} {\bibfnamefont {N.}~\bibnamefont
  {Killoran}}, \bibinfo {author} {\bibfnamefont {T.~R.}\ \bibnamefont
  {Bromley}}, \bibinfo {author} {\bibfnamefont {J.~M.}\ \bibnamefont
  {Arrazola}}, \bibinfo {author} {\bibfnamefont {M.}~\bibnamefont {Schuld}},
  \bibinfo {author} {\bibfnamefont {N.}~\bibnamefont {Quesada}}, \ and\
  \bibinfo {author} {\bibfnamefont {S.}~\bibnamefont {Lloyd}},\ }\bibfield
  {title} {\enquote {\bibinfo {title} {Continuous-variable quantum neural
  networks},}\ }\href {\doibase 10.1103/PhysRevResearch.1.033063} {\bibfield
  {journal} {\bibinfo  {journal} {Phys. Rev. Research}\ }\textbf {\bibinfo
  {volume} {1}},\ \bibinfo {pages} {033063} (\bibinfo {year}
  {2019})}\BibitemShut {NoStop}%
\bibitem [{\citenamefont {Gao}\ \emph {et~al.}(2018)\citenamefont {Gao},
  \citenamefont {Zhang},\ and\ \citenamefont
  {Duan}}]{doi:10.1126/sciadv.aat9004}%
  \BibitemOpen
  \bibfield  {author} {\bibinfo {author} {\bibfnamefont {X.}~\bibnamefont
  {Gao}}, \bibinfo {author} {\bibfnamefont {Z.-Y.}\ \bibnamefont {Zhang}}, \
  and\ \bibinfo {author} {\bibfnamefont {L.-M.}\ \bibnamefont {Duan}},\
  }\bibfield  {title} {\enquote {\bibinfo {title} {A quantum machine learning
  algorithm based on generative models},}\ }\href {\doibase
  10.1126/sciadv.aat9004} {\bibfield  {journal} {\bibinfo  {journal} {Sci.
  Adv.}\ }\textbf {\bibinfo {volume} {4}},\ \bibinfo {pages} {eaat9004}
  (\bibinfo {year} {2018})}\BibitemShut {NoStop}%
\bibitem [{\citenamefont {Coyle}\ \emph {et~al.}(2020)\citenamefont {Coyle},
  \citenamefont {Mills}, \citenamefont {Danos},\ and\ \citenamefont
  {Kashefi}}]{coyle2020born}%
  \BibitemOpen
  \bibfield  {author} {\bibinfo {author} {\bibfnamefont {B.}~\bibnamefont
  {Coyle}}, \bibinfo {author} {\bibfnamefont {D.}~\bibnamefont {Mills}},
  \bibinfo {author} {\bibfnamefont {V.}~\bibnamefont {Danos}}, \ and\ \bibinfo
  {author} {\bibfnamefont {E.}~\bibnamefont {Kashefi}},\ }\bibfield  {title}
  {\enquote {\bibinfo {title} {The {B}orn supremacy: quantum advantage and
  training of an {I}sing {B}orn machine},}\ }\href {\doibase
  10.1038/s41534-020-00288-9} {\bibfield  {journal} {\bibinfo  {journal} {npj
  Quantum Inf.}\ }\textbf {\bibinfo {volume} {6}},\ \bibinfo {pages} {1}
  (\bibinfo {year} {2020})}\BibitemShut {NoStop}%
\bibitem [{\citenamefont {Du}\ \emph {et~al.}(2020)\citenamefont {Du},
  \citenamefont {Hsieh}, \citenamefont {Liu},\ and\ \citenamefont
  {Tao}}]{PhysRevResearch.2.033125}%
  \BibitemOpen
  \bibfield  {author} {\bibinfo {author} {\bibfnamefont {Y.}~\bibnamefont
  {Du}}, \bibinfo {author} {\bibfnamefont {M.-H.}\ \bibnamefont {Hsieh}},
  \bibinfo {author} {\bibfnamefont {T.}~\bibnamefont {Liu}}, \ and\ \bibinfo
  {author} {\bibfnamefont {D.}~\bibnamefont {Tao}},\ }\bibfield  {title}
  {\enquote {\bibinfo {title} {Expressive power of parametrized quantum
  circuits},}\ }\href {\doibase 10.1103/PhysRevResearch.2.033125} {\bibfield
  {journal} {\bibinfo  {journal} {Phys. Rev. Research}\ }\textbf {\bibinfo
  {volume} {2}},\ \bibinfo {pages} {033125} (\bibinfo {year}
  {2020})}\BibitemShut {NoStop}%
\bibitem [{\citenamefont {Sweke}\ \emph {et~al.}(2021)\citenamefont {Sweke},
  \citenamefont {Seifert}, \citenamefont {Hangleiter},\ and\ \citenamefont
  {Eisert}}]{sweke2021quantum}%
  \BibitemOpen
  \bibfield  {author} {\bibinfo {author} {\bibfnamefont {R.}~\bibnamefont
  {Sweke}}, \bibinfo {author} {\bibfnamefont {J.-P.}\ \bibnamefont {Seifert}},
  \bibinfo {author} {\bibfnamefont {D.}~\bibnamefont {Hangleiter}}, \ and\
  \bibinfo {author} {\bibfnamefont {J.}~\bibnamefont {Eisert}},\ }\bibfield
  {title} {\enquote {\bibinfo {title} {On the {Q}uantum versus {C}lassical
  {L}earnability of {D}iscrete {D}istributions},}\ }\href {\doibase
  10.22331/q-2021-03-23-417} {\bibfield  {journal} {\bibinfo  {journal}
  {{Quantum}}\ }\textbf {\bibinfo {volume} {5}},\ \bibinfo {pages} {417}
  (\bibinfo {year} {2021})}\BibitemShut {NoStop}%
\bibitem [{\citenamefont {Liu}\ \emph {et~al.}(2021)\citenamefont {Liu},
  \citenamefont {Arunachalam},\ and\ \citenamefont {Temme}}]{liu2021rigorous}%
  \BibitemOpen
  \bibfield  {author} {\bibinfo {author} {\bibfnamefont {Y.}~\bibnamefont
  {Liu}}, \bibinfo {author} {\bibfnamefont {S.}~\bibnamefont {Arunachalam}}, \
  and\ \bibinfo {author} {\bibfnamefont {K.}~\bibnamefont {Temme}},\ }\bibfield
   {title} {\enquote {\bibinfo {title} {A rigorous and robust quantum speed-up
  in supervised machine learning},}\ }\href {\doibase
  10.1038/s41567-021-01287-z} {\bibfield  {journal} {\bibinfo  {journal} {Nat.
  Phys.}\ }\textbf {\bibinfo {volume} {17}},\ \bibinfo {pages} {1013} (\bibinfo
  {year} {2021})}\BibitemShut {NoStop}%
\bibitem [{\citenamefont {Gil-Fuster}\ \emph {et~al.}(2024)\citenamefont
  {Gil-Fuster}, \citenamefont {Gyurik}, \citenamefont {P\'{e}rez-Salinas},\
  and\ \citenamefont {Dunjko}}]{gilfuster2024relation}%
  \BibitemOpen
  \bibfield  {author} {\bibinfo {author} {\bibfnamefont {E.}~\bibnamefont
  {Gil-Fuster}}, \bibinfo {author} {\bibfnamefont {C.}~\bibnamefont {Gyurik}},
  \bibinfo {author} {\bibfnamefont {A.}~\bibnamefont {P\'{e}rez-Salinas}}, \
  and\ \bibinfo {author} {\bibfnamefont {V.}~\bibnamefont {Dunjko}},\
  }\href@noop {} {\enquote {\bibinfo {title} {On the relation between
  trainability and dequantization of variational quantum learning models},}\ }
  (\bibinfo {year} {2024}),\ \Eprint
  {http://arxiv.org/abs/2406.07072}{arXiv:2406.07072 [quant-ph]}\BibitemShut
  {NoStop}%
\bibitem [{\citenamefont {Shor}(1997)}]{doi:10.1137/S0097539795293172}%
  \BibitemOpen
  \bibfield  {author} {\bibinfo {author} {\bibfnamefont {P.~W.}\ \bibnamefont
  {Shor}},\ }\bibfield  {title} {\enquote {\bibinfo {title} {Polynomial-time
  algorithms for prime factorization and discrete logarithms on a quantum
  computer},}\ }\href {\doibase 10.1137/S0097539795293172} {\bibfield
  {journal} {\bibinfo  {journal} {SIAM J. Comput.}\ }\textbf {\bibinfo {volume}
  {26}},\ \bibinfo {pages} {1484} (\bibinfo {year} {1997})}\BibitemShut
  {NoStop}%
\bibitem [{\citenamefont {Choromanska}\ \emph {et~al.}(2015)\citenamefont
  {Choromanska}, \citenamefont {Henaff}, \citenamefont {Mathieu}, \citenamefont
  {Arous},\ and\ \citenamefont {LeCun}}]{pmlr-v38-choromanska15}%
  \BibitemOpen
  \bibfield  {author} {\bibinfo {author} {\bibfnamefont {A.}~\bibnamefont
  {Choromanska}}, \bibinfo {author} {\bibfnamefont {M.}~\bibnamefont {Henaff}},
  \bibinfo {author} {\bibfnamefont {M.}~\bibnamefont {Mathieu}}, \bibinfo
  {author} {\bibfnamefont {G.~B.}\ \bibnamefont {Arous}}, \ and\ \bibinfo
  {author} {\bibfnamefont {Y.}~\bibnamefont {LeCun}},\ }\bibfield  {title}
  {\enquote {\bibinfo {title} {The loss surfaces of multilayer networks},}\
  }in\ \href {http://proceedings.mlr.press/v38/choromanska15.html} {\emph
  {\bibinfo {booktitle} {Proceedings of the Eighteenth International Conference
  on Artificial Intelligence and Statistics}}},\ \bibinfo {series} {Proceedings
  of Machine Learning Research}, Vol.~\bibinfo {volume} {38},\ \bibinfo
  {editor} {edited by\ \bibinfo {editor} {\bibfnamefont {G.}~\bibnamefont
  {Lebanon}}\ and\ \bibinfo {editor} {\bibfnamefont {S.~V.~N.}\ \bibnamefont
  {Vishwanathan}}}\ (\bibinfo  {publisher} {PMLR},\ \bibinfo {address} {San
  Diego, CA, USA},\ \bibinfo {year} {2015})\ pp.\ \bibinfo {pages}
  {192--204}\BibitemShut {NoStop}%
\bibitem [{\citenamefont {Chaudhari}(2018)}]{chaudhari2017energy}%
  \BibitemOpen
  \bibfield  {author} {\bibinfo {author} {\bibfnamefont {P.~A.}\ \bibnamefont
  {Chaudhari}},\ }\emph {\bibinfo {title} {A Picture of the Energy Landscape of
  Deep Neural Networks}},\ \href {https://escholarship.org/uc/item/26h5787r}
  {Ph.D. thesis},\ \bibinfo  {school} {University of California, Los Angeles}
  (\bibinfo {year} {2018})\BibitemShut {NoStop}%
\bibitem [{\citenamefont {McClean}\ \emph {et~al.}(2018)\citenamefont
  {McClean}, \citenamefont {Boixo}, \citenamefont {Smelyanskiy}, \citenamefont
  {Babbush},\ and\ \citenamefont {Neven}}]{mcclean2018barren}%
  \BibitemOpen
  \bibfield  {author} {\bibinfo {author} {\bibfnamefont {J.~R.}\ \bibnamefont
  {McClean}}, \bibinfo {author} {\bibfnamefont {S.}~\bibnamefont {Boixo}},
  \bibinfo {author} {\bibfnamefont {V.~N.}\ \bibnamefont {Smelyanskiy}},
  \bibinfo {author} {\bibfnamefont {R.}~\bibnamefont {Babbush}}, \ and\
  \bibinfo {author} {\bibfnamefont {H.}~\bibnamefont {Neven}},\ }\bibfield
  {title} {\enquote {\bibinfo {title} {Barren plateaus in quantum neural
  network training landscapes},}\ }\href {\doibase 10.1038/s41467-018-07090-4}
  {\bibfield  {journal} {\bibinfo  {journal} {Nat. Commun.}\ }\textbf {\bibinfo
  {volume} {9}},\ \bibinfo {pages} {4812} (\bibinfo {year} {2018})}\BibitemShut
  {NoStop}%
\bibitem [{\citenamefont {Cerezo}\ \emph {et~al.}(2021)\citenamefont {Cerezo},
  \citenamefont {Sone}, \citenamefont {Volkoff}, \citenamefont {Cincio},\ and\
  \citenamefont {Coles}}]{cerezo2021cost}%
  \BibitemOpen
  \bibfield  {author} {\bibinfo {author} {\bibfnamefont {M.}~\bibnamefont
  {Cerezo}}, \bibinfo {author} {\bibfnamefont {A.}~\bibnamefont {Sone}},
  \bibinfo {author} {\bibfnamefont {T.}~\bibnamefont {Volkoff}}, \bibinfo
  {author} {\bibfnamefont {L.}~\bibnamefont {Cincio}}, \ and\ \bibinfo {author}
  {\bibfnamefont {P.~J.}\ \bibnamefont {Coles}},\ }\bibfield  {title} {\enquote
  {\bibinfo {title} {Cost function dependent barren plateaus in shallow
  parametrized quantum circuits},}\ }\href {\doibase
  10.1038/s41467-021-21728-w} {\bibfield  {journal} {\bibinfo  {journal} {Nat.
  Commun.}\ }\textbf {\bibinfo {volume} {12}},\ \bibinfo {pages} {1791}
  (\bibinfo {year} {2021})}\BibitemShut {NoStop}%
\bibitem [{\citenamefont {Napp}(2022)}]{napp2022quantifying}%
  \BibitemOpen
  \bibfield  {author} {\bibinfo {author} {\bibfnamefont {J.}~\bibnamefont
  {Napp}},\ }\href@noop {} {\enquote {\bibinfo {title} {Quantifying the barren
  plateau phenomenon for a model of unstructured variational ans\"{a}tze},}\ }
  (\bibinfo {year} {2022}),\ \Eprint
  {http://arxiv.org/abs/2203.06174}{arXiv:2203.06174 [quant-ph]}\BibitemShut
  {NoStop}%
\bibitem [{\citenamefont {Cerezo}\ and\ \citenamefont
  {Coles}(2021)}]{cerezo2020impact}%
  \BibitemOpen
  \bibfield  {author} {\bibinfo {author} {\bibfnamefont {M.}~\bibnamefont
  {Cerezo}}\ and\ \bibinfo {author} {\bibfnamefont {P.~J.}\ \bibnamefont
  {Coles}},\ }\bibfield  {title} {\enquote {\bibinfo {title} {Higher order
  derivatives of quantum neural networks with barren plateaus},}\ }\href
  {\doibase 10.1088/2058-9565/abf51a} {\bibfield  {journal} {\bibinfo
  {journal} {Quantum Sci. Technol.}\ }\textbf {\bibinfo {volume} {6}},\
  \bibinfo {pages} {035006} (\bibinfo {year} {2021})}\BibitemShut {NoStop}%
\bibitem [{\citenamefont {Arrasmith}\ \emph {et~al.}(2021)\citenamefont
  {Arrasmith}, \citenamefont {Cerezo}, \citenamefont {Czarnik}, \citenamefont
  {Cincio},\ and\ \citenamefont {Coles}}]{arrasmith2021effect}%
  \BibitemOpen
  \bibfield  {author} {\bibinfo {author} {\bibfnamefont {A.}~\bibnamefont
  {Arrasmith}}, \bibinfo {author} {\bibfnamefont {M.}~\bibnamefont {Cerezo}},
  \bibinfo {author} {\bibfnamefont {P.}~\bibnamefont {Czarnik}}, \bibinfo
  {author} {\bibfnamefont {L.}~\bibnamefont {Cincio}}, \ and\ \bibinfo {author}
  {\bibfnamefont {P.~J.}\ \bibnamefont {Coles}},\ }\bibfield  {title} {\enquote
  {\bibinfo {title} {Effect of barren plateaus on gradient-free
  optimization},}\ }\href {\doibase 10.22331/q-2021-10-05-558} {\bibfield
  {journal} {\bibinfo  {journal} {{Quantum}}\ }\textbf {\bibinfo {volume}
  {5}},\ \bibinfo {pages} {558} (\bibinfo {year} {2021})}\BibitemShut {NoStop}%
\bibitem [{\citenamefont {Anschuetz}(2022)}]{anschuetz2021critical}%
  \BibitemOpen
  \bibfield  {author} {\bibinfo {author} {\bibfnamefont {E.~R.}\ \bibnamefont
  {Anschuetz}},\ }\bibfield  {title} {\enquote {\bibinfo {title} {Critical
  points in quantum generative models},}\ }in\ \href
  {https://openreview.net/forum?id=2f1z55GVQN} {\emph {\bibinfo {booktitle}
  {{I}nternational {C}onference on {L}earning {R}epresentations}}},\ \bibinfo
  {editor} {edited by\ \bibinfo {editor} {\bibfnamefont {K.}~\bibnamefont
  {Hofmann}}, \bibinfo {editor} {\bibfnamefont {A.}~\bibnamefont {Rush}},
  \bibinfo {editor} {\bibfnamefont {Y.}~\bibnamefont {Liu}}, \bibinfo {editor}
  {\bibfnamefont {C.}~\bibnamefont {Finn}}, \bibinfo {editor} {\bibfnamefont
  {Y.}~\bibnamefont {Choi}}, \ and\ \bibinfo {editor} {\bibfnamefont
  {M.}~\bibnamefont {Deisenroth}}}\ (\bibinfo  {publisher} {OpenReview},\
  \bibinfo {year} {2022})\BibitemShut {NoStop}%
\bibitem [{\citenamefont {Anschuetz}\ and\ \citenamefont
  {Kiani}(2022)}]{anschuetz2022barren}%
  \BibitemOpen
  \bibfield  {author} {\bibinfo {author} {\bibfnamefont {E.~R.}\ \bibnamefont
  {Anschuetz}}\ and\ \bibinfo {author} {\bibfnamefont {B.~T.}\ \bibnamefont
  {Kiani}},\ }\bibfield  {title} {\enquote {\bibinfo {title} {Quantum
  variational algorithms are swamped with traps},}\ }\href {\doibase
  10.1038/s41467-022-35364-5} {\bibfield  {journal} {\bibinfo  {journal} {Nat.
  Commun.}\ }\textbf {\bibinfo {volume} {13}},\ \bibinfo {pages} {7760}
  (\bibinfo {year} {2022})}\BibitemShut {NoStop}%
\bibitem [{\citenamefont {Anschuetz}(2024)}]{anschuetz2024unified}%
  \BibitemOpen
  \bibfield  {author} {\bibinfo {author} {\bibfnamefont {E.~R.}\ \bibnamefont
  {Anschuetz}},\ }\href@noop {} {\enquote {\bibinfo {title} {A unified theory
  of quantum neural network loss landscapes},}\ } (\bibinfo {year} {2024}),\
  \Eprint {http://arxiv.org/abs/2408.11901}{arXiv:2408.11901
  [quant-ph]}\BibitemShut {NoStop}%
\bibitem [{\citenamefont {Fontana}\ \emph {et~al.}(2023)\citenamefont
  {Fontana}, \citenamefont {Herman}, \citenamefont {Chakrabarti}, \citenamefont
  {Kumar}, \citenamefont {Yalovetzky}, \citenamefont {Heredge}, \citenamefont
  {Sureshbabu},\ and\ \citenamefont {Pistoia}}]{fontana2023adjoint}%
  \BibitemOpen
  \bibfield  {author} {\bibinfo {author} {\bibfnamefont {E.}~\bibnamefont
  {Fontana}}, \bibinfo {author} {\bibfnamefont {D.}~\bibnamefont {Herman}},
  \bibinfo {author} {\bibfnamefont {S.}~\bibnamefont {Chakrabarti}}, \bibinfo
  {author} {\bibfnamefont {N.}~\bibnamefont {Kumar}}, \bibinfo {author}
  {\bibfnamefont {R.}~\bibnamefont {Yalovetzky}}, \bibinfo {author}
  {\bibfnamefont {J.}~\bibnamefont {Heredge}}, \bibinfo {author} {\bibfnamefont
  {S.~H.}\ \bibnamefont {Sureshbabu}}, \ and\ \bibinfo {author} {\bibfnamefont
  {M.}~\bibnamefont {Pistoia}},\ }\href@noop {} {\enquote {\bibinfo {title}
  {The adjoint is all you need: Characterizing barren plateaus in quantum
  ans\"atze},}\ } (\bibinfo {year} {2023}),\ \Eprint
  {http://arxiv.org/abs/2309.07902}{arXiv:2309.07902 [quant-ph]}\BibitemShut
  {NoStop}%
\bibitem [{\citenamefont {Ragone}\ \emph {et~al.}(2023)\citenamefont {Ragone},
  \citenamefont {Bakalov}, \citenamefont {Sauvage}, \citenamefont {Kemper},
  \citenamefont {Marrero}, \citenamefont {Larocca},\ and\ \citenamefont
  {Cerezo}}]{ragone2023unified}%
  \BibitemOpen
  \bibfield  {author} {\bibinfo {author} {\bibfnamefont {M.}~\bibnamefont
  {Ragone}}, \bibinfo {author} {\bibfnamefont {B.~N.}\ \bibnamefont {Bakalov}},
  \bibinfo {author} {\bibfnamefont {F.}~\bibnamefont {Sauvage}}, \bibinfo
  {author} {\bibfnamefont {A.~F.}\ \bibnamefont {Kemper}}, \bibinfo {author}
  {\bibfnamefont {C.~O.}\ \bibnamefont {Marrero}}, \bibinfo {author}
  {\bibfnamefont {M.}~\bibnamefont {Larocca}}, \ and\ \bibinfo {author}
  {\bibfnamefont {M.}~\bibnamefont {Cerezo}},\ }\href@noop {} {\enquote
  {\bibinfo {title} {A unified theory of barren plateaus for deep parametrized
  quantum circuits},}\ } (\bibinfo {year} {2023}),\ \Eprint
  {http://arxiv.org/abs/2309.09342}{arXiv:2309.09342 [quant-ph]}\BibitemShut
  {NoStop}%
\bibitem [{\citenamefont {Anschuetz}\ \emph
  {et~al.}(2023{\natexlab{a}})\citenamefont {Anschuetz}, \citenamefont {Bauer},
  \citenamefont {Kiani},\ and\ \citenamefont {Lloyd}}]{anschuetz2022efficient}%
  \BibitemOpen
  \bibfield  {author} {\bibinfo {author} {\bibfnamefont {E.~R.}\ \bibnamefont
  {Anschuetz}}, \bibinfo {author} {\bibfnamefont {A.}~\bibnamefont {Bauer}},
  \bibinfo {author} {\bibfnamefont {B.~T.}\ \bibnamefont {Kiani}}, \ and\
  \bibinfo {author} {\bibfnamefont {S.}~\bibnamefont {Lloyd}},\ }\bibfield
  {title} {\enquote {\bibinfo {title} {Efficient classical algorithms for
  simulating symmetric quantum systems},}\ }\href {\doibase
  10.22331/q-2023-11-28-1189} {\bibfield  {journal} {\bibinfo  {journal}
  {{Quantum}}\ }\textbf {\bibinfo {volume} {7}},\ \bibinfo {pages} {1189}
  (\bibinfo {year} {2023}{\natexlab{a}})}\BibitemShut {NoStop}%
\bibitem [{\citenamefont {Goh}\ \emph {et~al.}(2025)\citenamefont {Goh},
  \citenamefont {Larocca}, \citenamefont {Cincio}, \citenamefont {Cerezo},\
  and\ \citenamefont {Sauvage}}]{Goh_2025}%
  \BibitemOpen
  \bibfield  {author} {\bibinfo {author} {\bibfnamefont {M.~L.}\ \bibnamefont
  {Goh}}, \bibinfo {author} {\bibfnamefont {M.}~\bibnamefont {Larocca}},
  \bibinfo {author} {\bibfnamefont {L.}~\bibnamefont {Cincio}}, \bibinfo
  {author} {\bibfnamefont {M.}~\bibnamefont {Cerezo}}, \ and\ \bibinfo {author}
  {\bibfnamefont {F.}~\bibnamefont {Sauvage}},\ }\bibfield  {title} {\enquote
  {\bibinfo {title} {{Lie}-algebraic classical simulations for quantum
  computing},}\ }\href {\doibase 10.1103/3y65-f5w6} {\bibfield  {journal}
  {\bibinfo  {journal} {Phys. Rev. Res.}\ }\textbf {\bibinfo {volume} {7}},\
  \bibinfo {pages} {033266} (\bibinfo {year} {2025})}\BibitemShut {NoStop}%
\bibitem [{\citenamefont {Cerezo}\ \emph {et~al.}(2025)\citenamefont {Cerezo},
  \citenamefont {Larocca}, \citenamefont {García-Martín}, \citenamefont
  {Diaz}, \citenamefont {Braccia}, \citenamefont {Fontana}, \citenamefont
  {Rudolph}, \citenamefont {Bermejo}, \citenamefont {Ijaz}, \citenamefont
  {Thanasilp}, \citenamefont {Anschuetz},\ and\ \citenamefont
  {Holmes}}]{cerezo2023does}%
  \BibitemOpen
  \bibfield  {author} {\bibinfo {author} {\bibfnamefont {M.}~\bibnamefont
  {Cerezo}}, \bibinfo {author} {\bibfnamefont {M.}~\bibnamefont {Larocca}},
  \bibinfo {author} {\bibfnamefont {D.}~\bibnamefont {García-Martín}},
  \bibinfo {author} {\bibfnamefont {N.~L.}\ \bibnamefont {Diaz}}, \bibinfo
  {author} {\bibfnamefont {P.}~\bibnamefont {Braccia}}, \bibinfo {author}
  {\bibfnamefont {E.}~\bibnamefont {Fontana}}, \bibinfo {author} {\bibfnamefont
  {M.~S.}\ \bibnamefont {Rudolph}}, \bibinfo {author} {\bibfnamefont
  {P.}~\bibnamefont {Bermejo}}, \bibinfo {author} {\bibfnamefont
  {A.}~\bibnamefont {Ijaz}}, \bibinfo {author} {\bibfnamefont {S.}~\bibnamefont
  {Thanasilp}}, \bibinfo {author} {\bibfnamefont {E.~R.}\ \bibnamefont
  {Anschuetz}}, \ and\ \bibinfo {author} {\bibfnamefont {Z.}~\bibnamefont
  {Holmes}},\ }\bibfield  {title} {\enquote {\bibinfo {title} {Does provable
  absence of barren plateaus imply classical simulability?}}\ }\href {\doibase
  10.1038/s41467-025-63099-6} {\bibfield  {journal} {\bibinfo  {journal} {Nat.
  Commun.}\ }\textbf {\bibinfo {volume} {16}},\ \bibinfo {pages} {7907}
  (\bibinfo {year} {2025})}\BibitemShut {NoStop}%
\bibitem [{\citenamefont {Gao}\ \emph {et~al.}(2022)\citenamefont {Gao},
  \citenamefont {Anschuetz}, \citenamefont {Wang}, \citenamefont {Cirac},\ and\
  \citenamefont {Lukin}}]{gao2021enhancing}%
  \BibitemOpen
  \bibfield  {author} {\bibinfo {author} {\bibfnamefont {X.}~\bibnamefont
  {Gao}}, \bibinfo {author} {\bibfnamefont {E.~R.}\ \bibnamefont {Anschuetz}},
  \bibinfo {author} {\bibfnamefont {S.-T.}\ \bibnamefont {Wang}}, \bibinfo
  {author} {\bibfnamefont {J.~I.}\ \bibnamefont {Cirac}}, \ and\ \bibinfo
  {author} {\bibfnamefont {M.~D.}\ \bibnamefont {Lukin}},\ }\bibfield  {title}
  {\enquote {\bibinfo {title} {Enhancing generative models via quantum
  correlations},}\ }\href {\doibase 10.1103/PhysRevX.12.021037} {\bibfield
  {journal} {\bibinfo  {journal} {Phys. Rev. X}\ }\textbf {\bibinfo {volume}
  {12}},\ \bibinfo {pages} {021037} (\bibinfo {year} {2022})}\BibitemShut
  {NoStop}%
\bibitem [{\citenamefont {Zhao}\ and\ \citenamefont
  {Deng}(2024)}]{zhao2024entanglement}%
  \BibitemOpen
  \bibfield  {author} {\bibinfo {author} {\bibfnamefont {H.}~\bibnamefont
  {Zhao}}\ and\ \bibinfo {author} {\bibfnamefont {D.-L.}\ \bibnamefont
  {Deng}},\ }\href {https://arxiv.org/abs/2410.03094} {\enquote {\bibinfo
  {title} {Entanglement-induced provable and robust quantum learning
  advantages},}\ } (\bibinfo {year} {2024}),\ \Eprint
  {http://arxiv.org/abs/2410.03094}{arXiv:2410.03094 [quant-ph]}\BibitemShut
  {NoStop}%
\bibitem [{\citenamefont {Anschuetz}\ \emph
  {et~al.}(2023{\natexlab{b}})\citenamefont {Anschuetz}, \citenamefont {Hu},
  \citenamefont {Huang},\ and\ \citenamefont {Gao}}]{anschuetzgao2022}%
  \BibitemOpen
  \bibfield  {author} {\bibinfo {author} {\bibfnamefont {E.~R.}\ \bibnamefont
  {Anschuetz}}, \bibinfo {author} {\bibfnamefont {H.-Y.}\ \bibnamefont {Hu}},
  \bibinfo {author} {\bibfnamefont {J.-L.}\ \bibnamefont {Huang}}, \ and\
  \bibinfo {author} {\bibfnamefont {X.}~\bibnamefont {Gao}},\ }\bibfield
  {title} {\enquote {\bibinfo {title} {Interpretable quantum advantage in
  neural sequence learning},}\ }\href {\doibase 10.1103/PRXQuantum.4.020338}
  {\bibfield  {journal} {\bibinfo  {journal} {PRX Quantum}\ }\textbf {\bibinfo
  {volume} {4}},\ \bibinfo {pages} {020338} (\bibinfo {year}
  {2023}{\natexlab{b}})}\BibitemShut {NoStop}%
\bibitem [{\citenamefont {Babbush}\ \emph {et~al.}(2021)\citenamefont
  {Babbush}, \citenamefont {McClean}, \citenamefont {Newman}, \citenamefont
  {Gidney}, \citenamefont {Boixo},\ and\ \citenamefont
  {Neven}}]{PRXQuantum.2.010103}%
  \BibitemOpen
  \bibfield  {author} {\bibinfo {author} {\bibfnamefont {R.}~\bibnamefont
  {Babbush}}, \bibinfo {author} {\bibfnamefont {J.~R.}\ \bibnamefont
  {McClean}}, \bibinfo {author} {\bibfnamefont {M.}~\bibnamefont {Newman}},
  \bibinfo {author} {\bibfnamefont {C.}~\bibnamefont {Gidney}}, \bibinfo
  {author} {\bibfnamefont {S.}~\bibnamefont {Boixo}}, \ and\ \bibinfo {author}
  {\bibfnamefont {H.}~\bibnamefont {Neven}},\ }\bibfield  {title} {\enquote
  {\bibinfo {title} {Focus beyond quadratic speedups for error-corrected
  quantum advantage},}\ }\href {\doibase 10.1103/PRXQuantum.2.010103}
  {\bibfield  {journal} {\bibinfo  {journal} {PRX Quantum}\ }\textbf {\bibinfo
  {volume} {2}},\ \bibinfo {pages} {010103} (\bibinfo {year}
  {2021})}\BibitemShut {NoStop}%
\bibitem [{\citenamefont {Vaswani}\ \emph {et~al.}(2017)\citenamefont
  {Vaswani}, \citenamefont {Shazeer}, \citenamefont {Parmar}, \citenamefont
  {Uszkoreit}, \citenamefont {Jones}, \citenamefont {Gomez}, \citenamefont
  {Kaiser},\ and\ \citenamefont {Polosukhin}}]{10.5555/3295222.3295349}%
  \BibitemOpen
  \bibfield  {author} {\bibinfo {author} {\bibfnamefont {A.}~\bibnamefont
  {Vaswani}}, \bibinfo {author} {\bibfnamefont {N.}~\bibnamefont {Shazeer}},
  \bibinfo {author} {\bibfnamefont {N.}~\bibnamefont {Parmar}}, \bibinfo
  {author} {\bibfnamefont {J.}~\bibnamefont {Uszkoreit}}, \bibinfo {author}
  {\bibfnamefont {L.}~\bibnamefont {Jones}}, \bibinfo {author} {\bibfnamefont
  {A.~N.}\ \bibnamefont {Gomez}}, \bibinfo {author} {\bibfnamefont
  {L.}~\bibnamefont {Kaiser}}, \ and\ \bibinfo {author} {\bibfnamefont
  {I.}~\bibnamefont {Polosukhin}},\ }\bibfield  {title} {\enquote {\bibinfo
  {title} {Attention is all you need},}\ }in\ \href
  {https://papers.nips.cc/paper_files/paper/2017/hash/3f5ee243547dee91fbd053c1c4a845aa-Abstract.html}
  {\emph {\bibinfo {booktitle} {Proceedings of the 31st International
  Conference on Neural Information Processing Systems}}},\ \bibinfo {series and
  number} {NIPS'17},\ \bibinfo {editor} {edited by\ \bibinfo {editor}
  {\bibfnamefont {I.}~\bibnamefont {Guyon}}, \bibinfo {editor} {\bibfnamefont
  {U.~V.}\ \bibnamefont {Luxburg}}, \bibinfo {editor} {\bibfnamefont
  {S.}~\bibnamefont {Bengio}}, \bibinfo {editor} {\bibfnamefont
  {H.}~\bibnamefont {Wallach}}, \bibinfo {editor} {\bibfnamefont
  {R.}~\bibnamefont {Fergus}}, \bibinfo {editor} {\bibfnamefont
  {S.}~\bibnamefont {Vishwanathan}}, \ and\ \bibinfo {editor} {\bibfnamefont
  {R.}~\bibnamefont {Garnett}}}\ (\bibinfo  {publisher} {Curran Associates,
  Inc.},\ \bibinfo {address} {Red Hook, NY, USA},\ \bibinfo {year} {2017})\
  pp.\ \bibinfo {pages} {6000--6010}\BibitemShut {NoStop}%
\bibitem [{\citenamefont {Takeuchi}\ \emph {et~al.}(2019)\citenamefont
  {Takeuchi}, \citenamefont {Mantri}, \citenamefont {Morimae}, \citenamefont
  {Mizutani},\ and\ \citenamefont {Fitzsimons}}]{takeuchi2019resource}%
  \BibitemOpen
  \bibfield  {author} {\bibinfo {author} {\bibfnamefont {Y.}~\bibnamefont
  {Takeuchi}}, \bibinfo {author} {\bibfnamefont {A.}~\bibnamefont {Mantri}},
  \bibinfo {author} {\bibfnamefont {T.}~\bibnamefont {Morimae}}, \bibinfo
  {author} {\bibfnamefont {A.}~\bibnamefont {Mizutani}}, \ and\ \bibinfo
  {author} {\bibfnamefont {J.~F.}\ \bibnamefont {Fitzsimons}},\ }\bibfield
  {title} {\enquote {\bibinfo {title} {Resource-efficient verification of
  quantum computing using serfling’s bound},}\ }\href {\doibase
  10.1038/s41534-019-0142-2} {\bibfield  {journal} {\bibinfo  {journal} {npj
  Quantum Inf.}\ }\textbf {\bibinfo {volume} {5}},\ \bibinfo {pages} {27}
  (\bibinfo {year} {2019})}\BibitemShut {NoStop}%
\bibitem [{\citenamefont {Webster}\ \emph {et~al.}(2022)\citenamefont
  {Webster}, \citenamefont {Brown},\ and\ \citenamefont
  {Bartlett}}]{Webster2022xpstabiliser}%
  \BibitemOpen
  \bibfield  {author} {\bibinfo {author} {\bibfnamefont {M.~A.}\ \bibnamefont
  {Webster}}, \bibinfo {author} {\bibfnamefont {B.~J.}\ \bibnamefont {Brown}},
  \ and\ \bibinfo {author} {\bibfnamefont {S.~D.}\ \bibnamefont {Bartlett}},\
  }\bibfield  {title} {\enquote {\bibinfo {title} {The {XP} stabiliser
  formalism: a generalisation of the {P}auli stabiliser formalism with
  arbitrary phases},}\ }\href {\doibase 10.22331/q-2022-09-22-815} {\bibfield
  {journal} {\bibinfo  {journal} {Quantum}\ }\textbf {\bibinfo {volume} {6}},\
  \bibinfo {pages} {815} (\bibinfo {year} {2022})}\BibitemShut {NoStop}%
\bibitem [{\citenamefont {Somma}\ \emph {et~al.}(2006)\citenamefont {Somma},
  \citenamefont {Barnum}, \citenamefont {Ortiz},\ and\ \citenamefont
  {Knill}}]{PhysRevLett.97.190501}%
  \BibitemOpen
  \bibfield  {author} {\bibinfo {author} {\bibfnamefont {R.}~\bibnamefont
  {Somma}}, \bibinfo {author} {\bibfnamefont {H.}~\bibnamefont {Barnum}},
  \bibinfo {author} {\bibfnamefont {G.}~\bibnamefont {Ortiz}}, \ and\ \bibinfo
  {author} {\bibfnamefont {E.}~\bibnamefont {Knill}},\ }\bibfield  {title}
  {\enquote {\bibinfo {title} {Efficient solvability of {Hamiltonians} and
  limits on the power of some quantum computational models},}\ }\href {\doibase
  10.1103/PhysRevLett.97.190501} {\bibfield  {journal} {\bibinfo  {journal}
  {Phys. Rev. Lett.}\ }\textbf {\bibinfo {volume} {97}},\ \bibinfo {pages}
  {190501} (\bibinfo {year} {2006})}\BibitemShut {NoStop}%
\bibitem [{\citenamefont {Mints}\ and\ \citenamefont
  {Anschuetz}(2025)}]{mints2025fragmentationefficientlylearnablequantum}%
  \BibitemOpen
  \bibfield  {author} {\bibinfo {author} {\bibfnamefont {M.}~\bibnamefont
  {Mints}}\ and\ \bibinfo {author} {\bibfnamefont {E.~R.}\ \bibnamefont
  {Anschuetz}},\ }\href@noop {} {\enquote {\bibinfo {title} {Fragmentation is
  efficiently learnable by quantum neural networks},}\ } (\bibinfo {year}
  {2025}),\ \Eprint {http://arxiv.org/abs/2512.00751}{arXiv:2512.00751
  [quant-ph]}\BibitemShut {NoStop}%
\bibitem [{\citenamefont {Wu}\ \emph {et~al.}(2022)\citenamefont {Wu},
  \citenamefont {Raghavendra}, \citenamefont {Gupta}, \citenamefont {Acun},
  \citenamefont {Ardalani}, \citenamefont {Maeng}, \citenamefont {Chang},
  \citenamefont {Aga}, \citenamefont {Huang}, \citenamefont {Bai},
  \citenamefont {Gschwind}, \citenamefont {Gupta}, \citenamefont {Ott},
  \citenamefont {Melnikov}, \citenamefont {Candido}, \citenamefont {Brooks},
  \citenamefont {Chauhan}, \citenamefont {Lee}, \citenamefont {Lee},
  \citenamefont {Akyildiz}, \citenamefont {Balandat}, \citenamefont {Spisak},
  \citenamefont {Jain}, \citenamefont {Rabbat},\ and\ \citenamefont
  {Hazelwood}}]{MLSYS2022_462211f6}%
  \BibitemOpen
  \bibfield  {author} {\bibinfo {author} {\bibfnamefont {C.-J.}\ \bibnamefont
  {Wu}}, \bibinfo {author} {\bibfnamefont {R.}~\bibnamefont {Raghavendra}},
  \bibinfo {author} {\bibfnamefont {U.}~\bibnamefont {Gupta}}, \bibinfo
  {author} {\bibfnamefont {B.}~\bibnamefont {Acun}}, \bibinfo {author}
  {\bibfnamefont {N.}~\bibnamefont {Ardalani}}, \bibinfo {author}
  {\bibfnamefont {K.}~\bibnamefont {Maeng}}, \bibinfo {author} {\bibfnamefont
  {G.}~\bibnamefont {Chang}}, \bibinfo {author} {\bibfnamefont
  {F.}~\bibnamefont {Aga}}, \bibinfo {author} {\bibfnamefont {J.}~\bibnamefont
  {Huang}}, \bibinfo {author} {\bibfnamefont {C.}~\bibnamefont {Bai}}, \bibinfo
  {author} {\bibfnamefont {M.}~\bibnamefont {Gschwind}}, \bibinfo {author}
  {\bibfnamefont {A.}~\bibnamefont {Gupta}}, \bibinfo {author} {\bibfnamefont
  {M.}~\bibnamefont {Ott}}, \bibinfo {author} {\bibfnamefont {A.}~\bibnamefont
  {Melnikov}}, \bibinfo {author} {\bibfnamefont {S.}~\bibnamefont {Candido}},
  \bibinfo {author} {\bibfnamefont {D.}~\bibnamefont {Brooks}}, \bibinfo
  {author} {\bibfnamefont {G.}~\bibnamefont {Chauhan}}, \bibinfo {author}
  {\bibfnamefont {B.}~\bibnamefont {Lee}}, \bibinfo {author} {\bibfnamefont
  {H.-H.}\ \bibnamefont {Lee}}, \bibinfo {author} {\bibfnamefont
  {B.}~\bibnamefont {Akyildiz}}, \bibinfo {author} {\bibfnamefont
  {M.}~\bibnamefont {Balandat}}, \bibinfo {author} {\bibfnamefont
  {J.}~\bibnamefont {Spisak}}, \bibinfo {author} {\bibfnamefont
  {R.}~\bibnamefont {Jain}}, \bibinfo {author} {\bibfnamefont {M.}~\bibnamefont
  {Rabbat}}, \ and\ \bibinfo {author} {\bibfnamefont {K.}~\bibnamefont
  {Hazelwood}},\ }\bibfield  {title} {\enquote {\bibinfo {title} {Sustainable
  {AI}: Environmental implications, challenges and opportunities},}\ }in\ \href
  {https://proceedings.mlsys.org/paper_files/paper/2022/hash/462211f67c7d858f663355eff93b745e-Abstract.html}
  {\emph {\bibinfo {booktitle} {Proceedings of Machine Learning and
  Systems}}},\ Vol.~\bibinfo {volume} {4},\ \bibinfo {editor} {edited by\
  \bibinfo {editor} {\bibfnamefont {D.}~\bibnamefont {Marculescu}}, \bibinfo
  {editor} {\bibfnamefont {Y.}~\bibnamefont {Chi}}, \ and\ \bibinfo {editor}
  {\bibfnamefont {C.}~\bibnamefont {Wu}}}\ (\bibinfo {year} {2022})\ pp.\
  \bibinfo {pages} {795--813}\BibitemShut {NoStop}%
\bibitem [{\citenamefont {Mavromatis}\ \emph {et~al.}(2024)\citenamefont
  {Mavromatis}, \citenamefont {Katsaros},\ and\ \citenamefont
  {Khan}}]{mavromatis2024computinglimitsempiricalstudy}%
  \BibitemOpen
  \bibfield  {author} {\bibinfo {author} {\bibfnamefont {I.}~\bibnamefont
  {Mavromatis}}, \bibinfo {author} {\bibfnamefont {K.}~\bibnamefont
  {Katsaros}}, \ and\ \bibinfo {author} {\bibfnamefont {A.}~\bibnamefont
  {Khan}},\ }\href@noop {} {\enquote {\bibinfo {title} {Computing within
  limits: An empirical study of energy consumption in {ML} training and
  inference},}\ } (\bibinfo {year} {2024}),\ \Eprint
  {http://arxiv.org/abs/2406.14328}{arXiv:2406.14328 [cs.LG]}\BibitemShut
  {NoStop}%
\bibitem [{\citenamefont {Kochen}\ and\ \citenamefont
  {Specker}(1967)}]{kochen1975problem}%
  \BibitemOpen
  \bibfield  {author} {\bibinfo {author} {\bibfnamefont {S.}~\bibnamefont
  {Kochen}}\ and\ \bibinfo {author} {\bibfnamefont {E.~P.}\ \bibnamefont
  {Specker}},\ }\bibfield  {title} {\enquote {\bibinfo {title} {The problem of
  hidden variables in quantum mechanics},}\ }\href
  {http://www.jstor.org/stable/24902153} {\bibfield  {journal} {\bibinfo
  {journal} {J. Math. Mech.}\ }\textbf {\bibinfo {volume} {17}},\ \bibinfo
  {pages} {59} (\bibinfo {year} {1967})}\BibitemShut {NoStop}%
\bibitem [{\citenamefont {Abramsky}\ and\ \citenamefont
  {Brandenburger}(2011)}]{Abramsky_2011}%
  \BibitemOpen
  \bibfield  {author} {\bibinfo {author} {\bibfnamefont {S.}~\bibnamefont
  {Abramsky}}\ and\ \bibinfo {author} {\bibfnamefont {A.}~\bibnamefont
  {Brandenburger}},\ }\bibfield  {title} {\enquote {\bibinfo {title} {The
  sheaf-theoretic structure of non-locality and contextuality},}\ }\href
  {\doibase 10.1088/1367-2630/13/11/113036} {\bibfield  {journal} {\bibinfo
  {journal} {New J. Phys.}\ }\textbf {\bibinfo {volume} {13}},\ \bibinfo
  {pages} {113036} (\bibinfo {year} {2011})}\BibitemShut {NoStop}%
\bibitem [{\citenamefont {Abramsky}\ and\ \citenamefont
  {Sadrzadeh}(2014)}]{abramsky2014}%
  \BibitemOpen
  \bibfield  {author} {\bibinfo {author} {\bibfnamefont {S.}~\bibnamefont
  {Abramsky}}\ and\ \bibinfo {author} {\bibfnamefont {M.}~\bibnamefont
  {Sadrzadeh}},\ }\enquote {\bibinfo {title} {Semantic unification},}\ in\
  \href {\doibase 10.1007/978-3-642-54789-8_1} {\emph {\bibinfo {booktitle}
  {Categories and Types in Logic, Language, and Physics: Essays Dedicated to
  {Jim} {Lambek} on the Occasion of His 90th Birthday}}},\ \bibinfo {editor}
  {edited by\ \bibinfo {editor} {\bibfnamefont {C.}~\bibnamefont {Casadio}},
  \bibinfo {editor} {\bibfnamefont {B.}~\bibnamefont {Coecke}}, \bibinfo
  {editor} {\bibfnamefont {M.}~\bibnamefont {Moortgat}}, \ and\ \bibinfo
  {editor} {\bibfnamefont {P.}~\bibnamefont {Scott}}}\ (\bibinfo  {publisher}
  {Springer Berlin Heidelberg},\ \bibinfo {address} {Berlin, Heidelberg,
  Germany},\ \bibinfo {year} {2014})\ pp.\ \bibinfo {pages} {1--13}\BibitemShut
  {NoStop}%
\bibitem [{\citenamefont {Abramsky}(2016)}]{doi:10.1142/9789814730617_0002}%
  \BibitemOpen
  \bibfield  {author} {\bibinfo {author} {\bibfnamefont {S.}~\bibnamefont
  {Abramsky}},\ }\enquote {\bibinfo {title} {Contextual semantics: From quantum
  mechanics to logic, databases, constraints, and complexity},}\ in\ \href
  {\doibase 10.1142/9789814730617_0002} {\emph {\bibinfo {booktitle}
  {Contextuality from Quantum Physics to Psychology}}}\ (\bibinfo  {publisher}
  {World Scientific Publishing},\ \bibinfo {year} {2016})\ Chap.\ \bibinfo
  {chapter} {Chapter 2}, pp.\ \bibinfo {pages} {23--50}\BibitemShut {NoStop}%
\bibitem [{\citenamefont {Abramsky}\ \emph {et~al.}(2015)\citenamefont
  {Abramsky}, \citenamefont {Soares~Barbosa}, \citenamefont {Kishida},
  \citenamefont {Lal},\ and\ \citenamefont {Mansfield}}]{abramsky_et_al}%
  \BibitemOpen
  \bibfield  {author} {\bibinfo {author} {\bibfnamefont {S.}~\bibnamefont
  {Abramsky}}, \bibinfo {author} {\bibfnamefont {R.}~\bibnamefont
  {Soares~Barbosa}}, \bibinfo {author} {\bibfnamefont {K.}~\bibnamefont
  {Kishida}}, \bibinfo {author} {\bibfnamefont {R.}~\bibnamefont {Lal}}, \ and\
  \bibinfo {author} {\bibfnamefont {S.}~\bibnamefont {Mansfield}},\ }\bibfield
  {title} {\enquote {\bibinfo {title} {Contextuality, cohomology and
  paradox},}\ }in\ \href {\doibase 10.4230/LIPIcs.CSL.2015.211} {\emph
  {\bibinfo {booktitle} {24th EACSL Annual Conference on Computer Science Logic
  ({CSL} 2015)}}},\ \bibinfo {series} {{Leibniz} International Proceedings in
  Informatics ({LIPIcs})}, Vol.~\bibinfo {volume} {41},\ \bibinfo {editor}
  {edited by\ \bibinfo {editor} {\bibfnamefont {S.}~\bibnamefont {Kreutzer}}}\
  (\bibinfo  {publisher} {Schloss Dagstuhl -- Leibniz-Zentrum f{\"u}r
  Informatik},\ \bibinfo {address} {Dagstuhl, Germany},\ \bibinfo {year}
  {2015})\ pp.\ \bibinfo {pages} {211--228}\BibitemShut {NoStop}%
\bibitem [{\citenamefont {Cervantes}\ and\ \citenamefont
  {Dzhafarov}(2020)}]{e22090981}%
  \BibitemOpen
  \bibfield  {author} {\bibinfo {author} {\bibfnamefont {V.~H.}\ \bibnamefont
  {Cervantes}}\ and\ \bibinfo {author} {\bibfnamefont {E.~N.}\ \bibnamefont
  {Dzhafarov}},\ }\bibfield  {title} {\enquote {\bibinfo {title} {Contextuality
  analysis of impossible figures},}\ }\href {\doibase 10.3390/e22090981}
  {\bibfield  {journal} {\bibinfo  {journal} {Entropy}\ }\textbf {\bibinfo
  {volume} {22}},\ \bibinfo {pages} {981} (\bibinfo {year} {2020})}\BibitemShut
  {NoStop}%
\bibitem [{\citenamefont {Raz}(1999)}]{10.1145/301250.301343}%
  \BibitemOpen
  \bibfield  {author} {\bibinfo {author} {\bibfnamefont {R.}~\bibnamefont
  {Raz}},\ }\bibfield  {title} {\enquote {\bibinfo {title} {Exponential
  separation of quantum and classical communication complexity},}\ }in\ \href
  {\doibase 10.1145/301250.301343} {\emph {\bibinfo {booktitle} {Proceedings of
  the Thirty-First Annual ACM Symposium on Theory of Computing}}},\ \bibinfo
  {series and number} {STOC '99}\ (\bibinfo  {publisher} {Association for
  Computing Machinery},\ \bibinfo {address} {New York, NY, USA},\ \bibinfo
  {year} {1999})\ pp.\ \bibinfo {pages} {358--367}\BibitemShut {NoStop}%
\bibitem [{\citenamefont {Buhrman}\ \emph {et~al.}(2001)\citenamefont
  {Buhrman}, \citenamefont {Cleve}, \citenamefont {Watrous},\ and\
  \citenamefont {de~Wolf}}]{PhysRevLett.87.167902}%
  \BibitemOpen
  \bibfield  {author} {\bibinfo {author} {\bibfnamefont {H.}~\bibnamefont
  {Buhrman}}, \bibinfo {author} {\bibfnamefont {R.}~\bibnamefont {Cleve}},
  \bibinfo {author} {\bibfnamefont {J.}~\bibnamefont {Watrous}}, \ and\
  \bibinfo {author} {\bibfnamefont {R.}~\bibnamefont {de~Wolf}},\ }\bibfield
  {title} {\enquote {\bibinfo {title} {Quantum fingerprinting},}\ }\href
  {\doibase 10.1103/PhysRevLett.87.167902} {\bibfield  {journal} {\bibinfo
  {journal} {Phys. Rev. Lett.}\ }\textbf {\bibinfo {volume} {87}},\ \bibinfo
  {pages} {167902} (\bibinfo {year} {2001})}\BibitemShut {NoStop}%
\bibitem [{\citenamefont {Le~Gall}(2006)}]{10.1145/1148109.1148119}%
  \BibitemOpen
  \bibfield  {author} {\bibinfo {author} {\bibfnamefont {F.}~\bibnamefont
  {Le~Gall}},\ }\bibfield  {title} {\enquote {\bibinfo {title} {Exponential
  separation of quantum and classical online space complexity},}\ }in\ \href
  {\doibase 10.1145/1148109.1148119} {\emph {\bibinfo {booktitle} {Proceedings
  of the Eighteenth Annual ACM Symposium on Parallelism in Algorithms and
  Architectures}}},\ \bibinfo {series and number} {SPAA '06},\ \bibinfo
  {editor} {edited by\ \bibinfo {editor} {\bibfnamefont {P.~B.}\ \bibnamefont
  {Gibbons}}\ and\ \bibinfo {editor} {\bibfnamefont {U.}~\bibnamefont
  {Vishkin}}}\ (\bibinfo  {publisher} {Association for Computing Machinery},\
  \bibinfo {address} {New York, NY, USA},\ \bibinfo {year} {2006})\ pp.\
  \bibinfo {pages} {67--73}\BibitemShut {NoStop}%
\bibitem [{\citenamefont {Gavinsky}\ \emph {et~al.}(2007)\citenamefont
  {Gavinsky}, \citenamefont {Kempe}, \citenamefont {Kerenidis}, \citenamefont
  {Raz},\ and\ \citenamefont {de~Wolf}}]{10.1145/1250790.1250866}%
  \BibitemOpen
  \bibfield  {author} {\bibinfo {author} {\bibfnamefont {D.}~\bibnamefont
  {Gavinsky}}, \bibinfo {author} {\bibfnamefont {J.}~\bibnamefont {Kempe}},
  \bibinfo {author} {\bibfnamefont {I.}~\bibnamefont {Kerenidis}}, \bibinfo
  {author} {\bibfnamefont {R.}~\bibnamefont {Raz}}, \ and\ \bibinfo {author}
  {\bibfnamefont {R.}~\bibnamefont {de~Wolf}},\ }\bibfield  {title} {\enquote
  {\bibinfo {title} {Exponential separations for one-way quantum communication
  complexity, with applications to cryptography},}\ }in\ \href {\doibase
  10.1145/1250790.1250866} {\emph {\bibinfo {booktitle} {Proceedings of the
  Thirty-Ninth Annual ACM Symposium on Theory of Computing}}},\ \bibinfo
  {series and number} {STOC '07},\ \bibinfo {editor} {edited by\ \bibinfo
  {editor} {\bibfnamefont {D.}~\bibnamefont {Johnson}}\ and\ \bibinfo {editor}
  {\bibfnamefont {U.}~\bibnamefont {Feige}}}\ (\bibinfo  {publisher}
  {Association for Computing Machinery},\ \bibinfo {address} {New York, NY,
  USA},\ \bibinfo {year} {2007})\ pp.\ \bibinfo {pages} {516--525}\BibitemShut
  {NoStop}%
\bibitem [{\citenamefont {Gupta}\ \emph {et~al.}(2023)\citenamefont {Gupta},
  \citenamefont {Saha}, \citenamefont {Xu}, \citenamefont {Cabello},\ and\
  \citenamefont {Majumdar}}]{PhysRevLett.130.080802}%
  \BibitemOpen
  \bibfield  {author} {\bibinfo {author} {\bibfnamefont {S.}~\bibnamefont
  {Gupta}}, \bibinfo {author} {\bibfnamefont {D.}~\bibnamefont {Saha}},
  \bibinfo {author} {\bibfnamefont {Z.-P.}\ \bibnamefont {Xu}}, \bibinfo
  {author} {\bibfnamefont {A.}~\bibnamefont {Cabello}}, \ and\ \bibinfo
  {author} {\bibfnamefont {A.~S.}\ \bibnamefont {Majumdar}},\ }\bibfield
  {title} {\enquote {\bibinfo {title} {Quantum contextuality provides
  communication complexity advantage},}\ }\href {\doibase
  10.1103/PhysRevLett.130.080802} {\bibfield  {journal} {\bibinfo  {journal}
  {Phys. Rev. Lett.}\ }\textbf {\bibinfo {volume} {130}},\ \bibinfo {pages}
  {080802} (\bibinfo {year} {2023})}\BibitemShut {NoStop}%
\bibitem [{\citenamefont {Kallaugher}\ \emph {et~al.}(2023)\citenamefont
  {Kallaugher}, \citenamefont {Parekh},\ and\ \citenamefont
  {Voronova}}]{kallaugher2023exponential}%
  \BibitemOpen
  \bibfield  {author} {\bibinfo {author} {\bibfnamefont {J.}~\bibnamefont
  {Kallaugher}}, \bibinfo {author} {\bibfnamefont {O.}~\bibnamefont {Parekh}},
  \ and\ \bibinfo {author} {\bibfnamefont {N.}~\bibnamefont {Voronova}},\
  }\href@noop {} {\enquote {\bibinfo {title} {Exponential quantum space
  advantage for approximating maximum directed cut in the streaming model},}\ }
  (\bibinfo {year} {2023}),\ \Eprint
  {http://arxiv.org/abs/2311.14123}{arXiv:2311.14123 [quant-ph]}\BibitemShut
  {NoStop}%
\bibitem [{\citenamefont {Manna}\ \emph {et~al.}(2024)\citenamefont {Manna},
  \citenamefont {Chaturvedi},\ and\ \citenamefont {Saha}}]{manna2024unbounded}%
  \BibitemOpen
  \bibfield  {author} {\bibinfo {author} {\bibfnamefont {S.}~\bibnamefont
  {Manna}}, \bibinfo {author} {\bibfnamefont {A.}~\bibnamefont {Chaturvedi}}, \
  and\ \bibinfo {author} {\bibfnamefont {D.}~\bibnamefont {Saha}},\ }\href@noop
  {} {\enquote {\bibinfo {title} {Unbounded quantum advantage in communication
  complexity measured by distinguishability},}\ } (\bibinfo {year} {2024}),\
  \Eprint {http://arxiv.org/abs/2401.12903}{arXiv:2401.12903
  [quant-ph]}\BibitemShut {NoStop}%
\bibitem [{\citenamefont {Rudolph}\ \emph {et~al.}(2023)\citenamefont
  {Rudolph}, \citenamefont {Lerch}, \citenamefont {Thanasilp}, \citenamefont
  {Kiss}, \citenamefont {Vallecorsa}, \citenamefont {Grossi},\ and\
  \citenamefont {Holmes}}]{rudolph2023trainability}%
  \BibitemOpen
  \bibfield  {author} {\bibinfo {author} {\bibfnamefont {M.~S.}\ \bibnamefont
  {Rudolph}}, \bibinfo {author} {\bibfnamefont {S.}~\bibnamefont {Lerch}},
  \bibinfo {author} {\bibfnamefont {S.}~\bibnamefont {Thanasilp}}, \bibinfo
  {author} {\bibfnamefont {O.}~\bibnamefont {Kiss}}, \bibinfo {author}
  {\bibfnamefont {S.}~\bibnamefont {Vallecorsa}}, \bibinfo {author}
  {\bibfnamefont {M.}~\bibnamefont {Grossi}}, \ and\ \bibinfo {author}
  {\bibfnamefont {Z.}~\bibnamefont {Holmes}},\ }\href@noop {} {\enquote
  {\bibinfo {title} {Trainability barriers and opportunities in quantum
  generative modeling},}\ } (\bibinfo {year} {2023}),\ \Eprint
  {http://arxiv.org/abs/2305.02881}{arXiv:2305.02881 [quant-ph]}\BibitemShut
  {NoStop}%
\bibitem [{\citenamefont {Sutskever}\ \emph {et~al.}(2014)\citenamefont
  {Sutskever}, \citenamefont {Vinyals},\ and\ \citenamefont
  {Le}}]{10.5555/2969033.2969173}%
  \BibitemOpen
  \bibfield  {author} {\bibinfo {author} {\bibfnamefont {I.}~\bibnamefont
  {Sutskever}}, \bibinfo {author} {\bibfnamefont {O.}~\bibnamefont {Vinyals}},
  \ and\ \bibinfo {author} {\bibfnamefont {Q.~V.}\ \bibnamefont {Le}},\
  }\bibfield  {title} {\enquote {\bibinfo {title} {Sequence to sequence
  learning with neural networks},}\ }in\ \href
  {https://papers.nips.cc/paper_files/paper/2014/hash/a14ac55a4f27472c5d894ec1c3c743d2-Abstract.html}
  {\emph {\bibinfo {booktitle} {Proceedings of the 27th International
  Conference on Neural Information Processing Systems - Volume 2}}},\ \bibinfo
  {series and number} {NIPS'14},\ \bibinfo {editor} {edited by\ \bibinfo
  {editor} {\bibfnamefont {Z.}~\bibnamefont {Ghahramani}}, \bibinfo {editor}
  {\bibfnamefont {M.}~\bibnamefont {Welling}}, \bibinfo {editor} {\bibfnamefont
  {C.}~\bibnamefont {Cortes}}, \ and\ \bibinfo {editor} {\bibfnamefont
  {N.}~\bibnamefont {Lawrence}}}\ (\bibinfo  {publisher} {Curran Associates,
  Inc.},\ \bibinfo {address} {Red Hook, NY, USA},\ \bibinfo {year} {2014})\
  pp.\ \bibinfo {pages} {3104--3112}\BibitemShut {NoStop}%
\bibitem [{\citenamefont {Prabhavalkar}\ \emph {et~al.}(2017)\citenamefont
  {Prabhavalkar}, \citenamefont {Rao}, \citenamefont {Sainath}, \citenamefont
  {Li}, \citenamefont {Johnson},\ and\ \citenamefont
  {Jaitly}}]{Prabhavalkar2017}%
  \BibitemOpen
  \bibfield  {author} {\bibinfo {author} {\bibfnamefont {R.}~\bibnamefont
  {Prabhavalkar}}, \bibinfo {author} {\bibfnamefont {K.}~\bibnamefont {Rao}},
  \bibinfo {author} {\bibfnamefont {T.~N.}\ \bibnamefont {Sainath}}, \bibinfo
  {author} {\bibfnamefont {B.}~\bibnamefont {Li}}, \bibinfo {author}
  {\bibfnamefont {L.}~\bibnamefont {Johnson}}, \ and\ \bibinfo {author}
  {\bibfnamefont {N.}~\bibnamefont {Jaitly}},\ }\bibfield  {title} {\enquote
  {\bibinfo {title} {A comparison of sequence-to-sequence models for speech
  recognition},}\ }in\ \href {\doibase 10.21437/Interspeech.2017-233} {\emph
  {\bibinfo {booktitle} {Proc. Interspeech 2017}}},\ \bibinfo {editor} {edited
  by\ \bibinfo {editor} {\bibfnamefont {F.}~\bibnamefont {Lacerda}}}\ (\bibinfo
   {publisher} {Curran Associates, Inc.},\ \bibinfo {address} {Red Hook, NY,
  USA},\ \bibinfo {year} {2017})\ pp.\ \bibinfo {pages} {939--943}\BibitemShut
  {NoStop}%
\bibitem [{\citenamefont {Vinyals}\ \emph {et~al.}(2015)\citenamefont
  {Vinyals}, \citenamefont {Toshev}, \citenamefont {Bengio},\ and\
  \citenamefont {Erhan}}]{Vinyals_2015_CVPR}%
  \BibitemOpen
  \bibfield  {author} {\bibinfo {author} {\bibfnamefont {O.}~\bibnamefont
  {Vinyals}}, \bibinfo {author} {\bibfnamefont {A.}~\bibnamefont {Toshev}},
  \bibinfo {author} {\bibfnamefont {S.}~\bibnamefont {Bengio}}, \ and\ \bibinfo
  {author} {\bibfnamefont {D.}~\bibnamefont {Erhan}},\ }\bibfield  {title}
  {\enquote {\bibinfo {title} {Show and tell: A neural image caption
  generator},}\ }in\ \href {\doibase 10.1109/CVPR.2015.7298935} {\emph
  {\bibinfo {booktitle} {2015 IEEE Conference on Computer Vision and Pattern
  Recognition (CVPR)}}},\ \bibinfo {editor} {edited by\ \bibinfo {editor}
  {\bibfnamefont {K.}~\bibnamefont {Barnard}}, \bibinfo {editor} {\bibfnamefont
  {H.}~\bibnamefont {Bischof}}, \bibinfo {editor} {\bibfnamefont
  {P.}~\bibnamefont {Felzenszwalb}}, \bibinfo {editor} {\bibfnamefont
  {D.}~\bibnamefont {Forsyth}}, \bibinfo {editor} {\bibfnamefont
  {S.}~\bibnamefont {Lazebnik}}, \ and\ \bibinfo {editor} {\bibfnamefont
  {J.}~\bibnamefont {Matas}}}\ (\bibinfo  {publisher} {Curran Associates,
  Inc.},\ \bibinfo {address} {Red Hook, NY, USA},\ \bibinfo {year} {2015})\
  pp.\ \bibinfo {pages} {3156--3164}\BibitemShut {NoStop}%
\bibitem [{\citenamefont {Hopfield}(1982)}]{Hopfield2554}%
  \BibitemOpen
  \bibfield  {author} {\bibinfo {author} {\bibfnamefont {J.~J.}\ \bibnamefont
  {Hopfield}},\ }\bibfield  {title} {\enquote {\bibinfo {title} {Neural
  networks and physical systems with emergent collective computational
  abilities},}\ }\href {\doibase 10.1073/pnas.79.8.2554} {\bibfield  {journal}
  {\bibinfo  {journal} {Proc. Natl. Acad. Sci. U.S.A.}\ }\textbf {\bibinfo
  {volume} {79}},\ \bibinfo {pages} {2554} (\bibinfo {year}
  {1982})}\BibitemShut {NoStop}%
\bibitem [{\citenamefont {Hochreiter}\ and\ \citenamefont
  {Schmidhuber}(1997)}]{10.1162/neco.1997.9.8.1735}%
  \BibitemOpen
  \bibfield  {author} {\bibinfo {author} {\bibfnamefont {S.}~\bibnamefont
  {Hochreiter}}\ and\ \bibinfo {author} {\bibfnamefont {J.}~\bibnamefont
  {Schmidhuber}},\ }\bibfield  {title} {\enquote {\bibinfo {title} {Long
  short-term memory},}\ }\href {\doibase 10.1162/neco.1997.9.8.1735} {\bibfield
   {journal} {\bibinfo  {journal} {Neural Comput.}\ }\textbf {\bibinfo {volume}
  {9}},\ \bibinfo {pages} {1735} (\bibinfo {year} {1997})}\BibitemShut
  {NoStop}%
\bibitem [{\citenamefont {Cho}\ \emph {et~al.}(2014)\citenamefont {Cho},
  \citenamefont {van Merri{\"e}nboer}, \citenamefont {Gulcehre}, \citenamefont
  {Bahdanau}, \citenamefont {Bougares}, \citenamefont {Schwenk},\ and\
  \citenamefont {Bengio}}]{cho-etal-2014-learning}%
  \BibitemOpen
  \bibfield  {author} {\bibinfo {author} {\bibfnamefont {K.}~\bibnamefont
  {Cho}}, \bibinfo {author} {\bibfnamefont {B.}~\bibnamefont {van
  Merri{\"e}nboer}}, \bibinfo {author} {\bibfnamefont {C.}~\bibnamefont
  {Gulcehre}}, \bibinfo {author} {\bibfnamefont {D.}~\bibnamefont {Bahdanau}},
  \bibinfo {author} {\bibfnamefont {F.}~\bibnamefont {Bougares}}, \bibinfo
  {author} {\bibfnamefont {H.}~\bibnamefont {Schwenk}}, \ and\ \bibinfo
  {author} {\bibfnamefont {Y.}~\bibnamefont {Bengio}},\ }\bibfield  {title}
  {\enquote {\bibinfo {title} {Learning phrase representations using {RNN}
  encoder{--}decoder for statistical machine translation},}\ }in\ \href
  {\doibase 10.3115/v1/D14-1179} {\emph {\bibinfo {booktitle} {Proceedings of
  the 2014 Conference on Empirical Methods in Natural Language Processing
  ({EMNLP})}}},\ \bibinfo {editor} {edited by\ \bibinfo {editor} {\bibfnamefont
  {A.}~\bibnamefont {Moschitti}}, \bibinfo {editor} {\bibfnamefont
  {B.}~\bibnamefont {Pang}}, \ and\ \bibinfo {editor} {\bibfnamefont
  {W.}~\bibnamefont {Daelemans}}}\ (\bibinfo  {publisher} {Association for
  Computational Linguistics},\ \bibinfo {address} {Doha},\ \bibinfo {year}
  {2014})\ pp.\ \bibinfo {pages} {1724--1734}\BibitemShut {NoStop}%
\bibitem [{\citenamefont {Mokhtari}\ \emph {et~al.}(2019)\citenamefont
  {Mokhtari}, \citenamefont {Ozdaglar},\ and\ \citenamefont
  {Jadbabaie}}]{pmlr-v89-mokhtari19a}%
  \BibitemOpen
  \bibfield  {author} {\bibinfo {author} {\bibfnamefont {A.}~\bibnamefont
  {Mokhtari}}, \bibinfo {author} {\bibfnamefont {A.}~\bibnamefont {Ozdaglar}},
  \ and\ \bibinfo {author} {\bibfnamefont {A.}~\bibnamefont {Jadbabaie}},\
  }\bibfield  {title} {\enquote {\bibinfo {title} {Efficient nonconvex
  empirical risk minimization via adaptive sample size methods},}\ }in\ \href
  {https://proceedings.mlr.press/v89/mokhtari19a.html} {\emph {\bibinfo
  {booktitle} {Proceedings of the Twenty-Second International Conference on
  Artificial Intelligence and Statistics}}},\ \bibinfo {series} {Proceedings of
  Machine Learning Research}, Vol.~\bibinfo {volume} {89},\ \bibinfo {editor}
  {edited by\ \bibinfo {editor} {\bibfnamefont {K.}~\bibnamefont {Chaudhuri}}\
  and\ \bibinfo {editor} {\bibfnamefont {M.}~\bibnamefont {Sugiyama}}}\
  (\bibinfo  {publisher} {PMLR},\ \bibinfo {address} {Naha, Japan},\ \bibinfo
  {year} {2019})\ pp.\ \bibinfo {pages} {2485--2494}\BibitemShut {NoStop}%
\bibitem [{\citenamefont {Mei}\ \emph {et~al.}(2018)\citenamefont {Mei},
  \citenamefont {Bai},\ and\ \citenamefont
  {Montanari}}]{64871590-990e-34a9-8236-2954b5e72da7}%
  \BibitemOpen
  \bibfield  {author} {\bibinfo {author} {\bibfnamefont {S.}~\bibnamefont
  {Mei}}, \bibinfo {author} {\bibfnamefont {Y.}~\bibnamefont {Bai}}, \ and\
  \bibinfo {author} {\bibfnamefont {A.}~\bibnamefont {Montanari}},\ }\bibfield
  {title} {\enquote {\bibinfo {title} {The landscape of empirical risk for
  nonconvex losses},}\ }\href {\doibase 10.1214/17-AOS1637} {\bibfield
  {journal} {\bibinfo  {journal} {Ann. Stat.}\ }\textbf {\bibinfo {volume}
  {46}},\ \bibinfo {pages} {2747} (\bibinfo {year} {2018})}\BibitemShut
  {NoStop}%
\bibitem [{\citenamefont {Yang}\ \emph {et~al.}(2021)\citenamefont {Yang},
  \citenamefont {Hu},\ and\ \citenamefont {Li}}]{yang2021}%
  \BibitemOpen
  \bibfield  {author} {\bibinfo {author} {\bibfnamefont {J.}~\bibnamefont
  {Yang}}, \bibinfo {author} {\bibfnamefont {W.}~\bibnamefont {Hu}}, \ and\
  \bibinfo {author} {\bibfnamefont {C.~J.}\ \bibnamefont {Li}},\ }\bibfield
  {title} {\enquote {\bibinfo {title} {On the fast convergence of random
  perturbations of the gradient flow},}\ }\href {\doibase 10.3233/ASY-201622}
  {\bibfield  {journal} {\bibinfo  {journal} {Asymptot. Anal.}\ }\textbf
  {\bibinfo {volume} {122}},\ \bibinfo {pages} {371} (\bibinfo {year}
  {2021})}\BibitemShut {NoStop}%
\bibitem [{\citenamefont {Dixit}\ \emph {et~al.}(2023)\citenamefont {Dixit},
  \citenamefont {Gurbuzbalaban},\ and\ \citenamefont
  {Bajwa}}]{dixit2023accelerated}%
  \BibitemOpen
  \bibfield  {author} {\bibinfo {author} {\bibfnamefont {R.}~\bibnamefont
  {Dixit}}, \bibinfo {author} {\bibfnamefont {M.}~\bibnamefont
  {Gurbuzbalaban}}, \ and\ \bibinfo {author} {\bibfnamefont {W.~U.}\
  \bibnamefont {Bajwa}},\ }\href@noop {} {\enquote {\bibinfo {title}
  {Accelerated gradient methods for nonconvex optimization: Escape trajectories
  from strict saddle points and convergence to local minima},}\ } (\bibinfo
  {year} {2023}),\ \Eprint {http://arxiv.org/abs/2307.07030}{arXiv:2307.07030
  [math.OC]}\BibitemShut {NoStop}%
\bibitem [{\citenamefont {Kurochkin}(2021)}]{kurochkin2021neural}%
  \BibitemOpen
  \bibfield  {author} {\bibinfo {author} {\bibfnamefont {S.~V.}\ \bibnamefont
  {Kurochkin}},\ }\bibfield  {title} {\enquote {\bibinfo {title} {Neural
  network with smooth activation functions and without bottlenecks is almost
  surely a {Morse} function},}\ }\href {\doibase 10.1134/S0965542521070101}
  {\bibfield  {journal} {\bibinfo  {journal} {Comput. Math. and Math. Phys}\
  }\textbf {\bibinfo {volume} {61}},\ \bibinfo {pages} {1162} (\bibinfo {year}
  {2021})}\BibitemShut {NoStop}%
\bibitem [{\citenamefont {Liu}\ \emph {et~al.}(2016)\citenamefont {Liu},
  \citenamefont {Thompson}, \citenamefont {Weedbrook}, \citenamefont {Lloyd},
  \citenamefont {Vedral}, \citenamefont {Gu},\ and\ \citenamefont
  {Modi}}]{PhysRevA.93.052304}%
  \BibitemOpen
  \bibfield  {author} {\bibinfo {author} {\bibfnamefont {N.}~\bibnamefont
  {Liu}}, \bibinfo {author} {\bibfnamefont {J.}~\bibnamefont {Thompson}},
  \bibinfo {author} {\bibfnamefont {C.}~\bibnamefont {Weedbrook}}, \bibinfo
  {author} {\bibfnamefont {S.}~\bibnamefont {Lloyd}}, \bibinfo {author}
  {\bibfnamefont {V.}~\bibnamefont {Vedral}}, \bibinfo {author} {\bibfnamefont
  {M.}~\bibnamefont {Gu}}, \ and\ \bibinfo {author} {\bibfnamefont
  {K.}~\bibnamefont {Modi}},\ }\bibfield  {title} {\enquote {\bibinfo {title}
  {Power of one qumode for quantum computation},}\ }\href {\doibase
  10.1103/PhysRevA.93.052304} {\bibfield  {journal} {\bibinfo  {journal} {Phys.
  Rev. A}\ }\textbf {\bibinfo {volume} {93}},\ \bibinfo {pages} {052304}
  (\bibinfo {year} {2016})}\BibitemShut {NoStop}%
\bibitem [{\citenamefont {Gottesman}\ \emph {et~al.}(2001)\citenamefont
  {Gottesman}, \citenamefont {Kitaev},\ and\ \citenamefont
  {Preskill}}]{PhysRevA.64.012310}%
  \BibitemOpen
  \bibfield  {author} {\bibinfo {author} {\bibfnamefont {D.}~\bibnamefont
  {Gottesman}}, \bibinfo {author} {\bibfnamefont {A.}~\bibnamefont {Kitaev}}, \
  and\ \bibinfo {author} {\bibfnamefont {J.}~\bibnamefont {Preskill}},\
  }\bibfield  {title} {\enquote {\bibinfo {title} {Encoding a qubit in an
  oscillator},}\ }\href {\doibase 10.1103/PhysRevA.64.012310} {\bibfield
  {journal} {\bibinfo  {journal} {Phys. Rev. A}\ }\textbf {\bibinfo {volume}
  {64}},\ \bibinfo {pages} {012310} (\bibinfo {year} {2001})}\BibitemShut
  {NoStop}%
\bibitem [{\citenamefont {Nielsen}\ and\ \citenamefont
  {Chuang}(2010)}]{Nielsen_Chuang_2010_fourier}%
  \BibitemOpen
  \bibfield  {author} {\bibinfo {author} {\bibfnamefont {M.~A.}\ \bibnamefont
  {Nielsen}}\ and\ \bibinfo {author} {\bibfnamefont {I.~L.}\ \bibnamefont
  {Chuang}},\ }\enquote {\bibinfo {title} {The quantum {Fourier} transform and
  its applications},}\ in\ \href {\doibase 10.1017/CBO9780511976667.009} {\emph
  {\bibinfo {booktitle} {Quantum Computation and Quantum Information: 10th
  Anniversary Edition}}}\ (\bibinfo  {publisher} {Cambridge University Press},\
  \bibinfo {address} {Cambridge, UK},\ \bibinfo {year} {2010})\ pp.\ \bibinfo
  {pages} {216--247}\BibitemShut {NoStop}%
\bibitem [{\citenamefont {Bremner}\ \emph {et~al.}(2017)\citenamefont
  {Bremner}, \citenamefont {Montanaro},\ and\ \citenamefont
  {Shepherd}}]{bremner2017achievingquantum}%
  \BibitemOpen
  \bibfield  {author} {\bibinfo {author} {\bibfnamefont {M.~J.}\ \bibnamefont
  {Bremner}}, \bibinfo {author} {\bibfnamefont {A.}~\bibnamefont {Montanaro}},
  \ and\ \bibinfo {author} {\bibfnamefont {D.~J.}\ \bibnamefont {Shepherd}},\
  }\bibfield  {title} {\enquote {\bibinfo {title} {Achieving quantum supremacy
  with sparse and noisy commuting quantum computations},}\ }\href {\doibase
  10.22331/q-2017-04-25-8} {\bibfield  {journal} {\bibinfo  {journal}
  {Quantum}\ }\textbf {\bibinfo {volume} {1}},\ \bibinfo {pages} {8} (\bibinfo
  {year} {2017})}\BibitemShut {NoStop}%
\bibitem [{\citenamefont {Mermin}(1990)}]{PhysRevLett.65.3373}%
  \BibitemOpen
  \bibfield  {author} {\bibinfo {author} {\bibfnamefont {N.~D.}\ \bibnamefont
  {Mermin}},\ }\bibfield  {title} {\enquote {\bibinfo {title} {Simple unified
  form for the major no-hidden-variables theorems},}\ }\href {\doibase
  10.1103/PhysRevLett.65.3373} {\bibfield  {journal} {\bibinfo  {journal}
  {Phys. Rev. Lett.}\ }\textbf {\bibinfo {volume} {65}},\ \bibinfo {pages}
  {3373} (\bibinfo {year} {1990})}\BibitemShut {NoStop}%
\bibitem [{\citenamefont {Gottesman}\ and\ \citenamefont
  {Chuang}(1999)}]{gottesman1999demonstrating}%
  \BibitemOpen
  \bibfield  {author} {\bibinfo {author} {\bibfnamefont {D.}~\bibnamefont
  {Gottesman}}\ and\ \bibinfo {author} {\bibfnamefont {I.~L.}\ \bibnamefont
  {Chuang}},\ }\bibfield  {title} {\enquote {\bibinfo {title} {Demonstrating
  the viability of universal quantum computation using teleportation and
  single-qubit operations},}\ }\href {\doibase 10.1038/46503} {\bibfield
  {journal} {\bibinfo  {journal} {Nature}\ }\textbf {\bibinfo {volume} {402}},\
  \bibinfo {pages} {390} (\bibinfo {year} {1999})}\BibitemShut {NoStop}%
\bibitem [{\citenamefont {Rabier}(1997)}]{rabier1997ehresmann}%
  \BibitemOpen
  \bibfield  {author} {\bibinfo {author} {\bibfnamefont {P.~J.}\ \bibnamefont
  {Rabier}},\ }\bibfield  {title} {\enquote {\bibinfo {title} {{Ehresmann}
  fibrations and {Palais}-{Smale} conditions for morphisms of {Finsler}
  manifolds},}\ }\href {\doibase 10.2307/2952457} {\bibfield  {journal}
  {\bibinfo  {journal} {Ann. Math.}\ }\textbf {\bibinfo {volume} {146}},\
  \bibinfo {pages} {647} (\bibinfo {year} {1997})}\BibitemShut {NoStop}%
\bibitem [{\citenamefont {Isenhower}\ \emph {et~al.}(2011)\citenamefont
  {Isenhower}, \citenamefont {Saffman},\ and\ \citenamefont
  {M{\o}lmer}}]{isenhower2011multibit}%
  \BibitemOpen
  \bibfield  {author} {\bibinfo {author} {\bibfnamefont {L.}~\bibnamefont
  {Isenhower}}, \bibinfo {author} {\bibfnamefont {M.}~\bibnamefont {Saffman}},
  \ and\ \bibinfo {author} {\bibfnamefont {K.}~\bibnamefont {M{\o}lmer}},\
  }\bibfield  {title} {\enquote {\bibinfo {title} {Multibit
  $\operatorname{C}_k\operatorname{NOT}$ quantum gates via {Rydberg}
  blockade},}\ }\href {\doibase 10.1007/s11128-011-0292-4} {\bibfield
  {journal} {\bibinfo  {journal} {Quantum. Inf. Process.}\ }\textbf {\bibinfo
  {volume} {10}},\ \bibinfo {pages} {755} (\bibinfo {year} {2011})}\BibitemShut
  {NoStop}%
\bibitem [{\citenamefont {Bluvstein}\ \emph {et~al.}(2023)\citenamefont
  {Bluvstein}, \citenamefont {Evered}, \citenamefont {Geim}, \citenamefont
  {Li}, \citenamefont {Zhou}, \citenamefont {Manovitz}, \citenamefont {Ebadi},
  \citenamefont {Cain}, \citenamefont {Kalinowski}, \citenamefont {Hangleiter}
  \emph {et~al.}}]{bluvstein2023logical}%
  \BibitemOpen
  \bibfield  {author} {\bibinfo {author} {\bibfnamefont {D.}~\bibnamefont
  {Bluvstein}}, \bibinfo {author} {\bibfnamefont {S.~J.}\ \bibnamefont
  {Evered}}, \bibinfo {author} {\bibfnamefont {A.~A.}\ \bibnamefont {Geim}},
  \bibinfo {author} {\bibfnamefont {S.~H.}\ \bibnamefont {Li}}, \bibinfo
  {author} {\bibfnamefont {H.}~\bibnamefont {Zhou}}, \bibinfo {author}
  {\bibfnamefont {T.}~\bibnamefont {Manovitz}}, \bibinfo {author}
  {\bibfnamefont {S.}~\bibnamefont {Ebadi}}, \bibinfo {author} {\bibfnamefont
  {M.}~\bibnamefont {Cain}}, \bibinfo {author} {\bibfnamefont {M.}~\bibnamefont
  {Kalinowski}}, \bibinfo {author} {\bibfnamefont {D.}~\bibnamefont
  {Hangleiter}},  \emph {et~al.},\ }\bibfield  {title} {\enquote {\bibinfo
  {title} {Logical quantum processor based on reconfigurable atom arrays},}\
  }\href {\doibase 10.1038/s41586-023-06927-3} {\bibfield  {journal} {\bibinfo
  {journal} {Nature}\ ,\ \bibinfo {pages} {1}} (\bibinfo {year}
  {2023})}\BibitemShut {NoStop}%
\bibitem [{\citenamefont {Kerman}(2013)}]{Kerman_2013}%
  \BibitemOpen
  \bibfield  {author} {\bibinfo {author} {\bibfnamefont {A.~J.}\ \bibnamefont
  {Kerman}},\ }\bibfield  {title} {\enquote {\bibinfo {title} {Quantum
  information processing using quasiclassical electromagnetic interactions
  between qubits and electrical resonators},}\ }\href {\doibase
  10.1088/1367-2630/15/12/123011} {\bibfield  {journal} {\bibinfo  {journal}
  {New J. Phys.}\ }\textbf {\bibinfo {volume} {15}},\ \bibinfo {pages} {123011}
  (\bibinfo {year} {2013})}\BibitemShut {NoStop}%
\bibitem [{\citenamefont {Blais}\ \emph {et~al.}(2021)\citenamefont {Blais},
  \citenamefont {Grimsmo}, \citenamefont {Girvin},\ and\ \citenamefont
  {Wallraff}}]{RevModPhys.93.025005}%
  \BibitemOpen
  \bibfield  {author} {\bibinfo {author} {\bibfnamefont {A.}~\bibnamefont
  {Blais}}, \bibinfo {author} {\bibfnamefont {A.~L.}\ \bibnamefont {Grimsmo}},
  \bibinfo {author} {\bibfnamefont {S.~M.}\ \bibnamefont {Girvin}}, \ and\
  \bibinfo {author} {\bibfnamefont {A.}~\bibnamefont {Wallraff}},\ }\bibfield
  {title} {\enquote {\bibinfo {title} {Circuit quantum electrodynamics},}\
  }\href {\doibase 10.1103/RevModPhys.93.025005} {\bibfield  {journal}
  {\bibinfo  {journal} {Rev. Mod. Phys.}\ }\textbf {\bibinfo {volume} {93}},\
  \bibinfo {pages} {025005} (\bibinfo {year} {2021})}\BibitemShut {NoStop}%
\bibitem [{\citenamefont {Leifer}(2014)}]{Quanta22}%
  \BibitemOpen
  \bibfield  {author} {\bibinfo {author} {\bibfnamefont {M.}~\bibnamefont
  {Leifer}},\ }\bibfield  {title} {\enquote {\bibinfo {title} {Is the quantum
  state real? an extended review of $\uppsi$-ontology theorems},}\ }\href
  {\doibase 10.12743/quanta.v3i1.22} {\bibfield  {journal} {\bibinfo  {journal}
  {Quanta}\ }\textbf {\bibinfo {volume} {3}},\ \bibinfo {pages} {67} (\bibinfo
  {year} {2014})}\BibitemShut {NoStop}%
\bibitem [{\citenamefont {Bravyi}\ \emph {et~al.}(2020)\citenamefont {Bravyi},
  \citenamefont {Gosset}, \citenamefont {Koenig},\ and\ \citenamefont
  {Tomamichel}}]{bravyi2020quantum}%
  \BibitemOpen
  \bibfield  {author} {\bibinfo {author} {\bibfnamefont {S.}~\bibnamefont
  {Bravyi}}, \bibinfo {author} {\bibfnamefont {D.}~\bibnamefont {Gosset}},
  \bibinfo {author} {\bibfnamefont {R.}~\bibnamefont {Koenig}}, \ and\ \bibinfo
  {author} {\bibfnamefont {M.}~\bibnamefont {Tomamichel}},\ }\bibfield  {title}
  {\enquote {\bibinfo {title} {Quantum advantage with noisy shallow
  circuits},}\ }\href {\doibase 10.1038/s41567-020-0948-z} {\bibfield
  {journal} {\bibinfo  {journal} {Nat. Phys.}\ }\textbf {\bibinfo {volume}
  {16}},\ \bibinfo {pages} {1040} (\bibinfo {year} {2020})}\BibitemShut
  {NoStop}%
\bibitem [{\citenamefont {Caha}\ \emph {et~al.}(2023)\citenamefont {Caha},
  \citenamefont {Coiteux-Roy},\ and\ \citenamefont
  {Koenig}}]{caha2023colossal}%
  \BibitemOpen
  \bibfield  {author} {\bibinfo {author} {\bibfnamefont {L.}~\bibnamefont
  {Caha}}, \bibinfo {author} {\bibfnamefont {X.}~\bibnamefont {Coiteux-Roy}}, \
  and\ \bibinfo {author} {\bibfnamefont {R.}~\bibnamefont {Koenig}},\
  }\href@noop {} {\enquote {\bibinfo {title} {A colossal advantage: {3D-local}
  noisy shallow quantum circuits defeat unbounded fan-in classical circuits},}\
  } (\bibinfo {year} {2023}),\ \Eprint
  {http://arxiv.org/abs/2312.09209}{arXiv:2312.09209 [quant-ph]}\BibitemShut
  {NoStop}%
\bibitem [{\citenamefont {Teo}\ \emph {et~al.}(2025)\citenamefont {Teo},
  \citenamefont {Yang}, \citenamefont {Sud}, \citenamefont {Tomesh},
  \citenamefont {Chong},\ and\ \citenamefont
  {Anschuetz}}]{teo2025kcontextualityheuristicmemoryseparations}%
  \BibitemOpen
  \bibfield  {author} {\bibinfo {author} {\bibfnamefont {M.~H.}\ \bibnamefont
  {Teo}}, \bibinfo {author} {\bibfnamefont {W.}~\bibnamefont {Yang}}, \bibinfo
  {author} {\bibfnamefont {J.}~\bibnamefont {Sud}}, \bibinfo {author}
  {\bibfnamefont {T.}~\bibnamefont {Tomesh}}, \bibinfo {author} {\bibfnamefont
  {F.~T.}\ \bibnamefont {Chong}}, \ and\ \bibinfo {author} {\bibfnamefont
  {E.~R.}\ \bibnamefont {Anschuetz}},\ }\href@noop {} {\enquote {\bibinfo
  {title} {k-contextuality as a heuristic for memory separations in
  learning},}\ } (\bibinfo {year} {2025}),\ \Eprint
  {http://arxiv.org/abs/2507.11604}{arXiv:2507.11604 [quant-ph]}\BibitemShut
  {NoStop}%
\bibitem [{\citenamefont {Brown}\ \emph {et~al.}(2020)\citenamefont {Brown}
  \emph {et~al.}}]{brown2020}%
  \BibitemOpen
  \bibfield  {author} {\bibinfo {author} {\bibfnamefont {T.}~\bibnamefont
  {Brown}} \emph {et~al.},\ }\bibfield  {title} {\enquote {\bibinfo {title}
  {Language models are few-shot learners},}\ }in\ \href
  {https://papers.nips.cc/paper/2020/hash/1457c0d6bfcb4967418bfb8ac142f64a-Abstract.html}
  {\emph {\bibinfo {booktitle} {Proceedings of the 34th International
  Conference on Neural Information Processing Systems}}},\ \bibinfo {series and
  number} {NIPS'20},\ \bibinfo {editor} {edited by\ \bibinfo {editor}
  {\bibfnamefont {H.}~\bibnamefont {Larochelle}}, \bibinfo {editor}
  {\bibfnamefont {M.}~\bibnamefont {Ranzato}}, \bibinfo {editor} {\bibfnamefont
  {R.}~\bibnamefont {Hadsell}}, \bibinfo {editor} {\bibfnamefont
  {M.}~\bibnamefont {Balcan}}, \ and\ \bibinfo {editor} {\bibfnamefont
  {H.}~\bibnamefont {Lin}}}\ (\bibinfo  {publisher} {Curran Associates, Inc.},\
  \bibinfo {address} {Red Hook, NY, USA},\ \bibinfo {year} {2020})\ pp.\
  \bibinfo {pages} {1877--1901}\BibitemShut {NoStop}%
\bibitem [{\citenamefont {OpenAI}(2023)}]{openai2023}%
  \BibitemOpen
  \bibfield  {author} {\bibinfo {author} {\bibnamefont {OpenAI}},\ }\href@noop
  {} {\enquote {\bibinfo {title} {{GPT}-4 technical report},}\ } (\bibinfo
  {year} {2023}),\ \Eprint {http://arxiv.org/abs/2303.08774}{arXiv:2303.08774
  [cs.CL]}\BibitemShut {NoStop}%
\bibitem [{\citenamefont {Wan}\ \emph {et~al.}(2021)\citenamefont {Wan},
  \citenamefont {Choi}, \citenamefont {Kim}, \citenamefont {Shutty},\ and\
  \citenamefont {Hayden}}]{PRXQuantum.2.040345}%
  \BibitemOpen
  \bibfield  {author} {\bibinfo {author} {\bibfnamefont {K.}~\bibnamefont
  {Wan}}, \bibinfo {author} {\bibfnamefont {S.}~\bibnamefont {Choi}}, \bibinfo
  {author} {\bibfnamefont {I.~H.}\ \bibnamefont {Kim}}, \bibinfo {author}
  {\bibfnamefont {N.}~\bibnamefont {Shutty}}, \ and\ \bibinfo {author}
  {\bibfnamefont {P.}~\bibnamefont {Hayden}},\ }\bibfield  {title} {\enquote
  {\bibinfo {title} {Fault-tolerant qubit from a constant number of
  components},}\ }\href {\doibase 10.1103/PRXQuantum.2.040345} {\bibfield
  {journal} {\bibinfo  {journal} {PRX Quantum}\ }\textbf {\bibinfo {volume}
  {2}},\ \bibinfo {pages} {040345} (\bibinfo {year} {2021})}\BibitemShut
  {NoStop}%
\bibitem [{\citenamefont {Srikumar}\ \emph {et~al.}(2024)\citenamefont
  {Srikumar}, \citenamefont {Bartlett},\ and\ \citenamefont
  {Karanjai}}]{srikumar2024contextuality}%
  \BibitemOpen
  \bibfield  {author} {\bibinfo {author} {\bibfnamefont {M.}~\bibnamefont
  {Srikumar}}, \bibinfo {author} {\bibfnamefont {S.~D.}\ \bibnamefont
  {Bartlett}}, \ and\ \bibinfo {author} {\bibfnamefont {A.}~\bibnamefont
  {Karanjai}},\ }\href@noop {} {\enquote {\bibinfo {title} {How contextuality
  and antidistinguishability are related},}\ } (\bibinfo {year} {2024}),\
  \Eprint {http://arxiv.org/abs/2411.09919}{arXiv:2411.09919
  [quant-ph]}\BibitemShut {NoStop}%
\bibitem [{\citenamefont {Mari}\ and\ \citenamefont
  {Eisert}(2012)}]{PhysRevLett.109.230503}%
  \BibitemOpen
  \bibfield  {author} {\bibinfo {author} {\bibfnamefont {A.}~\bibnamefont
  {Mari}}\ and\ \bibinfo {author} {\bibfnamefont {J.}~\bibnamefont {Eisert}},\
  }\bibfield  {title} {\enquote {\bibinfo {title} {Positive {W}igner functions
  render classical simulation of quantum computation efficient},}\ }\href
  {\doibase 10.1103/PhysRevLett.109.230503} {\bibfield  {journal} {\bibinfo
  {journal} {Phys. Rev. Lett.}\ }\textbf {\bibinfo {volume} {109}},\ \bibinfo
  {pages} {230503} (\bibinfo {year} {2012})}\BibitemShut {NoStop}%
\bibitem [{\citenamefont {Goodfellow}\ \emph {et~al.}(2014)\citenamefont
  {Goodfellow}, \citenamefont {Pouget-Abadie}, \citenamefont {Mirza},
  \citenamefont {Xu}, \citenamefont {Warde-Farley}, \citenamefont {Ozair},
  \citenamefont {Courville},\ and\ \citenamefont {Bengio}}]{NIPS2014_5ca3e9b1}%
  \BibitemOpen
  \bibfield  {author} {\bibinfo {author} {\bibfnamefont {I.~J.}\ \bibnamefont
  {Goodfellow}}, \bibinfo {author} {\bibfnamefont {J.}~\bibnamefont
  {Pouget-Abadie}}, \bibinfo {author} {\bibfnamefont {M.}~\bibnamefont
  {Mirza}}, \bibinfo {author} {\bibfnamefont {B.}~\bibnamefont {Xu}}, \bibinfo
  {author} {\bibfnamefont {D.}~\bibnamefont {Warde-Farley}}, \bibinfo {author}
  {\bibfnamefont {S.}~\bibnamefont {Ozair}}, \bibinfo {author} {\bibfnamefont
  {A.}~\bibnamefont {Courville}}, \ and\ \bibinfo {author} {\bibfnamefont
  {Y.}~\bibnamefont {Bengio}},\ }\bibfield  {title} {\enquote {\bibinfo {title}
  {Generative adversarial nets},}\ }in\ \href
  {https://papers.nips.cc/paper_files/paper/2014/hash/5ca3e9b122f61f8f06494c97b1afccf3-Abstract.html}
  {\emph {\bibinfo {booktitle} {Proceedings of the 27th International
  Conference on Neural Information Processing Systems - Volume 2}}},\ \bibinfo
  {series and number} {NIPS'14},\ \bibinfo {editor} {edited by\ \bibinfo
  {editor} {\bibfnamefont {Z.}~\bibnamefont {Ghahramani}}, \bibinfo {editor}
  {\bibfnamefont {M.}~\bibnamefont {Welling}}, \bibinfo {editor} {\bibfnamefont
  {C.}~\bibnamefont {Cortes}}, \bibinfo {editor} {\bibfnamefont
  {N.}~\bibnamefont {Lawrence}}, \ and\ \bibinfo {editor} {\bibfnamefont
  {K.}~\bibnamefont {Weinberger}}}\ (\bibinfo  {publisher} {Curran Associates,
  Inc.},\ \bibinfo {address} {Red Hook, NY, USA},\ \bibinfo {year} {2014})\
  pp.\ \bibinfo {pages} {2672--2680}\BibitemShut {NoStop}%
\bibitem [{\citenamefont {Shi}\ \emph {et~al.}(2020)\citenamefont {Shi},
  \citenamefont {Su},\ and\ \citenamefont {Jordan}}]{shi2020learning}%
  \BibitemOpen
  \bibfield  {author} {\bibinfo {author} {\bibfnamefont {B.}~\bibnamefont
  {Shi}}, \bibinfo {author} {\bibfnamefont {W.~J.}\ \bibnamefont {Su}}, \ and\
  \bibinfo {author} {\bibfnamefont {M.~I.}\ \bibnamefont {Jordan}},\
  }\href@noop {} {\enquote {\bibinfo {title} {On learning rates and
  {Schr\"odinger} operators},}\ } (\bibinfo {year} {2020}),\ \Eprint
  {http://arxiv.org/abs/2004.06977}{arXiv:2004.06977 [cs.LG]}\BibitemShut
  {NoStop}%
\bibitem [{\citenamefont {Duan}\ \emph {et~al.}(2022)\citenamefont {Duan},
  \citenamefont {Wu},\ and\ \citenamefont {Zhou}}]{duan2022faster}%
  \BibitemOpen
  \bibfield  {author} {\bibinfo {author} {\bibfnamefont {R.}~\bibnamefont
  {Duan}}, \bibinfo {author} {\bibfnamefont {H.}~\bibnamefont {Wu}}, \ and\
  \bibinfo {author} {\bibfnamefont {R.}~\bibnamefont {Zhou}},\ }\href@noop {}
  {\enquote {\bibinfo {title} {Faster matrix multiplication via asymmetric
  hashing},}\ } (\bibinfo {year} {2022}),\ \Eprint
  {http://arxiv.org/abs/2210.10173}{arXiv:2210.10173 [cs.DS]}\BibitemShut
  {NoStop}%
\bibitem [{\citenamefont {Strassen}(1969)}]{strassen1969gaussian}%
  \BibitemOpen
  \bibfield  {author} {\bibinfo {author} {\bibfnamefont {V.}~\bibnamefont
  {Strassen}},\ }\bibfield  {title} {\enquote {\bibinfo {title} {Gaussian
  elimination is not optimal},}\ }\href {\doibase 10.1007/BF02165411}
  {\bibfield  {journal} {\bibinfo  {journal} {Numer. Math.}\ }\textbf {\bibinfo
  {volume} {13}},\ \bibinfo {pages} {354} (\bibinfo {year} {1969})}\BibitemShut
  {NoStop}%
\bibitem [{\citenamefont {Aaronson}\ and\ \citenamefont
  {Gottesman}(2004)}]{PhysRevA.70.052328}%
  \BibitemOpen
  \bibfield  {author} {\bibinfo {author} {\bibfnamefont {S.}~\bibnamefont
  {Aaronson}}\ and\ \bibinfo {author} {\bibfnamefont {D.}~\bibnamefont
  {Gottesman}},\ }\bibfield  {title} {\enquote {\bibinfo {title} {Improved
  simulation of stabilizer circuits},}\ }\href {\doibase
  10.1103/PhysRevA.70.052328} {\bibfield  {journal} {\bibinfo  {journal} {Phys.
  Rev. A}\ }\textbf {\bibinfo {volume} {70}},\ \bibinfo {pages} {052328}
  (\bibinfo {year} {2004})}\BibitemShut {NoStop}%
\bibitem [{\citenamefont {Calcluth}\ \emph {et~al.}(2022)\citenamefont
  {Calcluth}, \citenamefont {Ferraro},\ and\ \citenamefont
  {Ferrini}}]{calcuth2022}%
  \BibitemOpen
  \bibfield  {author} {\bibinfo {author} {\bibfnamefont {C.}~\bibnamefont
  {Calcluth}}, \bibinfo {author} {\bibfnamefont {A.}~\bibnamefont {Ferraro}}, \
  and\ \bibinfo {author} {\bibfnamefont {G.}~\bibnamefont {Ferrini}},\
  }\href@noop {} {\enquote {\bibinfo {title} {Efficient simulation of
  {G}ottesman-{K}itaev-{P}reskill states with {G}aussian circuits},}\ }
  (\bibinfo {year} {2022}),\ \Eprint
  {http://arxiv.org/abs/2203.11182}{arXiv:2203.11182 [quant-ph]}\BibitemShut
  {NoStop}%
\bibitem [{\citenamefont {Larocca}\ \emph {et~al.}(2022)\citenamefont
  {Larocca}, \citenamefont {Czarnik}, \citenamefont {Sharma}, \citenamefont
  {Muraleedharan}, \citenamefont {Coles},\ and\ \citenamefont
  {Cerezo}}]{larocca2022diagnosingbarren}%
  \BibitemOpen
  \bibfield  {author} {\bibinfo {author} {\bibfnamefont {M.}~\bibnamefont
  {Larocca}}, \bibinfo {author} {\bibfnamefont {P.}~\bibnamefont {Czarnik}},
  \bibinfo {author} {\bibfnamefont {K.}~\bibnamefont {Sharma}}, \bibinfo
  {author} {\bibfnamefont {G.}~\bibnamefont {Muraleedharan}}, \bibinfo {author}
  {\bibfnamefont {P.~J.}\ \bibnamefont {Coles}}, \ and\ \bibinfo {author}
  {\bibfnamefont {M.}~\bibnamefont {Cerezo}},\ }\bibfield  {title} {\enquote
  {\bibinfo {title} {Diagnosing barren plateaus with tools from quantum optimal
  control},}\ }\href {\doibase 10.22331/q-2022-09-29-824} {\bibfield  {journal}
  {\bibinfo  {journal} {{Quantum}}\ }\textbf {\bibinfo {volume} {6}},\ \bibinfo
  {pages} {824} (\bibinfo {year} {2022})}\BibitemShut {NoStop}%
\end{thebibliography}%

\appendix

\section{Weighted Hypergraph States}\label{sec:seq_learning_background}

\subsection{Qubit Weighted Hypergraph States}\label{sec:qubit_hypergraph}

We here give a bevy of facts about \emph{qubit weighted hypergraph states}. Though first introduced in Ref.~\cite{Webster2022xpstabiliser} in the context of the so-called XP stabilizer formalism, we here directly consider a stabilizer formalism of such states.

We first review the definition of qubit weighted hypergraph states, focusing for simplicity on cases where the associated hypergraphs are \emph{$k$-uniform hypergraphs}; namely, their only nontrivial hyperedges are associated with cardinality $k$ subsets of vertices. Let $G$ be a $k$-uniform weighted hypergraph with hyperedges labeled by sets $\overline{v}$ and with associated hyperedge weights $e_{\overline{v}}$. We denote as $\ket{G}$ the \emph{qubit $k$-uniform weighted hypergraph state associated with $G$}:
\begin{equation}
    \ket{G}:=\prod\limits_{\overline{v}\in G}U_{\overline{v}}\left(e_{\overline{v}}\right)\ket{+}^{\otimes n},
\end{equation}
where we use the shorthand:
\begin{equation}
    U_{\overline{v}}\left(e_{\overline{v}}\right):=\operatorname{C}^{k-1}Z_{\overline{v}}\left(-e_{\overline{v}}\right)=\exp\left(\ci\cpi e_{\overline{v}}\prod\limits_{i\in\overline{v}}\frac{I-Z_i}{2}\right)\label{eq:ckz_def}
\end{equation}
to ease on notation. It is easy to check the commutation relations when $i\in\overline{v}$:
\begin{align}
    X_iU_{\overline{v}}\left(e_{\overline{v}}\right)&=U_{\overline{v}\setminus\left\{i\right\}}\left(e_{\overline{v}}\right)U_{\overline{v}}\left(-e_{\overline{v}}\right)X_i,\label{eq:hypergraph_stab_comm_1}\\
    U_{\overline{v}}\left(e_{\overline{v}}\right)X_i&=X_iU_{\overline{v}\setminus\left\{i\right\}}\left(e_{\overline{v}}\right)U_{\overline{v}}\left(-e_{\overline{v}}\right).\label{eq:hypergraph_stab_comm_2}
\end{align}
By construction, $\ket{G}$ is the mutual $+1$ eigenstate of the commuting operators:
\begin{equation}
    \begin{aligned}
        s_i&=\prod\limits_{\overline{v}\ni i}U_{\overline{v}}\left(e_{\overline{v}}\right)X_i\left(\prod\limits_{\overline{v}\ni i}U_{\overline{v}}\left(e_{\overline{v}}\right)\right)^\dagger\\
        &=X_i\prod\limits_{\overline{v}\ni i}U_{\overline{v}\setminus\left\{i\right\}}\left(e_{\overline{v}}\right)U_{\overline{v}}\left(-2e_{\overline{v}}\right)\\
        &=\left(\prod\limits_{\overline{v}\ni i}U_{\overline{v}\setminus\left\{i\right\}}\left(-e_{\overline{v}}\right)U_{\overline{v}}\left(2e_{\overline{v}}\right)\right)X_i.\label{eq:hermiticity_of_stab}
    \end{aligned}
\end{equation}
From this it is immediately clear that the $s_i$ are their own inverses:
\begin{equation}
    s_i^2=I.
\end{equation}
Thus,
\begin{equation}
    s_i=\exp\left(-\ci\frac{\cpi}{2}\right)\exp\left(\ci\frac{\cpi}{2}s_i\right)
\end{equation}
and in particular $\ket{G}$ is also stabilized by all powers $s_i^\alpha$. It is also apparent from Eq.~\eqref{eq:hermiticity_of_stab} that the $s_i$ are Hermitian.

We now consider the special case where $G$ has only a single hyperedge $\overline{v}$. Consider the stabilizer of $\ket{G}$:
\begin{equation}
    s_{\overline{v}}:=\prod\limits_{i\in\overline{v}}s_i
\end{equation}
associated with some vertex set $\overline{v}$. By Eq.~\eqref{eq:hypergraph_stab_comm_2} it is apparent this can be written as:
\begin{equation}
    s_{\overline{v}}=\left(\prod\limits_{i\in\overline{v}}X_i\right)\exp\left(\ci\cpi e_{\overline{v}}R\right),\label{eq:R_def}
\end{equation}
where $R$ is some $e_{\overline{v}}$-independent multilinear degree-$k$ polynomial in the $\frac{I-Z_i}{2}$ with integer coefficients. As the commutation in Eq.~\eqref{eq:hypergraph_stab_comm_2} only introduces terms of up to a single degree lower than the commuted $U_{\overline{v}}$, we can also calculate the constant term $R_0$ in $R$:
\begin{equation}
    R_0=1.
\end{equation}
In particular,
\begin{equation}
    \exp\left(\ci\cpi e_{\overline{v}}R\right)\ket{0}^{\otimes n}=\exp\left(\ci\cpi e_{\overline{v}}\right)\ket{0}^{\otimes n}.
\end{equation}
Similarly, by Eq.~\eqref{eq:hypergraph_stab_comm_1}
\begin{equation}
    s_{\overline{v}}=\exp\left(\ci\cpi e_{\overline{v}}L\right)\prod\limits_{i\in\overline{v}}X_i,
\end{equation}
where once again $L$ is some $e_{\overline{v}}$-independent multilinear degree-$k$ polynomial in the $\frac{I-Z_i}{2}$. By the Hermiticity of $s_{\overline{v}}$ we must have that:
\begin{equation}
    L=-R.
\end{equation}
In particular, when $e_{\overline{v}}$ is an integer, we have by equating these expressions for $s_{\overline{v}}$ that:
\begin{equation}
    \left[\cos\left(\cpi e_{\overline{v}}\left(R-1\right)\right),\prod\limits_{i\in\overline{v}}X_i\right]=\bm{0},\label{eq:qubit_meas_commute}
\end{equation}
where:
\begin{equation}
    \cos\left(\cpi e_{\overline{v}}\left(R-1\right)\right)\ket{0}^{\otimes n}=\ket{0}^{\otimes n}.
\end{equation}

\subsection{Qumode Weighted Hypergraph States}\label{sec:qumode_hypergraph}

One can consider a qumode version of $k$-uniform weighted hypergraph states completely analogously to qubit weighted hypergraph states~\cite{takeuchi2019resource}. The standard definition of such states involves the sequential application of qumode multi-controlled phase operators associated with a $k$-uniform hypergraph $G$ on a momentum-squeezed state:
\begin{equation}
    \ket{G}:=\prod\limits_{\overline{v}\in G}\exp\left(\ci\cpi e_{\overline{v}}\prod\limits_{j\in\overline{v}}\hat{q}_j\right)\ket{\bm{0}}_{\bm{\hat{p}}}.
\end{equation}
Here, we use the notation $\ket{\bm{a}}_{\bm{\hat{p}}}=\bigotimes\ket{a_i}_{\hat{p}_i}=\bigotimes\ket{\hat{p}_i=a_i}$. This is the definition presented in Sec.~\ref{sec:qumode_hrnn_inf}.

In the formal proof of our results---presented in Appendix~\ref{sec:proof_of_express_sep}---we will consider a slightly modified definition. Rather than applying gates to a squeezed momentum state, we consider an initial \emph{GKP state}~\cite{PhysRevA.64.012310}:
\begin{equation}
    \ket{\text{GKP}}\propto\bigotimes_{i=1}^n\left(\sum\limits_{\alpha=-\infty}^\infty\ket{\alpha}_{\hat{q}_i}\right).\label{eq:stand_gkp_state}
\end{equation}
Specifically, we take \emph{qumode $k$-uniform weighted hypergraph states} associated with a $k$-uniform hypergraph $G$ to be defined as:
\begin{equation}
    \ket{G}:=\prod\limits_{\overline{v}\in G}U_{\overline{v}}\left(e_{\overline{v}}\right)\ket{\text{GKP}},
\end{equation}
where now in a qumode context we have defined:
\begin{equation}
    U_{\overline{v}}\left(e_{\overline{v}}\right):=\exp\left(\ci\cpi e_{\overline{v}}\prod\limits_{j\in\overline{v}}\hat{q}_j\right).
\end{equation}
It is once again easy to check the commutation relations when $i\in\overline{v}$:
\begin{align}
    X_iU_{\overline{v}}\left(e_{\overline{v}}\right)&=U_{\overline{v}\setminus\left\{i\right\}}\left(-e_{\overline{v}}\right)U_{\overline{v}}\left(e_{\overline{v}}\right)X_i,\label{eq:qumode_hypergraph_stab_comm_1}\\
    U_{\overline{v}}\left(e_{\overline{v}}\right)X_i&=X_iU_{\overline{v}\setminus\left\{i\right\}}\left(e_{\overline{v}}\right)U_{\overline{v}}\left(e_{\overline{v}}\right),\label{eq:qumode_hypergraph_stab_comm_2}
\end{align}
where now:
\begin{equation}
    X_i:=\exp\left(-2\ci\hat{p}_i\right).
\end{equation}
Similarly to qubit $k$-uniform weighted hypergraph states, the qumode version of these states are described by commuting stabilizers:
\begin{equation}
    \begin{aligned}
        s_i&=\prod\limits_{\overline{v}\ni i}U_{\overline{v}}\left(e_{\overline{v}}\right)X_i\left(\prod\limits_{\overline{v}\ni i}U_{\overline{v}}\left(e_{\overline{v}}\right)\right)^\dagger\\
        &=X_i\prod\limits_{\overline{v}\ni i}U_{\overline{v}\setminus\left\{i\right\}}\left(e_{\overline{v}}\right)
    \end{aligned}
\end{equation}
and their Hermitian conjugates. Unlike the qubit case, however, the $s_i$ are now not Hermitian.

We now consider the special case where $G$ has only a single hyperedge $\overline{v}$. Consider the stabilizer of $\ket{G}$:
\begin{equation}
    s_{\overline{v}}:=\prod\limits_{i\in\overline{v}}s_i
\end{equation}
associated with some vertex set $\overline{v}$. By Eq.~\eqref{eq:qumode_hypergraph_stab_comm_2} it is apparent this can be written as:
\begin{equation}
    s_{\overline{v}}=\left(\prod\limits_{i\in\overline{v}}X_i\right)\exp\left(\ci\cpi e_{\overline{v}}R\right),\label{eq:qumode_R_def}
\end{equation}
where $R$ is some $e_{\overline{v}}$-independent multilinear degree-$\left(k-1\right)$ polynomial in the $\hat{q}_i$ with integer coefficients. As the commutation in Eq.~\eqref{eq:qumode_hypergraph_stab_comm_2} only introduces terms of up to a single degree lower than the commuted $U_{\overline{v}}$, we can also calculate the constant term $R_0$ in $R$:
\begin{equation}
    R_0=1.
\end{equation}
In particular,
\begin{equation}
    \exp\left(\ci\cpi e_{\overline{v}}R\right)\ket{\bm{0}}_{\bm{\hat{q}}}=\exp\left(\ci\cpi e_{\overline{v}}\right)\ket{\bm{0}}_{\bm{\hat{q}}}.
\end{equation}

Unlike the qubit case, the $s_{\overline{v}}$ are now \emph{not} Hermitian, and particular it is not true that $\exp\left(\ci\cpi e_{\overline{v}}R\right)$ commutes through $\prod\limits_{i\in\overline{v}}X_i$ up to conjugation. However, this \emph{is} true up to stabilizers of $\ket{G}$ when $e_{\overline{v}}$ is an integer. Namely, using Eq.~\eqref{eq:qumode_hypergraph_stab_comm_2} we have that:
\begin{equation}
    \exp\left(\ci\cpi e_{\overline{v}}R\right)\left(\prod\limits_{i\in\overline{v}}X_i\right)=\left(\prod\limits_{i\in\overline{v}}X_i\right)\exp\left(\ci\cpi e_{\overline{v}}R\right)\exp\left(2\ci\cpi e_{\overline{v}}\tilde{R}\right),
\end{equation}
where $\tilde{R}$ is another $e_{\overline{v}}$-independent multilinear degree-$k$ polynomial in the $\hat{q}_i$ with integer coefficients and no constant term. As $\exp\left(2\ci\cpi e_{\overline{v}}\tilde{R}\right)$ commutes with $U_{\overline{v}}\left(e_{\overline{v}}\right)$ and as $\ket{\text{GKP}}$ is a uniform superposition state over squeezed integer positions,
\begin{equation}
    \exp\left(\ci\cpi e_{\overline{v}}R\right)\left(\prod\limits_{i\in\overline{v}}X_i\right)\ket{G}=\left(\prod\limits_{i\in\overline{v}}X_i\right)\exp\left(\ci\cpi e_{\overline{v}}R\right)\ket{G}.
\end{equation}
This is similar to an idea used in the work of Ref.~\cite{PRXQuantum.2.040345}, where generated high-weight errors in a proposed fault-tolerant quantum computing scheme were shown to be equivalent to low-weight errors up to stabilizers. We thus have when $e_{\overline{v}}$ is an integer that for all integer $\alpha$:
\begin{equation}
    \left[\exp\left(\ci\alpha\cpi e_{\overline{v}}\left(R-1\right)\right),\prod\limits_{i\in\overline{v}}X_i\right]\ket{G}=\bm{0},\label{eq:qumode_meas_commute}
\end{equation}
where:
\begin{equation}
    \exp\left(\ci\cpi e_{\overline{v}}\left(R-1\right)\right)\ket{\bm{0}}_{\bm{\hat{q}}}=\ket{\bm{0}}_{\bm{\hat{q}}}.
\end{equation}
As the projector onto the $\exp\left(\ci\cpi e_{\overline{v}}\left(R-1\right)\right)=1$ subspace can be written as integer powers of $\exp\left(\ci\cpi e_{\overline{v}}\left(R-1\right)\right)$, this implies that the measurement (given this measurement outcome) commutes with the product of vertex stabilizers of $\ket{G}$.

\section{Proofs of Expressivity Separations}\label{sec:proof_of_express_sep}

\subsection{\texorpdfstring{$\left(\ell,n,k\right)$}{(l,n,k)}-Hypergraph Stabilizer Measurement Translation}\label{sec:task_desc}

Before giving proofs of expressivity separations between our quantum model and classical models, we first give a formal definition of the translation task we will prove a separation on: namely, \emph{$\left(\ell,n,k\right)$-hypergraph stabilizer measurement translation} ($\left(\ell,n,k\right)$-HSMT), parameterized by $\ell$, $n$, and $k$, where $\ell\geq\binom{n}{k}+3n-k$. Note that the technical description of the task described here differs from construction presented in the main text, mainly in that there are additional constraints on the forms the various input tokens $\bm{x_i}$ can take. The fact that the $k$-HRNN described in Sec.~\ref{sec:hrnn} (or a slight variant thereof) can perfectly perform the task using $\operatorname{O}\left(n\right)$ qubits or qumodes will follow immediately. In what follows, we use the notation $\binom{\left[n\right]}{r}$ to denote the set of $r$-tuples with distinct elements all drawn from $\left\{1,\ldots,n\right\}$.

We consider an input language given by $\ell$-long sequences of tokens:
\begin{equation}
    \bm{x}=\left(\bm{x}^{\left(1\right)},\ldots,\bm{x}^{\left(\ell\right)}\right)^\intercal.
\end{equation}
Recall from the main text that a qubit $k$-HRNN has initial latent state $\ket{\lambda_1}=\ket{0}^{\otimes n}$, and a qumode $k$-HRNN has initial latent state $\ket{\lambda_1}=\ket{\bm{0}}_{\hat{q}}$ (i.e., it is initially squeezed in the position basis). Just as in the main text, tokens in the $\left(\ell,n,k\right)$-HSMT task will be of the form $\left(\bm{\alpha},\bm{\beta},\gamma,\bm{\phi},\bm{\theta}\right)$, describing parameter settings of a $k$-HRNN. However, to aid in our proofs, we assume a variant of the translation task with certain additional constraints on the tokens.

We assume that the first $n$ tokens are of the form:
\begin{equation}
    \left(\bm{\alpha},\bm{\beta},\gamma,\bm{\phi},\bm{\theta}^{\left(i\right)}\right)=\left(0,i,0,-2\upsilon_i,0\right);
\end{equation}
i.e., it corresponds to a $k$-HRNN measuring either $\exp\left(2\ci\upsilon_i X_i\right)$ (in the qubit setting) or $\exp\left(2\ci\upsilon_i\hat{p}_i\right)$ (in the qumode setting) via phase estimation. In the qumode setting, when $\upsilon_i=1$ and the measurement outcomes are $+1$, this sequence of measurements projects onto the GKP state described in Eq.~\eqref{eq:stand_gkp_state}; this can be seen by considering the phase estimation circuit~\cite{PhysRevA.93.052304} measuring $O_{\theta_{i,p}}$ in the position basis, and considering its action on the state $\ket{0}_{\hat{q}_i}$.

The next $\binom{n}{k}$ tokens---labeled by $k$-tuples $\overline{v}\in\binom{\left[n\right]}{k}$ in an arbitrary order---are of the form:
\begin{equation}
    \left(\bm{\alpha},\bm{\beta},\gamma,\bm{\phi},\bm{\theta}^{\left(i\right)}\right)=\left(\overline{v},0,-\gamma_{\overline{v}},0,0\right),
\end{equation}
i.e., it corresponds to a $k$-HRNN applying $\operatorname{C}^{k-1}Z_{\overline{v}}\left(\gamma_{\overline{v}}\right)$ sequentially.

The next $n$ tokens query the model to reproduce the first $n$ measurement results (or are excluded, if we are studying hypergraph stabilizer measurement \emph{without repetition}). This is essentially just a technical assumption that will allow us to assume any states indistinguishable by a classical model at this point in the sequence have identical outputs. This can be captured by a slightly-modified $k$-HRNN architecture which explicitly stores its first $n$ measurement results in an extra $n$ qumodes of memory.

The next $n-k$ tokens are of the form:
\begin{equation}
    \left(\bm{\alpha},\bm{\beta},\gamma,\bm{\phi},\bm{\theta}\right)=\left(0,b_i,0,0,\theta_{\left\{b_i\right\}}=\epsilon\right)
\end{equation}
for distinct integer $b_i$ and $\epsilon=1$ (in the qubit setting) or $\epsilon\to 0$ (in the qumode setting). That is, it corresponds to a $k$-HRNN measuring $Z_{b_i}$ (in the qubit setting) or $\hat{q}_{b_i}$ (in the qumode setting) via phase estimation.\footnote{In practice, one would change the ancillary register $\ket{a}$ from $\ket{\text{GKP}}$ to $\ket{0}_{\hat{p}}$ in performing phase estimation at this stage in the qumode setting; this only slightly modifies the $k$-HRNN architecture.}

Finally, all remaining tokens are of the form:
\begin{equation}
    \left(\bm{\alpha},\bm{\beta},\gamma,\bm{\phi},\bm{\theta}\right)=\left(0,\overline{c},0,-\bm{\phi},-\bm{\theta}\right)
\end{equation}
where $\overline{c}$ is the remaining $k$ qumodes that were not one of the $b_i$, i.e., it corresponds to a general $k$-HRNN measuring general $G_\cdot$ on the $k$ qumodes previously unmeasured.

By construction, all input sequences $\bm{x}$ are then described by $\bm{\upsilon}\in\mathbb{R}^n$, $\bm{\gamma}\in\mathbb{R}^{\binom{n}{k}}$, $\overline{b}\in\binom{\left[n\right]}{n-k}$, and $\left(\bm{\phi},\bm{\theta}\right)$ describe the remainder of the sequence. A translation $\bm{y}$ of $\bm{x}$ is considered correct if it is of the form:
\begin{equation}
    \bm{y}=\left(m_1,\ldots,m_n,0,\ldots,0,m_1,\ldots,m_n,m_{\ell-3n+1},\ldots,m_\ell\right)^\intercal,
\end{equation}
where the measurement outcomes $m_i$ are consistent with those of a $k$-HRNN performing the previously described procedure.

\subsection{Single-copy Antidistinguishing Measurement Sequences via Quantum Contextuality}\label{sec:dist_meas_seqs}

We now prove two lemmas demonstrating the existence of \emph{antidistinguishing measurement sequences} given by qubit $k$-uniform weighted hypergraph state stabilizers, as well as qumode $k$-uniform hypergraph state stabilizers. Antidistinguishing measurement sequences are those which, given $r$ candidates for a given quantum state, rule out one candidate with certainty~\cite{Quanta22}. Their existence is equivalent to the presence of quantum contextuality in a set of quantum states~\cite{srikumar2024contextuality}.

We first prove that antidistinguishing measurement sequences via the measurement of quantum contextual observables between qubit $k$-uniform hypergraph states exist.
\begin{lemma}[Qubit $k$-uniform hypergraph states are antidistinguishable]
    Consider the state $\ket{0}^{\otimes n}$. Let $\ket{\psi_1}$ and $\ket{\psi_2}$ be two qubit $k$-uniform weighted hypergraph states with some hyperedge weight differing by $1$. There exists an antidistinguishing measurement sequence (via measuring quantum contextual observables of the form of $G_\cdot$) of length $n-k+2$ that with certainty discounts one of the three states.
    \label{lemma:hypergraph_state_context}
\end{lemma}
\begin{proof}
    Let $\overline{v}$ be the set of vertices associated with a hyperedge where $\ket{\psi_1}$ and $\ket{\psi_2}$ have weights $e_{\overline{v}}',e_{\overline{v}}''$, respectively, differing by $1$. For simplicity, we consider the basis transformed by powers of $\operatorname{C}^k Z$ operations such that $e_{\overline{v}}'=0$ and $e_{\overline{v}}''=1$; note that this can be done without loss of generality (WLOG) as this basis transformation maps the class of operators comprising the qubit $k$-HRNN to itself, as well as maps $\ket{0}^{\otimes n}$ to itself.

    The antidistinguishing measurement sequence first consists of measuring $Z_i$ for all qubits labeled by $i$ associated with the $n-k$ vertices not in $\overline{v}$; if one of these measurement results is not $1$, then $\ket{0}^{\otimes n}$ is discounted (i.e., the state in question is known not to be $\ket{0}^{\otimes n}$). We assume then that they are all $1$. For all $i\in\overline{v}$, we have by definition that the post-measurement state of $\ket{\psi_1}$ is then just:
    \begin{equation}
        \ket{+}^{\otimes n},
    \end{equation}
    and the post-measurement state of $\ket{\psi_2}$ just the unweighted qubit graph state with stabilizers:
    \begin{equation}
        s_i=X_i\operatorname{C}^kZ_{\overline{v}\setminus\left\{i\right\}}.
    \end{equation}
    Consider now the measurement of the operator:
    \begin{equation}
        M=\cos\left(\cpi\left(R-1\right)\right),
    \end{equation}
    where $R$ is defined as in Eq.~\eqref{eq:R_def}. As $\exp\left(\ci\cpi\left(R-1\right)\right)$ is a product of multi-controlled $Z$ operators, and as multi-controlled $Z$ operators are Hermitian, $M$ can be written as a product of multi-controlled $Z$ operators.

    By the discussion surrounding Eq.~\eqref{eq:R_def}, $R-1$ nullifies $\ket{0}^{\otimes n}$. Thus, the measurement result of $M$ must be $+1$; otherwise, $\ket{0}^{\otimes n}$ would be distinguished. Furthermore, from Eq.~\eqref{eq:qubit_meas_commute}, this measurement commutes with the product of vertex stabilizers for $\ket{\psi_1},\ket{\psi_2}$. Note also that the projection onto the subspace where $M=1$ is equal to the projection onto the subspace where $\exp\left(\ci\cpi\left(R-1\right)\right)=1$. Thus, the post-measurement states $\ket{\phi_1},\ket{\phi_2}$ of $\ket{\psi_1},\ket{\psi_2}$, respectively, then are respectively stabilized by:
    \begin{align}
        t'&=\prod\limits_{i\in\overline{v}}X_i,\\
        t''&=-\prod\limits_{i\in\overline{v}}X_i.
    \end{align}
    $\ket{\phi_1}$ and $\ket{\phi_2}$ are thus orthogonal with distinguishing measurement given by $t'$.
\end{proof}

The antidistinguishing measurement sequence for qumodes is similar.
\begin{lemma}[Qumode $k$-uniform hypergraph states are antidistinguishable]
    Consider the state $\ket{\bm{0}}_{\bm{\hat{q}}}$. Let $\ket{\psi_1}$ and $\ket{\psi_2}$ be two qumode $k$-uniform weighted hypergraph states with some hyperedge weight differing by $1$. There exists an antidistinguishing measurement sequence (via measuring quantum contextual observables of the form of $G_\cdot$) of length $n-k+2$ that with certainty discounts one of the three states.
    \label{lemma:qumode_hypergraph_state_context}
\end{lemma}
\begin{proof}
    Let $\overline{v}$ be the set of vertices associated with a hyperedge where $\ket{\psi_1}$ and $\ket{\psi_2}$ have weights $e_{\overline{v}}',e_{\overline{v}}''$, respectively, differing by $1$. We assume WLOG that $e_{\overline{v}}''<e_{\overline{v}}'$; note that we cannot assume WLOG anymore that one is $0$ and the other $1$, as the class of operators $\exp\left(2\ci\phi\hat{p}_i\right)$ no longer maps to itself under conjugation by the required basis transformation.

    The antidistinguishing measurement sequence first consists of measuring $\hat{q}_i$ for all qubits labeled by $i$ associated with the $n-k$ vertices not in $\overline{v}$; if one of these measurement results is not $0$, then $\ket{\bm{0}}_{\bm{\hat{q}}}$ is discounted (i.e., the state in question is known not to be $\ket{\bm{0}}_{\bm{\hat{q}}}$). We assume then that they are all $0$. For all $i\in\overline{v}$, we have by definition that the post-measurement state of $\ket{\psi_1}$ is then just the $k$-uniform weighted qumode graph state with stabilizers::
    \begin{equation}
        s_i'=X_i\exp\left(\ci\cpi e_{\overline{v}}'\prod\limits_{j\in\overline{v}}\hat{q}_j\right),
    \end{equation}
    and the post-measurement state of $\ket{\psi_2}$ just the unweighted qumode graph state with stabilizers:
    \begin{equation}
        s_i''=X_i\exp\left(\ci\cpi e_{\overline{v}}''\prod\limits_{j\in\overline{v}}\hat{q}_j\right).
    \end{equation}
    Consider now the measurement of the operator:
    \begin{equation}
        M=\exp\left(\ci\cpi\left(R-1\right)\right),
    \end{equation}
    where $R$ is defined as in Eq.~\eqref{eq:qumode_R_def}.

    By the discussion surrounding Eq.~\eqref{eq:R_def}, $\exp\left(\ci\cpi\left(R-1\right)\right)$ stabilizes $\ket{\bm{0}}_{\bm{\hat{q}}}$. Thus, the measurement result of $M$ must be $+1$; otherwise, $\ket{\bm{0}}_{\bm{\hat{q}}}$ would be distinguished. Furthermore, from Eq.~\eqref{eq:qumode_meas_commute}, this measurement commutes with the product of vertex stabilizers for $\ket{\psi_1},\ket{\psi_2}$, up to stabilizers of $\ket{\psi_1},\ket{\psi_2}$. Thus, the post-measurement states $\ket{\phi_1},\ket{\phi_2}$ of $\ket{\psi_1},\ket{\psi_2}$, respectively, then are respectively stabilized by:
    \begin{align}
        t'&=s',\\
        t''&=\exp\left(\ci\cpi\left(e_{\overline{v}}''-e_{\overline{v}}'\right)\right)s'=-s'.
    \end{align}
    $\ket{\phi_1}$ and $\ket{\phi_2}$ are thus orthogonal with distinguishing measurement given by $t'$.
\end{proof}

\subsection{Expressivity Separation for Autoregressive Models}\label{sec:online_sep}

We consider now an autoregressive learner with structure given by Fig.~\ref{fig:classical_models}(a). We assume for simplicity that the learner is deterministic after given a random vector $\bm{r}$. This class of models includes implementations of stochastic simulation algorithms such as Wigner function simulation~\cite{PhysRevLett.109.230503} as well as generative adversarial networks (GANs)~\cite{NIPS2014_5ca3e9b1}. As for each $\bm{r}$, we will demonstrate that there exists an input sequence such that any classical model with $\dim\left(L\right)<\binom{n}{k}-1$ deterministically outputs a measurement sequence inconsistent with quantum mechanics, our results will still hold when considering the model over its randomness. Due to this we will often take the $\bm{r}$-dependence to be implicit.

Autoregressive neural sequence models at time step $i$ map a $d$-dimensional input token $\bm{x_i}$ and a latent vector $\bm{\lambda_{i-1}}$ to an output token $y_i$ and a new latent vector $\bm{\lambda_i}$. After $\ell$ steps, then, we can consider the autoregressive model as a function:
\begin{equation}
    \mathcal{F}^{\bm{r}}:\left(\mathbb{R}^d\right)^\ell\to L\times\mathbb{R}^\ell
    \label{eq:online_map}
\end{equation}
where $\bm{r}$ is a random vector such that, for any $\bm{r}$, $\mathcal{F}^{\bm{r}}$ is deterministic. In the following we will assume that $L$ is a $C^2$ contractible Finsler manifold; concrete examples of such a manifold are the open ball $L=\left(0,1\right)^d$ and the reals $L=\mathbb{R}^D$ for some $D$.

We now state our conditions on $\mathcal{F}^{\bm{r}}$, stated informally in the main text as Assumptions~\ref{ass:c2} and~\ref{ass:strong_sub}. Our proof of Theorem~\ref{thm:online_sep_qubit} concern a subspace $K$ of inputs corresponding to measurements of hypergraph state stabilizers. We require that the projection of $\mathcal{F}^{\bm{r}}$ onto $L$ is $C^2$ on a contractible, $\left(\dim\left(L\right)+2\right)$-dimensional, open subset $U\subseteq K$. As $\dim\left(K\right)=\binom{n}{k}$, our assumption on the dimension of $U$ is satisfied when $U$ is not measure zero in $K$ and $L<\binom{n}{k}-1$, as stated in the main text. We also require that on $U$, $\mathcal{F}^{\bm{r}}$ is a \emph{strong submersion}~\cite{rabier1997ehresmann}---namely, $U$ is such that the Jacobian of $\mathcal{F}^{\bm{r}}$ has a $\dim\left(L\right)\times\dim\left(L\right)$ minor with lower-bounded determinant. This is effectively a requirement that the model is \emph{strongly Morse}~\cite{pmlr-v89-mokhtari19a} in a multivariate sense. Neural networks being strongly Morse (in the univariate case) is a common assumption~\cite{64871590-990e-34a9-8236-2954b5e72da7,pmlr-v89-mokhtari19a,yang2021,dixit2023accelerated}, and certain classes of neural networks are known to be almost surely Morse over settings of their parameters~\cite{kurochkin2021neural}. Finally, we require as a technical assumption that no sequence with limit point in $\partial U\not\subset U$ maps to a sequence with limit point in $L$ under $\mathcal{F}^{\bm{r}}$ (projected onto $L$). By the ``global implicit function theorem'' of Ref.~\cite{rabier1997ehresmann}, $\mathcal{F}^{\bm{r}}$ (projected onto $L$) is then ``globally'' (on $U$) a trivial bundle with base space $C^1$-diffeomorphic to $L$. This, along with Lemmas~\ref{lemma:hypergraph_state_context} and~\ref{lemma:qumode_hypergraph_state_context}, are the central implications which we will use to prove our theorems.

We give three simple examples that immediately imply our strong submersion condition on (an assumed $C^2$) $\mathcal{F}^{\bm{r}}$:
\begin{enumerate}
    \item The presence of an appropriate $\ell_2$ regularization term dominating at infinity.
    \item More generally, assuming that updates to the model during training (via, for instance, gradient descent) are bounded, as our conditions are implied by regularity of the model at infinity. A similar approach was taken in the work of Ref.~\cite{shi2020learning} in justifying certain regularity conditions of machine learning models at infinity.
    \item The model is implemented at a finite precision $\epsilon\ll 1$. A critical value, by definition, has at least one local coordinate that varies only by $\sim\epsilon^2\approx 0$ for an $\epsilon$-perturbation of a critical point. This locally-constant (at precision $\epsilon$) set of local coordinates can be projected out, then, yielding an equivalent model with maximal-rank Jacobian in the neighborhood of this point (with lower-bounded associated minor determinant), with effective latent space $\tilde{L}$ of lower dimension than $L$.
\end{enumerate}

We give two simple examples that imply our condition on certain limit points not existing in $L$:
\begin{enumerate}
    \item $U$ does not contain a boundary at all, e.g., a sufficiently strong regularity condition is imposed on $\mathcal{F}^{\bm{r}}$ such that it is a strong submersion on an unbounded subspace. This is similar to the first two examples previously mentioned.
    \item $\mathcal{F}^{\bm{r}}$ is a proper open map (onto its image) on $U\cup\partial U$, such as if it is composed of linear transformations and leaky rectified linear units. This is as under such a map, $\partial U$ maps to $\partial\left(\mathcal{F}^{\bm{r}}\left(U\right)\right)$, which has no intersection with $\mathcal{F}^{\bm{r}}\left(U\right)$ as $\mathcal{F}^{\bm{r}}$ is open. To see this, assume some point $d\in\partial U$ mapped to $y\in\mathcal{F}^{\bm{r}}\left(U\right)$. By the definition of this set, there must also be some $u\in U$ mapping to $y$ under $\mathcal{F}^{\bm{r}}\left(U\right)$. As $\mathcal{F}^{\bm{r}}\left(U\right)$ is proper, there must be a neighborhood of $y$ with preimage (intersected with $U$) compactly contained in $U$. However, the image of any Cauchy sequence $\left(x_i\right)$ in $U$ approaching $d\not\in U$ must be in this neighborhood of $y$ for sufficiently large $i$ by continuity, violating compactness and yielding a contradiction.
\end{enumerate}

Finally, we define a helper bump function that we use in our proof. Consider the smooth transition function:
\begin{equation}
    g\left(x\right):=\left(\frac{\exp\left(-x^{-1}\right)}{\exp\left(-x^{-1}\right)+\exp\left(-\left(1-x\right)^{-1}\right)}\right).
\end{equation}
The function
\begin{equation}
    h_{a,b,c,d}\left(x\right):=g\left(\frac{x-a}{b-a}\right)g\left(\frac{d-x}{d-c}\right)
\end{equation}
is smooth, equal to $1$ on the interval $\left[b,c\right]$, and vanishes outside of the interval $\left(a,d\right)$. The periodic bump function (the ``bubble wrap'' function):
\begin{equation}
    t\left(x\right):=\sum\limits_{i=-\infty}^\infty h_{2\cpi i,\frac{\cpi}{2}+2\cpi i,\cpi+2\cpi i,\frac{3\cpi}{2}+2\cpi i}\left(x\right)\label{eq:bubblewrap_function}
\end{equation}
is then a smooth function equal to $1$ when $x\equiv z\pmod{2\cpi}$ for all $\frac{\cpi}{2}\leq z\leq\cpi$, and equal to $0$ when $x\equiv z\pmod{2\cpi}$ for all $\frac{3\cpi}{2}\leq z\leq 2\cpi$.

With the preliminaries in place, we now prove our expressivity separation on the $\left(\ell,n,k\right)$-HSMT task.
\begin{theorem}[Autoregressive hypergraph stabilizer measurement translation memory lower bound]
    Consider an autoregressive model with $C^2$ contractible Finsler latent manifold $L$, and model function after $m=\binom{n}{k}+n$ tokens:\footnote{Recall from Appendix~\ref{sec:task_desc} that the first $\binom{n}{k}+n$ tokens are uniquely described by $\left(\bm{\upsilon},\bm{\gamma}\right)\in\mathbb{R}^m$.}
    \begin{equation}
        \mathcal{F}:\mathbb{R}^m\to L\times\mathbb{R}^n.
    \end{equation}
    Consider as well $K\subset\mathbb{R}^m$, the space of $\bm{Q}=\left(\bm{\upsilon},\bm{\gamma}\right)$ where all $\upsilon_i=t\left(\left\lVert\bm{\gamma}\right\rVert_2^2\right)$, with $t$ defined as in Eq.~\eqref{eq:bubblewrap_function}.

    Assume that the projection of $\mathcal{F}\left(U\right)$ onto $L$ is $C^2$ for some contractible, open subset $U\subseteq K$, where $U$ is a $C^1$ manifold of dimension at least $\dim\left(L\right)+2$. Assume also that $\mathcal{F}$ is a strong submersion on $U$, and that no sequence with limit point in $\partial U\not\subset U$ maps to a sequence with limit point in $L$ under $\mathcal{F}$ projected onto $L$.

    This model cannot achieve a finite backward empirical cross entropy on the $\left(\binom{n}{k}+3n-k+2,n,k\right)$-hypergraph stabilizer measurement translation task due to quantum contextuality.
    \label{thm:online_sep_qubit}
\end{theorem}
\begin{proof}
    Due to the bijection between $\bm{Q}\in K$ and the associated hypergraph adjacency tensor $\bm{A}$ (i.e., $\bm{\gamma}$ is the vectorized $\bm{A}$), we will often use $\bm{A}$ to refer to its associated $\bm{Q}$ by proxy. It is obvious from this construction that $K$ is an $\binom{n}{k}$-dimensional embedding of representations of $k$-uniform hypergraphs via their adjacency tensors. Thus, the states described by points in $K$ are exactly $k$-uniform weighted hypergraph states with, depending on the measurement results, perhaps overall phases on the stabilizers.

    We first show that $\mathcal{F}$ induces a fiber bundle. We use $\varPi_L\circ\mathcal{F}$ to denote the projection of $\mathcal{F}$ onto $L$, and $\left.\varPi_L\circ\mathcal{F}\right|_U$ its restriction to $U$. By Theorem 5.2 of Ref.~\cite{rabier1997ehresmann}, following our assumptions, $\left.\varPi_L\circ\mathcal{F}\right|_U$ is a fiber bundle with base space $C^1$-diffeomorphic to $L$ and induces a $C^1$-diffeomorphism $\eta:H\times L\to U$ such that:
    \begin{equation}
        H\times L\cong U
    \end{equation}
    for some fiber space $H$. In particular, $\left.\varPi_L\circ\mathcal{F}\right|_U$ is a trivial bundle, and $H$ is thus contractible and of dimension:
    \begin{equation}
        \dim\left(H\right)=\dim\left(U\right)-\dim\left(L\right)\geq 2.
    \end{equation}
    We claim that $\eta\left(H\times\left\{\lambda\right\}\right)$ is unbounded for all $\lambda\in L$. To see this, first note that as $\eta$ is a homeomorphism, the boundary of $\eta\left(H\times\left\{\lambda\right\}\right)$ for all $\lambda\in L$ must be a subset of $\partial U$. Assume now that it is nonempty, i.e., that there existed a Cauchy sequence with limit point in $\partial H\times\left\{\lambda\right\}$ that mapped to a Cauchy sequence with limit point in $\partial U$ under $\eta$. Then, a sequence with limit point in $\partial U$ would have base point $\lambda\in L$, violating our assumption that no sequence with limit point in the boundary of $U$ has limit point in $L$ under $\varPi_L\circ\mathcal{F}$. Thus, $\eta\left(H\times\left\{\lambda\right\}\right)$ has no boundary. In particular, as $\eta\left(H\times\left\{\lambda\right\}\right)$ is contractible, it is not compact as closed manifolds cannot be contractible. As $\eta\left(H\times\left\{\lambda\right\}\right)\subset K$ has no boundary it is therefore unbounded.

    We now show that each fiber contains a sequence describing trivial operations in the $\left(\ell,n,k\right)$-HSMT task. To see this, fix some $\lambda\in L$. As $\eta\left(H\times\left\{\lambda\right\}\right)$ is unbounded and path-connected, there exist points with arbitrary $2$-norm modulo $2\cpi$ in $\eta\left(H\times\left\{\lambda\right\}\right)$; namely, there exists a point $\left(\bm{h_0},\lambda\right)\in H\times L$ such that (mapped under $\eta$):
    \begin{equation}
        \left\lVert\bm{\gamma}\right\rVert_2^2\equiv 0\pmod{2\cpi}
    \end{equation}
    and thus all
    \begin{equation}
        \upsilon_i=t\left(\left\lVert\bm{\gamma}\right\rVert_2^2\right)=0.
    \end{equation}
    This corresponds to trivial measurements in the translation task, fixing the measurement outcomes to all be $+1$. Note that $\lambda$ is the image of $\bm{h_0}$ under $\varPi_L\circ\mathcal{F}$. Assuming the model performs the translation task correctly, this fixes the measurement outcomes of the first $n$ measurements to all be $+1$ for \emph{all} points mapping to $\lambda$ under $\varPi_L\circ\mathcal{F}$, as the next $n$ outputs are just recapitulations of these first $n$ measurement outcomes and they must be consistent for all inputs mapping to $\lambda$.

    We now show the existence of other, nontrivial sequences sharing a fiber with $\left(\bm{h_0},\lambda\right)$. Recall that $H$ is unbounded and path-connected, and thus there exists some choice of $\overline{v}\in\mathbb{R}^{\binom{\left[n\right]}{k}}$ and $\gamma^\ast\in\mathbb{R}$ such that, for any $\gamma\in\left[\gamma^\ast,\infty\right)$,\footnote{Or $\gamma\in\left(-\infty,-\gamma^\ast\right]$, though the logic proceeds identically.} there exists some point $\bm{Q}\in K$ with $\gamma_{\overline{v}}=\gamma$ that maps to $\lambda\in L$ under $\varPi_L\circ\mathcal{F}$. Furthermore, as $\dim\left(H\right)\geq 2$, we can assume WLOG that the associated $\bm{\upsilon}\not\equiv 0\pmod{2\cpi}$ by varying the other coordinate of $\eta\left(H\times\left\{\lambda\right\}\right)$. Thus, there exist $\bm{h'},\bm{h''}\in U$ mapping to $\lambda$ that describe $k$-uniform weighted hypergraph states with some respective hyperedges $e_{\overline{v}}',e_{\overline{v}}''$ differing by $1$.

    By Lemma~\ref{lemma:hypergraph_state_context} (in the qubit setting) or Lemma~\ref{lemma:qumode_hypergraph_state_context} (in the qumode setting), we have that there exists a measurement sequence (describable by the translation task) of length $n-k+2$ that antidistinguishes these three associated states with certainty. By then performing this antidistinguishing measurement sequence, the output for one of these $\bm{h_0},\bm{h'},\bm{h''}$ must then be incorrect. Thus, the model must obtain an infinite backward empirical cross entropy on these three sequences when followed by the antidistinguishing measurement sequence.
\end{proof}

\subsection{Expressivity Separation for Encoder-Decoder Models}\label{sec:enc_dec_sep}

Though autoregressive sequence models are perhaps conceptually the simplest as they directly map input tokens to output tokens, in practice encoder-decoder models outperform them~\cite{10.5555/2969033.2969173,10.5555/3295222.3295349}. We now show that no encoder-decoder model (with similar technical assumptions as the autoregressive separation) can perform the introduced tasks to finite backward empirical cross entropy. The proof will be similar to that of Theorem~\ref{thm:online_sep_qubit}; however, as encoder-decoder models can see the entire input sequence at once, we do not directly have the freedom to choose the antidistinguishing measurement sequence as in the proofs of those theorems. This will change the details of the proofs here.

We consider an encoder-decoder model with structure given by Fig.~\ref{fig:classical_models}(b). The encoder of such a model is a function:
\begin{equation}
    \mathcal{E}^{\bm{r}}:\left(\mathbb{R}^d\right)^\ell\to L
    \label{eq:model_map}
\end{equation}
where, as in the proofs of Theorem~\ref{thm:online_sep_qubit}, $\bm{r}$ is a random vector such that, for any $\bm{r}$, $\mathcal{E}^{\bm{r}}$ is deterministic. For the reasons discussed in Appendix~\ref{sec:online_sep} we will take the $\bm{r}$-dependence to be implicit WLOG. We once again assume that $L$ is a $C^2$ contractible Finsler manifold.

As $\mathcal{E}$ is a function of the \emph{entire} input sequence, we have to appropriately modify our assumptions on its derivatives. We once again assume the existence of a contractible, $\left(\dim\left(L\right)+2\right)$-dimensional, open subset $U\subseteq K$, where $K$ once again is a subspace of inputs corresponding to the measurement of hypergraph state stabilizers. To match the full sequence length, though, we now have to assume that for $\mathcal{E}$ on $U\times\left\{\bm{v}\right\}$---not just $U$, as in the autoregressive model separation---for all but isolated $\bm{v}\in\mathbb{R}^{\ell-\dim\left(K\right)}$, $\mathcal{E}$ is a strong submersion, and no sequence with limit point in $\partial U\times\left\{\bm{v}\right\}\not\subset U\times\left\{\bm{v}\right\}$ maps to a sequence with limit point in $L$ under $\mathcal{E}$.

With the preliminaries in place, we now prove our expressivity separation for encoder-decoder models.
\begin{theorem}[Encoder-decoder hypergraph stabilizer measurement translation memory lower bound]
    Consider an encoder-decoder model with $C^2$ contractible Finsler latent manifold $L$, and encoder function
    \begin{equation}
        \mathcal{E}:\mathbb{R}^{m+n+2^k+1}\to L,
    \end{equation}
    where $m=\binom{n}{k}+n$. Let $K\subset\mathbb{R}^m$ be as in Theorem~\ref{thm:online_sep_qubit}. Consider any $V\subseteq\mathbb{R}^{n+2^k+1}$ that equals $\mathbb{R}^{n+2^k+1}$ up to the subtraction of isolated points. Assume there exists a contractible, open subset $U\subset K$, where $U$ is a $C^1$ manifold of dimension at least $\dim\left(L\right)+2$ such that for all $\bm{v}\in V$, on $U\times\left\{\bm{v}\right\}$, $\mathcal{E}$ is $C^2$ and a strong submersion. Finally, assume no sequence with limit point in $\partial U\times\left\{\bm{v}\right\}\not\subset U\times\left\{\bm{v}\right\}$ maps to a sequence with limit point in $L$ under $\mathcal{E}$.

    This model cannot achieve a finite backward empirical cross entropy on the $\left(\binom{n}{k}+2n-k+3,n,k\right)$-hypergraph stabilizer measurement translation (without repetition) task due to quantum contextuality.\label{thm:enc_dec_sep_qubit}
\end{theorem}
\begin{proof}
    Consider $R:=K\times\mathbb{R}^{n+2^k+1}$, with points of the form $\left(\bm{Q},\bm{P}\right)$ with $\bm{Q}\in K$ and $\bm{P}\in\mathbb{R}^{n+2^k+1}$. Just as in the proof of Theorem~\ref{thm:online_sep_qubit}, it is obvious from this construction that $K$ is an $\binom{n}{k}$-dimensional embedding of representations of $k$-uniform hypergraphs via their adjacency tensors. Thus, the states described by points in $K$ are exactly $k$-uniform weighted hypergraph states with, depending on the measurement results, overall phases on the stabilizers.

    We will proceed as follows. First, we will show that when $\dim\left(L\right)$ is sufficiently small, $\mathcal{E}$ must map three distinct $\left(\bm{Q},\bm{P}\right)$, $\left(\bm{Q'},\bm{P}\right)$, and $\left(\bm{Q''},\bm{P}\right)$ to the same point in latent space $\lambda\in L$ for some fixed $\bm{P}=\bm{v_0}\in V$. Then, we will use Lemma~\ref{lemma:hypergraph_state_context} (in the qubit setting) or Lemma~\ref{lemma:qumode_hypergraph_state_context} (in the qumode setting) to show that the stabilizers of these states give rise to an antidistinguishing measurement sequence of length $n-k+2$ once the phases on their stabilizers have been fixed to be identical. Finally, we will show that one can find $\left(\bm{Q},\bm{P'}\right)$, $\left(\bm{Q'},\bm{P'}\right)$, and $\left(\bm{Q''},\bm{P'}\right)$ also mapping to $L$, where $\bm{P'}$ is the associated antidistinguishing measurement sequence. This forces an incorrect output on one of these inputs, implying an infinite backward empirical cross entropy on any finite set containing these three measurement sequences.

    We first show that $\mathcal{E}$ induces a fiber bundle. Let $\left.\mathcal{E}\right|_{U\times\left\{\bm{v}\right\}}$ denote $\mathcal{E}$ restricted to $U\times\left\{\bm{v}\right\}\subset R$ for any choice of $\bm{v}\in V$. By Theorem 5.2 of Ref.~\cite{rabier1997ehresmann}, following our assumptions, $\left.\mathcal{E}\right|_{U\times\left\{\bm{v}\right\}}$ is a fiber bundle for all $\bm{v}\in V$ with base space $C^1$-diffeomorphic to $L$ and induces a $C^1$-diffeomorphism $\eta_{\bm{v}}:H_{\bm{v}}\times L\to U\times\left\{\bm{v}\right\}$ such that:
    \begin{equation}
        H_{\bm{v}}\times L\cong U\times\left\{\bm{v}\right\}
    \end{equation}
    for some fiber space $H_{\bm{v}}$. In particular, $\left.\mathcal{E}\right|_{U\times\left\{\bm{v}\right\}}$ is a trivial bundle, and $H_{\bm{v}}$ is thus path-connected and of dimension:
    \begin{equation}
        \dim\left(H_{\bm{v}}\right)=\dim\left(U\right)-\dim\left(L\right)\geq 2.
    \end{equation}
    By the transitivity of $C^1$-diffeomorphisms, all $H_{\bm{v}}$ are $C^1$-diffeomorphic and can be taken to be some $H$ WLOG. We thus have the trivial bundle structure:
    \begin{equation}
        H\times L\times V\cong U\times V,\label{eq:fiber_def_qubit_enc_dec}
    \end{equation}
    with fibers $C^1$-diffeomorphic to $H\times V$ and base space $C^1$-diffeomorphic to $L$. We claim that $\eta_{\bm{v_0}}\left(H\times\left\{\lambda\right\}\times\left\{\bm{v_0}\right\}\right)$ is unbounded for all $\lambda\in L,\bm{v_0}\in V$. To see this, first note that as $\eta_{\bm{v_0}}$ is a homeomorphism, the boundary of $\eta_{\bm{v_0}}\left(H\times\left\{\lambda\right\}\times\left\{\bm{v_0}\right\}\right)$ for all $\lambda\in L$ must be a subset of $\partial U\times\left\{\bm{v_0}\right\}$ for all $\bm{v_0}\in V$. Assume now that it is nonempty, i.e., that there existed a Cauchy sequence with limit point in $\partial H\times\left\{\lambda\right\}\times\left\{\bm{v_0}\right\}$ that mapped to a Cauchy sequence with limit point in $\partial U\times\left\{\bm{v_0}\right\}$ under $\eta_{\bm{v_0}}$. Then, a sequence with limit point in $\partial U\times\left\{\bm{v_0}\right\}$ would have base point $\lambda\in L$, violating our assumption that no sequence with limit point in the boundary of $U\times\left\{\bm{v_0}\right\}$ has limit point in $L$ under $\mathcal{E}$. Thus, $\eta_{\bm{v_0}}\left(H\times\left\{\lambda\right\}\times\left\{\bm{v_0}\right\}\right)$ has no boundary. In particular, as $\eta_{\bm{v_0}}\left(H\times\left\{\lambda\right\}\times\left\{\bm{v_0}\right\}\right)$ is contractible, it is not compact as closed manifolds cannot be contractible. As $\eta_{\bm{v_0}}\left(H\times\left\{\lambda\right\}\times\left\{\bm{v_0}\right\}\right)\subset K\times\left\{\bm{v_0}\right\}$ has no boundary it is therefore unbounded.

    We now show that each fiber contains a sequence beginning with $m$ trivial operations in the $\left(\ell,n,k\right)$-HSMT task. To see this, fix $\lambda\in L,\bm{v_0}\in V$. As $\eta_{\bm{v_0}}\left(H\times\left\{\lambda\right\}\times\left\{\bm{v_0}\right\}\right)$ is unbounded and path-connected, there exist points with arbitrary $2$-norm modulo $2\cpi$ in $\eta_{\bm{v_0}}\left(H\times\left\{\lambda\right\}\times\left\{\bm{v_0}\right\}\right)$; namely, there exists a point $\left(\bm{h_0},\lambda,\bm{v_0}\right)\in H\times L\times\left\{\bm{v_0}\right\}$ such that (mapped under $\eta_{\bm{v_0}}$):
    \begin{equation}
        \left\lVert\bm{\gamma}\right\rVert_2^2\equiv 0\pmod{2\cpi}
    \end{equation}
    and thus all
    \begin{equation}
        \upsilon_i=t\left(\left\lVert\bm{\gamma}\right\rVert_2^2\right)=0.
    \end{equation}
    This corresponds to trivial measurements in the translation task, fixing the measurement outcomes to all be $+1$. Note that $\lambda$ is the image of $\left(\bm{h_0},\bm{v_0}\right)$ under $\mathcal{E}$. Assuming the model performs the translation task correctly, this fixes the measurement outcomes of the first $n$ measurements to all be $+1$ for \emph{all} points mapping to $\lambda$ under $\mathcal{E}$ as the decoder $\mathcal{D}$ only depends on $\lambda$.
    
    We now show the existence of other, nontrivial sequences sharing a fiber with $\left(\bm{h_0},\lambda,\bm{v_0}\right)$. Recall that $H\times\left\{\bm{v_0}\right\}$ is unbounded and path-connected, and thus there exists some choice of $\overline{v}\in\mathbb{R}^{\binom{\left[n\right]}{k}}$ and $\gamma^\ast\in\mathbb{R}$ such that, for any $\gamma\in\left[\gamma^\ast,\infty\right)$,\footnote{Or $\gamma\in\left(-\infty,-\gamma^\ast\right]$, though the logic proceeds identically.} there exists some point $\left(\bm{Q},\bm{v_0}\right)\in K\times\left\{\bm{v_0}\right\}$ with $\gamma_{\overline{v}}=\gamma$ that maps to $\lambda\in L$ under $\mathcal{E}$. Furthermore, as $\dim\left(H\right)\geq 2$, we can assume WLOG that the associated $\bm{\upsilon}\not\equiv 0\pmod{2\cpi}$ by varying the other coordinate of $\eta_{\bm{v_0}}\left(H\times\left\{\lambda\right\}\times\left\{\bm{v_0}\right\}\right)$. hus, there exist $\left(\bm{h'},\bm{v_0}\right),\left(\bm{h''},\bm{v_0}\right)\in U\times\left\{\bm{v_0}\right\}$ mapping to $\lambda$ that describe $k$-uniform weighted hypergraph states with some respective hyperedges $e_{\overline{v}}',e_{\overline{v}}''$ differing by $1$.

    By Lemma~\ref{lemma:hypergraph_state_context} (in the qubit setting) or Lemma~\ref{lemma:qumode_hypergraph_state_context} (in the qumode setting), we have that there exists a measurement sequence (describable by the translation task) of length $n-k+2$ that antidistinguishes these three associated states with certainty; by considering the dimensions of $\bm{b}$, $\bm{\phi}$, and $\bm{\theta}$ as described in Appendix~\ref{sec:task_desc}, this defines the first $n+2^k$ coordinates of $\bm{P}$. As $V$ is $\mathbb{R}^{n+2^k+1}$ up to potentially the removal of isolated points, by varying the final coordinate---corresponding to measurements \emph{after} the first $n-k+2$ measurements---any given antidistinguishing measurement of length $n-k+2$ can be assumed to be in $V$. As the construction given in Eq.~\eqref{eq:fiber_def_qubit_enc_dec} is unchanged when taking $\eta_{\bm{v}}$ rather than $\eta_{\bm{v_0}}$ as described, we can take $\bm{v}$ to be this antidistinguishing measurement sequence while still ensuring the associated three states map to the same point in the latent space of the model. Therefore, the output for one of these $\bm{h_0},\bm{h'},\bm{h''}$ must then be incorrect. Thus, the model must obtain an infinite backward empirical cross entropy on these three sequences when followed by the antidistinguishing measurement sequence.
\end{proof}

\section{CV \texorpdfstring{$k$}{k}-HRNNs as a Quantization of Recurrent Neural Networks}\label{sec:hrnns_as_q_rnns}

We assume the dimension of the output tokens $y_i$ is $m=1$ for notational simplicity in what follows; the generalization to larger $m$ is immediate. The construction of CV $k$-HRNNs described in Sec.~\ref{sec:hrnn} is the quantization of alternating the classical Hamiltonian dynamics under the Hamiltonians:
\begin{align}
    G\left(\bm{x_i}\right)&:=\gamma\left(\bm{x_i}\right)\prod_{w_j\in\bm{\alpha}\left(\bm{x_i}\right)}q_{w_j},\\
    H\left(\bm{x_i}\right)&:=\left(\sum_{j=1}^k\phi_j p_{v_j}+\sum\limits_{\overline{w}\in 2^{\overline{v}}}\theta_{\overline{w}}\prod_{w_j\in\overline{w}}q_{w_j}\right)p_0,
\end{align}
where in the definition of $H$:
\begin{align}
    \bm{\beta}\left(\bm{x_i}\right)&=\overline{v},\\
    \bm{\kappa}\left(\bm{x_i}\right)&=\left(\bm{\phi},\bm{\theta}\right)
\end{align}
and $p_0$ is the momentum of the ancilliary output register used to perform phase estimation in Fig.~\ref{fig:hrnn}.

Consider now discrete time dynamics under this Hamiltonian evolution with fresh ancillary modes labeled by $0$ at each time step. We assume that initially all canonical variables are $0$, other than an initial moment $p_0=1$ for the ancilliary mode added at each time step. After evolution by $G$ followed by $H$, we have from Hamilton's equations of motion the new canonical variables (given by primed notation):
\begin{align}
    p_a'&=p_a-\delta_{a\in\bm{\alpha}\left(\bm{x_i}\right)}\gamma\left(\bm{x_i}\right)\prod_{w_j\neq a\in\bm{\alpha}\left(\bm{x_i}\right)}q_{w_j}-\sum_{\overline{w}\in 2^{\overline{v}}:a\in\overline{w}}\theta_{\overline{w}}\prod_{w_j\neq a\in\overline{w}}q_{w_j},\\
    q_a'&=q_a+\phi_a,\\
    q_0'&=\sum_{a=1}^k\phi_a\left(p_a-\delta_{a\in\bm{\alpha}\left(\bm{x_i}\right)}\gamma\left(\bm{x_i}\right)\prod_{w_j\neq a\in\bm{\alpha}\left(\bm{x_i}\right)}q_{w_j}\right)+\sum\limits_{\overline{w}\in 2^{\overline{v}}}\theta_{\overline{w}}\prod_{w_j\in\overline{w}}q_{w_j},
\end{align}
where $\delta_\cdot$ denotes the indicator function.

In particular, a classical recurrent neural network~\cite{Hopfield2554} obeying these alternating Hamiltonian dynamics updates its latent vector as:
\begin{equation}
    \bm{\lambda'}=\bm{f_{k-1}}\left(\bm{\lambda};\bm{x}\right),
\end{equation}
where $\bm{f_{k-1}}\left(\bm{\lambda};\bm{x}\right)$ is some degree-$\left(k-1\right)$ polynomial in the entries of $\bm{\lambda}$. Similarly, it outputs:
\begin{equation}
    \bm{y}=\bm{g_{k-1}}\left(\bm{\lambda};\bm{x}\right),
\end{equation}
where $\bm{g_{k-1}}$ is some degree-$\left(k-1\right)$ polynomial in the entries of $\bm{\lambda}$.

\section{Time Complexity}\label{sec:deets_on_complexity}

We here discuss the implications of the existence of a polynomial-overhead classical simulation algorithm of the qumode $k$-HRNN as discussed in Sec.~\ref{sec:qumode_hrnn_inf}; we assume $k$ is constant with respect to $n$ in what follows.

As mentioned in the main text, our achieved memory separation is asymptotically optimal by Ref.~\cite{PhysRevLett.97.190501}, which gives an $\operatorname{O}\left(n^k\right)$-memory classical simulation algorithm for the quantum model we consider. This is accomplished via efficient representations of systems with dynamics constrained to belong to a low-dimensional Lie group; this is exactly our setting, where the dimension of the Lie group is $\operatorname{\Theta}\left(n^k\right)$. As this algorithm requires multiplications of $\operatorname{\Theta}\left(n^k\right)$-dimensional representations of members of this algebra, it achieves a time complexity of $\operatorname{O}\left(n^{\upomega k}\right)$ per unit cell of the $k$-HRNN, where $\upomega\geq 2$ is the matrix multiplication exponent (currently known to be at most $2.37187$~\cite{duan2022faster}, but for reasonable $n$ no better than $\log_2\left(7\right)\approx 2.8$ using Strassen's algorithm~\cite{strassen1969gaussian}). For $k=2$ it is further known that $\operatorname{O}\left(n^2\right)$-time classical simulation algorithms exist for a unit cell of a $k$-HRNN~\cite{PhysRevA.70.052328,calcuth2022}.

As discussed in Sec.~\ref{sec:qumode_hrnn_inf}, we expect training times for $k$-HRNNs to be at least $\operatorname{\Omega}\left(n^k\right)$ per training iteration; this is required to even resolve the gradients of the model~\cite{larocca2022diagnosingbarren,rudolph2023trainability,fontana2023adjoint,ragone2023unified}. Though the classical algorithm of Ref.~\cite{PhysRevLett.97.190501} does not quite reach this scaling, this is only a constant-degree polynomial overhead, not an arbitrary-degree polynomial separation.

Interestingly, this gives a nice way to train the quantum model: one can classically train a $k$-HRNN via the classical simulation algorithm of Ref.~\cite{PhysRevLett.97.190501} with only a cubic overhead, but still leverage an arbitrary-polynomial resource advantage by using a quantum computer at time of inference. In this way, one can consider the $k$-HRNN as a ``quantum compression'' of a quantum-inspired classical neural network model of size $N$ trained on a classical computer. Once training is finished, one implements the model on a quantum computer using $\sim N^{\frac{1}{k}}$ qumodes. All future evaluations of the model now use polynomially fewer resources than the classical implementation, with the degree of separation depending on $k$. This also opens the possibility for the quantum model to ``take over'' training when in the vicinity of a minimum, where the loss landscape is believed to be more amenable for training on a quantum device~\cite{anschuetz2021critical,anschuetz2022barren}. We leave further investigation of the implications of hybrid quantum-classical training algorithms for future work.

\section{Trainability in the Qubit-Based Construction}\label{sec:qubit_based_sep}

We here discuss a generalization of the qubit-based $k$-HRNNs considered in Sec.~\ref{sec:hrnn} which allows for efficient training of the model when $k$ is constant. We here assume that the model is initialized with latent state $\ket{\lambda_1}=\ket{+}^{\otimes n}$.\footnote{Other initial latent states can be considered by fixing some number of initial tokens inputted to the $k$-HRNN.} In each unit cell, a gate of the form:
\begin{equation}
\operatorname{C}^{k-1}Z_{\overline{v}}\left(\phi\right)=\exp\left(-\ci\phi\ket{\bm{1}}\bra{\bm{1}}_{\overline{v}}\right)
\end{equation}
is applied, and is followed by phase estimation of an operator of the form:
\begin{equation}
    G_{\overline{v}}\left(\bm{\phi},\bm{\theta}\right):=\exp\left(-\ci\sum_{i=1}^k\phi_i X_{v_i}\right)\exp\left(-\ci\sum_{\overline{w}\in 2^{\overline{v}}}^k\theta_{\overline{w}}\ket{\bm{1}}\bra{\bm{1}}_{\overline{w}}\right).
\end{equation}
Here, we have used the shorthand:
\begin{equation}
    \ket{\bm{1}}\bra{\bm{1}}_{\overline{v}}:=\prod_{v_i\in\overline{v}}\left(\frac{I_{2^n}-Z_{v_i}}{2}\right).
\end{equation}
We now present a modification of this construction which, though it is no longer translationally invariant, will allow for an efficient training algorithm while maintaining our proven classical lower bound for simulating the network.

We consider two predetermined time steps---that is, token indices---$\tau_2\geq\tau_1$. We constrain our model such that:
\begin{enumerate}
    \item For each time step indexed by $t\in\left[\tau_1\right]$, all $\bm{\kappa}\left(\bm{x_t}\right)$ are constrained to be $\bm{0}$.
    \item For each time step indexed by $t\in\left\{\tau_1+1,\ldots,\tau_2\right\}$, all $\gamma\left(\bm{x_t}\right)$ are constrained to be $0$ and all $G_{\overline{v}}\left(\bm{\kappa}\left(\bm{x_t}\right)\right)$ of the form:
    \begin{equation}
        G_j\left(\bm{\kappa}\left(\bm{x_t}\right)\right)=\exp\left(-\ci\cpi\ket{1}\bra{1}_j\right)
    \end{equation}
    for $j$ labeling a single qubit. We assume exactly $n-k$ distinct $j$ are measured for time steps $t\in\left\{\tau_1+1,\ldots,\tau_2\right]$.
    \item For each time step indexed by $t>\tau_2$, all $\gamma\left(\bm{x_t}\right)$ are constrained to be $0$.
\end{enumerate}

Though this construction seems arbitrary, the qubit $k$-HRNN is constructed in such a way such that:
\begin{enumerate}
    \item Time steps through $t\in\left[\tau_1\right]$ prepare a qubit hypergraph state.
    \item Time steps $t\in\left\{\tau_1+1,\ldots,\tau_2\right\}$ measure $n-k$ qubits in the computational basis.
    \item Time steps $t>\tau_2$ measures operators of the form of $G_\cdot$ on the remaining qubits.
\end{enumerate}
As quantum hypergraph states have a classically efficient stabilizer representation (see Appendix~\ref{sec:qubit_hypergraph}), for $k$ constant, all steps of this procedure are efficient to classically simulate. Just as is the case in the qumode setting, then---see Sec.~\ref{sec:qumode_hrnn_inf}---there are no quantum trainability barriers (QTBs) as the model can be simulated classically. Finally, our memory lower bound for classical simulation of the model still holds, as the subset of inputs $\bm{x_i}=\left(\bm{\alpha_i},\bm{\beta_i},\gamma_i,\bm{\kappa_i}\right)$ we consider in the $\left(\ell,n,k\right)$-HSMT task in our proofs of Theorems~\ref{thm:online_sep_qubit} and~\ref{thm:enc_dec_sep_qubit} obey these same constraints.

\end{document}